\documentclass[11pt,fullpage,letterpaper]{article} 

\usepackage{natbib}

\setcitestyle{authoryear,round,citesep={;},aysep={,},yysep={;}}
\usepackage[english]{babel}

\usepackage[margin=1in]{geometry}
\usepackage{fancyhdr}
\usepackage[utf8]{inputenc} 
\usepackage[T1]{fontenc}    

\usepackage{titletoc}
\usepackage[toc, page, header]{appendix} 

\usepackage{url}            
\usepackage{booktabs}       
\usepackage{amsfonts}       
\usepackage{nicefrac}       
\usepackage{microtype}      
\usepackage{xcolor}         

\usepackage{parskip}

\title{Efficiency, Feasibility, and Incentive-Awareness\\ in Constrained Online Resource Allocation}
\author{%
    Yan Dai\thanks{Operations Research Center, MIT. Email: \texttt{yandai20@mit.edu}.}
    \and
    Negin Golrezaei\thanks{Sloan School of Management, MIT. Email: \texttt{golrezae@mit.edu}.}
    \and
    Patrick Jaillet\thanks{Department of EECS, MIT. Email: \texttt{jaillet@mit.edu}.}
}
\date{First version: May 2025. This version: July 2026.\footnote{A preliminary version of this work was accepted to the 39th Annual Conference on Neural Information Processing Systems (NeurIPS 2025).
Compared to that version, this version has: expanded the lireature review and positioning of our paper (\Cref{sec:introduction}); developed our methodology into a unified Incentive-Aware Primal-Dual (\algname) framework (\Cref{sec:mechanism}); extensively reorganized and refined the technical exposition (\Cref{sec:mechanism,sec:mechanism dual}), especially using \Cref{sec:vicious cycle} and \Cref{fig:primal-dual,fig:vicious cycle,fig:lazy vs no lazy}; extended our proposed mechanisms to multi-unit, multi-demand resource allocation (\Cref{sec:extensions}); and substantially expanded our numerical simulations, featuring neural Q-learning agents, a broader range of resource allocation agents, and an ablation on the \algname components (\Cref{sec:numerical illustration setup}).}}

\usepackage{multirow}
\usepackage{arydshln}
\usepackage{footnote}
\usepackage{enumitem}
\setitemize{leftmargin=*, nosep, itemsep=2pt, parsep=3pt}
\setenumerate{leftmargin=*, nosep, itemsep=2pt, parsep=3pt}

\usepackage{amsmath}
\usepackage{amsthm}
\usepackage{amssymb}
\usepackage{mathtools}
\usepackage{thmtools}
\usepackage{comment}
\usepackage{bbm}
\usepackage{bm}
\usepackage{newtxtext,newtxmath}
\allowdisplaybreaks
\newcommand\scalemath[2]{\scalebox{#1}{\mbox{\ensuremath{\displaystyle #2}}}}
\let\originalmiddle=\middle
\def\middle#1{\mathrel{}\originalmiddle#1\mathrel{}}

\usepackage{algorithm}
\usepackage[noEnd,commentColor=blue!70!white]{algpseudocodex}

\makeatletter
\newcommand{\algeq}[2][]{%
  \par%
  \noindent%
  \makebox[\columnwidth]{%
    \hfill #2%
    \if\relax\detokenize{#1}\relax
      \hfill
    \else
      \hfill (\refstepcounter{equation}\theequation)\label{#1}%
    \fi
  }%
  \par%
}
\makeatother

\usepackage[colorlinks=true, linkcolor=blue!60!black, citecolor=blue!60!black,
urlcolor=black]{hyperref}
\usepackage[sort]{cleveref}
\AddToHook{cmd/appendix/before}{%
    \crefalias{section}{appendix}%
    \crefalias{subsection}{appendix}
}

\makeatletter
\Crefname{ALG@line}{Line}{Lines}
\renewcommand{\theHALG@line}{\thealgorithm.\arabic{ALG@line}}
\makeatother

\Crefname{assumption}{Assumption}{Assumptions}
\Crefname{claim}{Claim}{Claims}
\Crefname{protocol}{Protocol}{Protocols}
\Crefformat{equation}{Eq. #2(#1)#3}
\Crefrangeformat{equation}{Eqs. #3(#1)#4 to #5(#2)#6}
\Crefmultiformat{equation}{Eqs. #2(#1)#3}{ and #2(#1)#3}{, #2(#1)#3}{ and #2(#1)#3}
\Crefrangemultiformat{equation}{Eqs. #3(#1)#4 to #5(#2)#6}{ and #3(#1)#4 to #5(#2)#6}{, #3(#1)#4 to #5(#2)#6}{ and #3(#1)#4 to #5(#2)#6}

\newtheorem{theorem}{Theorem}
\newtheorem{lemma}[theorem]{Lemma}
\newtheorem{proposition}[theorem]{Proposition}
\newtheorem{corollary}[theorem]{Corollary}

\theoremstyle{definition}
\newtheorem{definition}{Definition}
\newtheorem{assumption}[definition]{Assumption}

\usepackage{tikz}
\usetikzlibrary{positioning}

\newcommand{\E}{\operatornamewithlimits{\mathbb{E}}}
\newcommand{\argmax}{\operatornamewithlimits{\mathrm{argmax}}}
\newcommand{\argmin}{\operatornamewithlimits{\mathrm{argmin}}}

\newcommand{\Law}{\operatornamewithlimits{\mathrm{PDF}}}
\newcommand{\truth}{\operatorname{\textsc{Truth}}}
\newcommand{\Unif}{\operatorname{Unif}}

\renewcommand{\O}{\operatorname{\mathcal O}}
\newcommand{\Otil}{\operatorname{\tilde{\O}}}

\newcommand{\trans}{\top}

\newcommand{\mA}{\mathcal{A}}
\newcommand{\mF}{\mathcal{F}}
\newcommand{\mT}{\mathcal{T}}
\newcommand{\mH}{\mathcal{H}}
\newcommand{\mE}{\mathcal{E}}

\newcommand{\mC}{\mathcal{C}}

\newcommand{\mV}{\mathcal{V}}
\newcommand{\mB}{\mathfrak{B}}
\newcommand{\mR}{\mathfrak{R}}
\newcommand{\1}{\mathbbm{1}}

\usepackage[table]{xcolor}
\usepackage{multirow}
\usepackage{arydshln}

\usepackage{xspace}
\newcommand{\algname}{\texttt{IAPD}\xspace}
\newcommand{\ftrl}{\texttt{FTRL}\xspace}
\newcommand{\oftrl}{\texttt{O-FTRL}\xspace}
\newcommand{\oftrlfp}{\texttt{O-FTRL-FP}\xspace}
\newcommand{\mech}{\texttt{IAPD.}\ftrl}
\newcommand{\mechO}{\texttt{IAPD.}\oftrlfp}
\newcommand{\primalreg}{\textsc{PrimalReg}\xspace}
\newcommand{\dualreg}{\textsc{DualReg}\xspace}

\usepackage{tikz}
\usetikzlibrary{arrows.meta, positioning, shapes.geometric, calc, decorations.pathreplacing, backgrounds, pgfplots.groupplots, patterns}
\usepackage{pgfplots}
\pgfplotsset{compat=1.18}
\usepackage{subcaption}

\renewcommand{\paragraph}[1]{\vspace{2pt}\noindent\textbf{#1}}

\begin{document}

\maketitle

\begin{abstract}
\noindent We study the dynamic allocation of indivisible resources to strategic agents under long-term constraints, where the planner aims to maximize social welfare, satisfy multiple constraints, and elicit near-truthful reports.
We find standard primal-dual methods fragile in this setting: agents easily manipulate their reports to distort dual variables, sacrificing social efficiency for individual utility. To address this, we propose the Incentive-Aware Primal-Dual (\algname) framework.
On the primal side, we integrate three components to suppress manipulation: a VCG-based payment neutralizes immediate misreporting benefits, while epoch-based lazy updates and random exploration together ensure potential future gains are outweighed by immediate penalties.
On the dual side, to overcome a learning barrier due to lazy updates---which we call the ``price of incentives''---we design a novel optimistic online learning algorithm, \oftrlfp. It utilizes a fixed-point oracle to resolve the circular dependency between optimistic dual variables and the resulting allocations.
Ultimately, our mechanism attains $\tilde{\mathcal O}(\sqrt T)$ social welfare regret, satisfies all long-term constraints, and induces a near-truthful equilibrium. It also smoothly generalizes to multi-unit multi-demand allocation problems. Notably, this $\tilde{\mathcal O}(\sqrt T)$ regret near-matches the non-strategic $\Omega(\sqrt T)$ lower bound, demonstrating that incentive-awareness can be accommodated at nearly no cost.
\end{abstract}

\section{Introduction}\label{sec:introduction}
Many firms and platforms face the challenge of dynamically allocating scarce resources to stakeholders with time-varying needs, subject to long-term constraints. For instance, cloud-based research platforms increasingly rely on GPUs to support AI and scientific workloads, which must be efficiently allocated to fulfill evolving demands while respecting energy and budget constraints \citep{perez2022dynamic,chen2023cloud}; sharing economy platforms need to dynamically match users with resources to maximize revenue subject to availability \citep{benjaafar2020operations}; in healthcare, governments need to distribute ventilators or mobile health units---crucial for delivering care in underserved areas---across jurisdictions with varying needs while constrained by staffing, transportation, and capacities \citep{mehrotra2020model,alban2022resource}.

In reality, such challenges are further exacerbated by agents' strategic behavior: Self-interested agencies or users may misreport their true demand in order to secure more favorable allocations, thus undermining system efficiency. In this paper, we study the problem of dynamically reassigning scarce resources---be they GPUs, rides, or mobile health units---to strategic agents over time. In these settings, a central planner---in charge of the dynamic allocation---must maximize overall social welfare while respecting multi-dimensional long-term constraints like staffing, emissions, or budgets. Henceforth, besides the standard planner's trade-off between \emph{efficiency} (satisfying more needs) and \emph{feasibility} (adhering to all constraints), it is also crucial for the allocation mechanism to be \emph{incentive-aware} regarding agents' strategic manipulation.

In terms of balancing efficiency and feasibility alone---a research problem known as constrained online resource allocation---the \emph{online primal-dual framework} serves as a powerful tool \citep{devanur2009adwords,golrezaei2014real, molinaro2014geometry,balseiro2020dual}. It offers a principled way to handle long-term constraints while adapting to changing demand by alternating between primal and dual components in every round: The dual component maintains a set of dual variables, acting as \emph{shadow prices}, that reflects the tightness of each constraint; based on these dual variables, the primal component then makes allocations that not only generate high efficiency but also respect the constraints. However, this framework assumes that agents' \emph{true valuation} for the resource is known in advance, ignoring any strategic responses from agents.

\begin{figure}[!t]
\centering
\includegraphics[width=\linewidth]{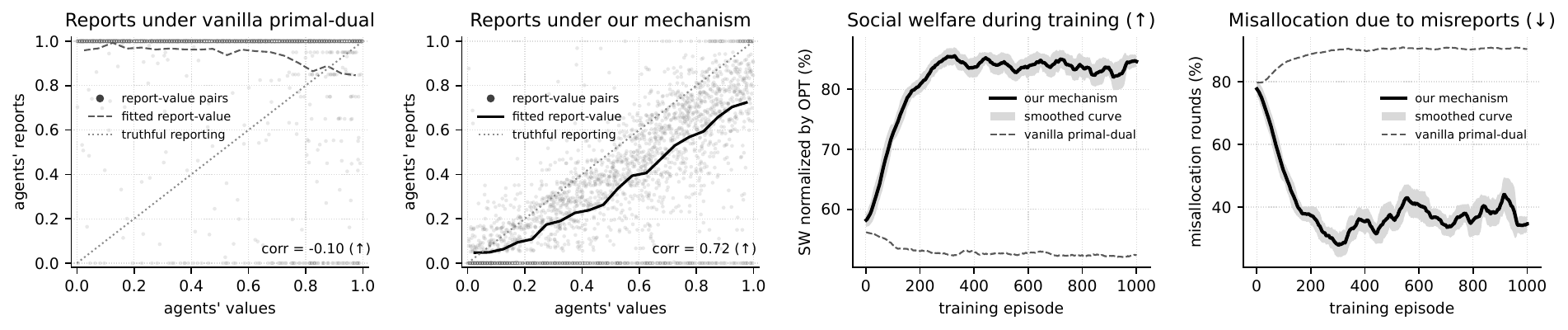}
\caption{Agents' strategic behavior over 1000 training trails in a 1000-round game. Two left panels plot agents' learned reports-versus-values under vanilla non-strategic primal-dual and under \mech, respectively, with the dotted diagonal denoting truthful reporting. The two right panels compare training dynamics of vanilla primal-dual (gray dashed) and \mech (black solid). 
Vanilla primal-dual induces misreporting and misallocation, while \mech elicits near-truthful reports, improves welfare, and reduces value lost from misallocation; 
see \Cref{sec:incentive-aware vs non-strategic,example:affect the future} for details.}
\label{fig:plot main}
\end{figure}

Indeed, as we illustrate in \Cref{fig:plot main}, online primal-dual methods are highly vulnerable to strategic agents: Such methods succeed in non-strategic cases by being highly reactive, i.e., instantly adjusting dual variables (shadow prices for constraints) based on the latest consumption. But this creates a predictable feedback loop that strategic agents can exploit. By inflating or withholding demand, agents easily manipulate future dual variables for their own interest, thus sacrificing social welfare for personal gain. This means that the online primal-dual framework---achieving optimal efficiency-feasibility trade-off with non-strategic agents---leads to severe welfare loss and resource misallocation in strategic cases.
This fragility raises a question:
\begin{center}
\textit{\textbf{(Q). When facing strategic agents who can misreport their true demand,\\is it still possible to optimally allocate resources subject to long-term constraints?}}
\end{center}

In other words, can a central planner attain efficiency, feasibility, and incentive-awareness \emph{simultaneously}? While achieving any two of these three objectives is well-studied, combining all three remains a trilemma. As discussed, standard primal-dual methods master the balance between efficiency and feasibility, but are fragile under strategic manipulation. Conversely, classical mechanism design---most notably the VCG mechanism \citep{vickrey1961counterspeculation,clarke1971multipart,groves1973incentives}---perfectly aligns individual incentives with social efficiency but lacks the capability to enforce long-term constraints. While a few recent studies explored this trilemma, they heavily rely on restrictive assumptions like homogeneous agents or specific constraints (as discussed in \Cref{sec:related work}). A general recipe for efficiency, feasibility, and incentive-awareness remains under-explored.

In this paper, we answer \textit{\textbf{(Q)}} affirmatively. We design a general algorithmic framework (named \textbf{I}ncentive-\textbf{A}ware \textbf{P}rimal-\textbf{D}ual, \algname) that recovers near-optimal social efficiency and adheres to all constraints when facing strategic agents. As we will overview in \Cref{sec:overview results}, this is achieved through two core innovations: \textit{(i)} the \algname framework that robustifies classical methods by \emph{breaking} the exploitable feedback loop between consumptions and duals, and \textit{(ii)} a novel optimistic learning method that \emph{anticipates} future consumptions to yield near-optimal duals, perfectly compensating for the delayed updates required for incentive-awareness.

Remarkably, our mechanism demonstrates that the central planner pays virtually \emph{no cost} for incentive-awareness. Our mechanism achieves $\Otil(\sqrt T)$ social welfare regret (defined as the efficiency gap between our allocations and the offline optima), which near-matches the $\Omega(\sqrt T)$ regret lower bound established for purely \emph{non-strategic} environments. We now overview our model and results in \Cref{sec:overview setup,sec:overview results}, respectively.

\subsection{Overview of Our Model}\label{sec:overview setup}

We consider a dynamic resource allocation game over $T$ rounds between a \emph{planner} and $K$ strategic and self-interested \emph{agents}. In each round, the planner allocates an indivisible resource to one agent (multi-unit multi-demand allocations are studied in \Cref{sec:extensions}), with the objective of maximizing social welfare subject to $d$ long-term \emph{constraints}. To align agents' individual incentives with this global objective, the planner utilizes the economic instrument of \emph{payments} (however, as we elaborate in \Cref{sec:mechanism}, long-term constraints inherently create inter-round dependencies, and standard single-round payment rules like VCG are insufficient).

Specifically, in each round $t=1,2,\ldots,T$, each agent $i=1,2,\ldots,K$ has a private value $v_{t,i}\in [0,1]$ for the resource, which is an independent sample from a fixed distribution $\mV_i$.
For the planner, allocating to agent $i$ incurs a consumption vector of $\bm c_{t,i}\in [0,1]^d$, which is independently sampled from another fixed distribution $\mC_i$.
All the distributions $\mV_i$ and $\mC_i$ are \emph{fixed but unknown} to the planner. In such a prior-free setting, the planner must learn the agents' type distributions on the fly based on their strategic feedback.

The planner announces a \emph{mechanism} $\bm M$---a collection of allocation and payment rules for each round---before the game.
The game proceeds as follows. In each round $t$, each agent $i$ observes their own private value $v_{t,i}$ and all agents' consumptions $(\bm c_{t,j})_{j=1}^K$.
Maximizing their expected discounted sum of value-minus-payment, agent $i$ strategically submits a \emph{report} $u_{t,i}$ that can differ from $v_{t,i}$. This happens simultaneously for all agents.
The planner, observing only the reports $\bm u_t=(u_{t,i})_{i}$ and consumptions $\bm c_t=(\bm c_{t,i})_{i}$, irrevocably allocates the resource to an agent $i_t$ and charges them a payment $p_{t,i_t}$ according to the announced mechanism $\bm M$.
The equilibrium concept we study is the {Perfect Bayesian Equilibrium} (PBE).
With formal setups and definitions detailed in \Cref{sec:setup}, on a high level, a desirable mechanism $\bm M$ should satisfy three properties:
\begin{enumerate}
\item \textbf{Efficiency.} The cumulative social welfare, $\E[\sum_{t=1}^T v_{t,i_t}]$, should be close to the offline optimum defined in \Cref{eq:offline optimal benchmark}. We call this difference the \emph{social welfare regret}, which is formalized in \Cref{eq:regret} of \Cref{sec:setup}.
\item \textbf{Feasibility.} The total consumption $\sum_{t=1}^T \bm c_{t,i_t}$ cannot exceed the pre-determined budget of $T\bm \rho$.
\item \textbf{Incentive-Awareness.} The mechanism should be robust to strategic manipulation. Specifically, it should induce a PBE where agents' reports are near-truthful, i.e., $\max_i\lvert u_{t,i} - v_{t,i}\rvert\le 1/\text{poly}(T)$ for $\Omega(T)$ rounds.
\end{enumerate}

\subsection{Overview of Our Results}\label{sec:overview results}
In this paper, we resolve the trilemma between efficiency, feasibility, and incentive-awareness by proposing the \textbf{I}ncentive-\textbf{A}ware \textbf{P}rimal-\textbf{D}ual (\algname) framework, thus answering \textbf{\textit{(Q)}} affirmatively.
Far from merely applying VCG payments to standard primal-dual methods, our \algname framework (and \algname-based mechanisms) introduces novel algorithmic designs to handle the complex interplay between learning and incentives:
\begin{enumerate}
\item \textbf{Blocking the Feedback Loop via Epoching.}
Standard primal-dual algorithms are vulnerable because they react {instantaneously} to reports, allowing agents to manipulate future prices via current actions. To mitigate this, we introduce \emph{epoch-based lazy updates}, where dual variables (the shadow prices) are frozen for intervals of time. This isolation decouples an agent's current report from immediate price fluctuations, allowing us to design payment rules---specifically, a \emph{dual-adjusted allocation and payment rule}, and a \emph{randomized exploration rule}---that align agents' individual incentives with the social welfare. However, this comes at a cost: Less frequent updates inherently slow down the learning of shadow prices, creating a theoretical {efficiency barrier}. This requires the second algorithmic innovation, detailed below.
\item \textbf{Restoring Efficiency via Predictability.}
To compensate for the efficiency loss due to infrequent updates, we switch from \emph{reacting} to the history to \emph{predicting} the future. Specifically, by exploiting the statistical predictability of demand produced by near-truthful agents, for any potential dual variable, we are able to predict the induced consumptions \emph{before} actually deploying it. While resonating with the optimistic online learning literature, there is a unique circular dependency between dual variables and consumptions: Predicting future consumptions requires knowing the dual being used, but deciding a near-optimal dual needs these predicted consumptions. To this end, we introduce a fixed-point-oracle-based online learning method called \oftrlfp---that can be of independent interest---to restore the near-optimal efficiency.
\end{enumerate}

Equipped with these algorithmic innovations, we develop two mechanisms based on the \algname framework and  establish their theoretical guarantees in terms of efficiency, feasibility, and incentive-awareness. These mechanisms offer a trade-off between algorithmic simplicity and theoretical optimality:

\paragraph{The Baseline Mechanism (\mech).}
In \Cref{thm:main theorem FTRL}, we first deploy the classical \ftrl algorithm as the dual-update subroutine. Although \ftrl is not optimized for the lazy updates and delayed feedback in \algname, we still yield a mechanism \mech inducing a Perfect Bayesian Equilibrium (PBE) $\bm \pi$ satisfying:
\begin{enumerate}
\item \textbf{Efficiency.} The social welfare regret incurred by $\bm \pi$ under \mech is bounded by $\Otil(T^{2/3})$. This implies that the per-round social welfare converges to the constrained optimum as $T\to \infty$.
\item \textbf{Feasibility.} All the long-term constraints are strictly satisfied with probability 1.
\item \textbf{Incentive-Awareness.} Agents only misreport at most $\Otil(T^{-1/3})$ in all but $\O(KT^{2/3})$ rounds.
\end{enumerate}

\paragraph{The Near-Optimal Mechanism (\mechO).}
We identify that the $\Otil(T^{2/3})$ regret is due to a fundamental limit resulted from the slow updates. Since epoching is necessary for incentive-awareness, we call the $\Omega(T^{2/3})$ barrier a \emph{price of incentives} (see more details in \Cref{sec:price of incentives}). To overcome it, we design a novel online learning method, \oftrlfp, which \textit{(i)} exploits the \emph{predictability} of near-truthful agents to accelerate learning, and \textit{(ii)} resolves a \textit{circular dependency} in our dynamic mechanism design problem by introducing fixed-point oracles.
The resulting \mechO mechanism---as shown in \Cref{thm:main theorem O-FTRL-FP}---improves the social welfare regret to $\Otil(\sqrt T)$, while still preserving feasibility and incentive-awareness.

Remarkably, $\Omega(\sqrt T)$ serves as the fundamental regret lower bound for online resource allocation \emph{even in non-strategic settings} \citep[Theorem 1]{arlotto2019uniformly}.
Therefore, our result demonstrates that the planner can achieve incentive-awareness at a \emph{minimal cost}: The regret of \mechO near-matches (up to $\Otil(1)$) that of non-strategic resource allocation algorithms, while additionally being incentive-aware.

\paragraph{Multi-Unit Resource Allocation.}
While the majority of this paper studies the single-unit resource allocation (i.e., the planner has one single indivisible resource in each round), in \Cref{sec:extensions}, we generalize to multi-unit multi-demand resource allocations. In that setting, the planner has multiple identical resources for allocation, and each agent has diminishing marginal values for resources. In \Cref{thm:multi-unit,thm:multi-demand}, we prove that, with a minimal amount of modifications, our mechanisms seamlessly maintain their $\Otil(T^{2/3})$ and $\Otil(\sqrt T)$ regret guarantees while remaining feasible and incentive-aware.

\subsection{Related Literature}\label{sec:related work}
\paragraph{Constrained Online Resource Allocation with Non-Strategic Agents.}
In non-strategic settings, constrained online resource allocation has been extensively studied. With perfectly known value and consumption distributions, adaptive policies can attain constant regret across a spectrum of problems. These range from the multi-secretary problem \citep{arlotto2019uniformly}---a special case with one-dimensional constraints---to multi-dimensional online packing \citep{vera2021online} or network revenue management \citep{sun2020near}---roughly equivalent to our constrained online resource allocation problem---and to more challenging setups like online bin packing \citep{banerjee2025good} or resource allocation with stochastically replenishing budgets \citep{vera2025dynamic}.
However, in our prior-free setup where distributions are unknown, even if agents remain non-strategic, an $\Omega(\sqrt T)$ regret lower bound is unavoidable \citep[Lemma 1]{arlotto2019uniformly}.

With unknown distributions, a rich body of online primal-dual algorithms has been developed to achieve the optimal $\O(\sqrt T)$ regret. \citet{li2023simple} consider improved computational efficiency of primal-dual under the i.i.d. model (all values and consumptions come from a fixed but unknown distributions). Several works study the more challenging secretary model, where an adversary decides all requests that are then randomly shuffled \citep{devanur2009adwords,feldman2010online,agrawal2014dynamic,gupta2016experts}; see the chapters by \citet{gupta2020random} and \citet{devanur2023online}.
The adversarial model \citep{balseiro2020dual} and the single-sample prophet model \citep{ghuge2025single} have also been studied, where $\Omega(T)$ regret is unavoidable and algorithms focus on the competitive ratio.
Another related problem is bandits with knapsacks, where the planner makes allocations without knowing the values or consumptions but rather learns them via post-decision feedback \citep{badanidiyuru2018bandits,sankararaman2018combinatorial}. This problem has been extended to adversarial models as well \citep{castiglioni2022online,immorlica2022adversarial}.

None of these works, however, consider the strategic behavior of agents: In these works, the values and consumptions are truthfully observable, either before or after the allocation decisions. In contrast, our main focus is to be \emph{incentive-aware} regarding agents' strategic manipulation while remaining efficient and feasible.

\paragraph{Online Resource Allocation with Strategic Agents.}
This stream of research focuses on resource allocation to strategic agents, aligning with our objectives of efficiency and incentive-awareness.
In static (single-round) environments, the VCG mechanisms \citep{vickrey1961counterspeculation,clarke1971multipart,groves1973incentives} provide a foundational framework for welfare maximization. Our setup---in addition to long-term constraints---deviates from these classics in two ways: First, we operate under the money-burning setting where payments cannot be redistributed \citep{hartline2008optimal,hoppe2009theory,condorelli2012money}. Second, we do not assume known value distributions, thus falling into the category of robust mechanism design \citep{wilson1987game,bergemann2005robust}. We refer the readers to \citet{parkes2007online} for a survey on online mechanism design.

With unknown value distributions, a closely related area is learning prices in repeated auctions. Various objectives have been studied, e.g., seller's revenue maximization facing strategic buyers \citep{amin2013learning,amin2014repeated,kanoria2014dynamic,cesa2014regret,mohri2014learning}, eliciting buyers' valuations from strategic reports \citep{braverman2018selling,golrezaei2021dynamic,golrezaei2023incentive, golrezaei2023pricing}, or buyers' bidding strategies under budget or return-on-investment (ROI) constraints \citep{balseiro2019learning, deng2023multi, castiglioni2024online,agrawal2024dynamic}. We direct the readers to \citet{nedelec2022learning} for a survey. We, on the other hand, study the different task of maximizing social welfare under \emph{planner's} constraints.

We also discuss a more challenging setting where monetary transfers are completely disallowed. In static setups, incentive-compatibility is generally impossible due to Arrow's impossibility theorem \citep{arrow1950difficulty,gi73,s75}, though positive results exist under specific utility structures \citep{miralles2012cardinal,guo2010strategy,han2011strategy,cole2013positive}.
In repeated settings, many recent works study non-monetary mechanisms for efficient resource allocations. They usually make much stronger assumptions than ours, for example known value distributions \citep{balseiro2019multiagent,gorokh2021monetary,blanchard2024near,galgana2026moneyless}, pre-determined fair share for agents \citep{gorokh2019remarkable,yin2022online,banerjee2023robust,fikioris2025beyond,onyeze2025allocating}, or extra verification power of the planner \citep{ben2014optimal,jalota2024catch,dai2025non}. Additionally, while these works often face no constraints or simplified fairness-based constraints, we study the general multi-dimensional constraints.

\paragraph{Constrained Online Resource Allocation with Strategic Agents.}
To our knowledge, the only attempt to attain efficiency, feasibility, and incentive-awareness in online resource allocation is by \citet{yin2022online}.
The first main difference is that they assume all the agents have \emph{identical} value distributions, which is crucial for their incentive-awareness: by comparing one agents' reports to the opponents', unilateral deviation from truth-telling is easily identifiable, which ensures incentive-awareness even without monetary transfers. Without homogeneity, we need more delicate algorithmic components, including dual-adjusted pricing, epoch-based lazy updates, and randomized exploration rounds to attain incentive-awareness.
Another main difference is the type of constraints. They study a specific type of ``fair share'' constraints, which says given some distribution $\bm p\in \triangle([K])$ across the agents, the number of allocations that each agent $i=1,2,\ldots,K$ receives should be $T p_i$. On the other hand, our multi-dimensional long-term constraint is strictly general.

We note that incentive-aware mechanisms have long been established in the context of multi-choice and matroid secretary problems \citep{kleinberg2005multiple,babaioff2007matroids}. These works operate under the random permutation model---where linear regret is unavoidable and the focus is a constant competitive ratio---and address combinatorial or single-capacity constraints. However, ensuring incentive-awareness while balancing general multi-dimensional constraints to achieve sublinear regret remains a non-trivial challenge.

\paragraph{Multi-Agent Learning.}
When a planner learns for better allocation mechanisms, the agents may also be learning to react (e.g., in our numerical illustration in \Cref{fig:plot main}, we let agents use Q-learning to learn their best reporting strategies under different mechanisms).
There is a rich literature studying the learning dynamics of agents in games. For example, \citet{balseiro2019learning,golrezaei2020no,dai2024refined,berriaud2024spend}, and \citet{galgana2025learning} characterized the convergence to equilibria when multiple agents deploy no-regret online learning algorithms at the same time in reaction to a fixed and known mechanism. While such results are stronger than our ``existence of equilibrium'' results in the sense that agents find such an equilibrium on their own, we remark that the main focus of this work is designing efficient, feasible, and incentive-aware mechanisms for the \emph{planner}, instead of algorithms for the agents.

\section{Problem Formulation}\label{sec:setup}
\paragraph{Notations.}
For integer $n\ge 1$, let $[n] := \{1,2,\ldots,n\}$. For set $\mathcal X$, let $\triangle(\mathcal X)$ be the set of all probability distributions over $\mathcal X$. We use boldface letters (e.g., $\bm v$) to denote vectors and non-bold letters (e.g., $v_i$) for their components.
$\O(\cdot)$ and $\Omega(\cdot)$ hide absolute constants, and $\Otil(\cdot)$ further hides polylogarithmic factors.

We consider the problem of dynamically reallocating an indivisible resource (multi-item cases are discussed in \Cref{sec:extensions}) over $T>0$ rounds from a central planner to $K>0$ strategic and self-interested agents, indexed by $1, 2, \ldots, K$. In each round $t=1,2,\ldots,T$, the planner irrevocably allocates the resource to one of the agents, aiming to maximize social welfare while satisfying $d$ long-term constraints simultaneously.

\subsection{Agents' Values and Costs, Planner's Allocations and Payments}\label{sec:setup values and consumptions}
In each round $t \in [T]$, agent $i \in [K]$ has a private scalar value $v_{t,i} \in [0,1]$ for the resource and a public $d$-dimensional consumption vector $\bm{c}_{t,i} \in [0,1]^d$. Consequently, allocating the resource to agent $i$ in round $t$ yields a value of $v_{t,i}$ to the agent and consumes $c_{t,i,j}$ units along each dimension $j=1,2,\ldots,d$.
We assume that values and consumptions are \emph{independently} generated across agents and rounds, but the underlying distributions are unknown to the planner. Specifically, we assume $v_{t,i}$ and $\bm{c}_{t,i}$ are drawn i.i.d. from some \emph{fixed but unknown} distributions $\mathcal{V}_i \in \triangle([0,1])$ and $\mathcal{C}_i \in \triangle([0,1]^d)$, respectively, for all $t$ and $i$.

Every agent $i\in [K]$, after observing their own private value $v_{t,i}$ and everyone's consumptions $\bm c_t$, strategically submits a report $u_{t,i}\in [0,1]$ to the planner. We defer the generation rule of reports to \Cref{sec:setup policy}.
After observing agents' strategic reports $\bm u_{t}$ and consumption vectors $\bm c_{t}$ (but without access to the true values $\bm v_{t}$), the planner either irrevocably allocates the resource to one of the agents $i_t\in [K]$ or forfeits it for this round, in which case we write $i_t=0$. Conventionally, $v_{t,0}=0$ and $\bm c_{t,0}=\bm 0$ for all $t\in [T]$, i.e., forfeiting the item does not generate any value nor consume anything. Assuming a null action ensures the existence of a feasible solution and is common in the resource allocation literature \citep[see, e.g.,][]{balseiro2020dual}.
After the allocation, the planner decides a payment to be charged from the winner $i_t$, denoted by $p_{t,i_t}\in \mathbb R$. For all remaining agents $i\ne i_t$, the payment $p_{t,i}=0$. The allocation and payment rule is also in \Cref{sec:setup policy}.

\subsection{History, Planner's Mechanism, and Agents' Strategies}\label{sec:setup policy}
At the beginning of round $t \in [T]$, the public history available to the planner includes all historical reports $\bm u_\tau$, consumption vectors $\bm c_\tau$, allocations $i_\tau$, and payment vectors $\bm p_\tau$, where $\tau=1,2,\ldots,t-1$. Formally, we denote the public history available to the planner in round $t\in [T]$ by $\mathcal{H}_{t,0} := \{(\bm{u}_\tau, \bm{c}_\tau, i_\tau, \bm p_\tau)\}_{\tau < t}$.

Each agent $i \in [K]$ additionally has access to their own historical values, namely $v_{\tau,i}$ where $\tau=1,2,\ldots,t-1$. We also assume that all agents know each other's value and consumption distribution. This assumption, as is standard in robust mechanism design \citep{wilson1987game,bergemann2005robust}, not only serves as a technical necessity to define Perfect Bayesian Equilibria (PBE), but also captures a practical challenge of information asymmetry: Market participants (agents) often possess localized knowledge about the industry and their peers, whereas the platform (planner) usually acts as an outsider and lacks prior information on agents.
Formally, the private history available to agent $i$ in round $t$ is $\mathcal{H}_{t,i} := \mathcal{H}_{t,0} \cup \{v_{\tau,i}\}_{\tau < t} \cup \{(\mathcal{V}_j, \mathcal{C}_j)\}_{j \in [K]}$.

\begin{algorithm}[t!]
\floatname{algorithm}{Protocol}
\caption{Dynamic Resource Allocation Game}
\label[protocol]{alg:protocol}
\begin{algorithmic}[1]
\Require {Number of rounds $T$, number of agents $K$, value distributions $\{ \mathcal{V}_i \}_{i \in [K]}$, cost distributions $\{ \mathcal{C}_i \}_{i \in [K]}$, planner's mechanism $\bm{M} = (M_t)_{t \in [T]}$, and agents' joint strategy $\bm{\pi} = (\pi_{t,i})_{t \in [T],\, i \in [K]}$.}

\State Initialize public history $\mathcal{H}_{1,0} \gets \varnothing$ and private histories $\mathcal{H}_{1,i} \gets \{(\mathcal{V}_j, \mathcal{C}_j)\}_{j \in [K]}$, $\forall i\in [K]$.

\For{each round $t = 1, 2, \dots, T$}
\State Each agent $i$ observes value $v_{t,i} \sim \mathcal{V}_i$ and consumption $\bm{c}_{t,i} \sim \mathcal{C}_i$.
\State Each agent $i$ submits report $u_{t,i} \sim \pi_{t,i}(v_{t,i}, \bm{c}_t; \mathcal{H}_{t,i})$.
\State Planner applies mechanism $(i_t, \bm p_t) \sim M_t(\bm{u}_t, \bm{c}_t; \mathcal{H}_{t,0})$.
\State Update public history: add $(\bm{u}_t, \bm{c}_t, i_t, p_{t,i_t})$ to $\mathcal{H}_{t+1,0}$.
\State Each agent $i$ updates private history: add $v_{t,i}$ and $(\bm{u}_t, \bm{c}_t, i_t, p_{t,i_t})$ to $\mathcal{H}_{t+1,i}$.
\EndFor
\end{algorithmic}
\end{algorithm}

In each round $t \in [T]$, the planner determines the allocation and payment $(i_t,\bm p_t)$ based on agents' current-round reports $\bm{u}_t$, consumption vectors $\bm{c}_t$, public history $\mathcal{H}_{t,0}$, and possibly some internal randomness used to break ties or randomize decisions. We write $i_t=0$ and $\bm p_t=\bm 0$ when the allocation is forfeited.
We denote the mapping from $(\bm u_t,\bm c_t,\mH_{t,0})$ to $(i_t,\bm p_t)$ by $M_t$, which may be randomized. We write $(i_t,\bm p_t)\sim M_t(\bm u_t,\bm c_t;\mH_{t,0})$ be a realization of the randomized allocation and payment given reports, consumptions, and history. We call the collection of all allocation and payment rules $\bm{M} = (M_t)_{t \in [T]}$ the planner's \textit{mechanism}.
We assume the planner has full power of commitment to any mechanism $\bm M$. Since it is always to the best interest of the planner to announce the commited mechanism \citep[see, e.g.,][]{borgers2015introduction}, we assume $\bm M$ is public information.

For each agent $i \in [K]$ and round $t\in [T]$, their report $u_{t,i}$ is determined based on their private value $v_{t,i}$, all agents' consumptions $\bm{c}_t$, their private history $\mathcal{H}_{t,i}$, and some randomness. Similar to the planner's case, we write $u_{t,i}\sim \pi_{t,i}(v_{t,i},\bm c_t;\mH_{t,i})$ as a realization of the randomized report.
All $\pi_{t,i}$'s collectively form agent $i$'s \emph{strategy} $\bm{\pi}_i := (\pi_{t,i})_{t \in [T]}$, and all agents' strategies together constitute a \emph{joint strategy} $\bm{\pi} := (\bm{\pi}_i)_{i \in [K]}$.
We summarize the overall interaction process between the planner and agents in \Cref{alg:protocol}. 

\subsection{Agents' Utility and Perfect Bayesian Equilibrium}\label{sec:setup agents}
To model agents'  behavior in a dynamic environment, we adopt the $\gamma$-impatient agent model by \citet{amin2013learning} and \citet{golrezaei2021dynamic,golrezaei2023incentive}, which captures the idea that agents often prioritize immediate rewards over long-term gains---due to bounded rationality, uncertainty about future rounds, or limited planning horizons---while the planner is more patiently optimizing the long-run social welfare.
\begin{assumption}[{$\gamma$-Impatient Agents}]\label{def:impatient}
For some fixed but unknown constant $\gamma\in (0,1)$, every agent $i\in [K]$ is \emph{$\gamma$-impatient} in the sense that they maximize their $\gamma$-discounted utility (defined in \Cref{eq:impatient}). Specifically, under agents' joint strategy $\bm \pi$ and planner mechanism $\bm M$, the $\gamma$-discounted utility of any agent $i\in [K]$ is
\begin{equation}\begin{aligned}\label{eq:impatient}
V_i^{\gamma}(\bm \pi;\bm M):=\E_{\text{\Cref{alg:protocol}}}\left [\sum_{t=1}^T \gamma^t (v_{t,i}-p_{t,i}) \1[i_t=i]\right ].
\end{aligned}\end{equation}
\end{assumption}

The parameter $\gamma$ essentially controls the ``farsightedness'' of agents: 
When $\gamma$ decreases, agents myopically optimize their immediate utility in each round without looking into the future; when $\gamma$ increases, agents instead maximize their undiscounted utility over the remainder of the game, as typical in extensive-form games. As a remark, we do \emph{not} require the constant $\gamma\in (0,1)$ to be known to the planner.

Having specified agents' utilities, we now define the Perfect Bayesian Equilibrium (PBE) of agents:
\begin{definition}[Perfect Bayesian Equilibrium]\label{def:PBE}
Fix any mechanism $\bm M$. Agents' joint strategy $\bm \pi$ is a Perfect Bayesian  Equilibrium (PBE) under $\bm M$, if any single agent's unilateral deviation from $\bm \pi$ does not increase their own $\gamma$-discounted utility $V_i^\gamma$. Formally, a joint strategy $\bm \pi$ is a PBE under mechanism $\bm M$ if
\begin{equation*}\begin{aligned}
&V_i^\gamma(\bm \pi;\bm M)\ge V_i^\gamma (\bm \pi_i'\circ \bm \pi_{-i};\bm M),\quad \forall i\in [K],\forall\bm \pi_i',
\end{aligned}\end{equation*}
where $\bm \pi_i'\circ \bm \pi_{-i}$ denotes the unilaterally deviated strategy: $\bm \pi_i'\circ \bm \pi_{-i}:=(\bm \pi_1,\ldots,\bm \pi_{i-1},\bm \pi_i',\bm \pi_{i+1},\ldots,\bm \pi_K)$.
\end{definition}

\subsection{Planner's Regret and Design Objectives}\label{sec:setup planner}
The benevolent planner, in contrast to agents who are self-interested and short-sighted, is patiently optimizing the cumulative \emph{social welfare} (expected total value yielded from allocations $\E[\sum_{t=1}^T v_{t,i_t}]$) while satisfying the $d$ long-term constraints $\frac{1}{T} \sum_{t=1}^T \bm{c}_{t,i_t} \le \bm{\rho}$.
We compare the planner's allocations against the \emph{offline optimal benchmark}, which performs a hindsight optimization using agents’ true values and consumptions:
\begin{equation}\begin{aligned}
\{i_t^\ast\}_{t \in [T]} := &\argmax_{i_1,\ldots,i_T \in \{0\} \cup [K]} \sum_{t=1}^T v_{t,i_t} \qquad \text{subject to} \quad \frac{1}{T} \sum_{t=1}^T \bm{c}_{t,i_t} \le \bm{\rho}. \label{eq:offline optimal benchmark}
\end{aligned}\end{equation}

We remark that the $\{i_t^\ast\}_{t\in [T]}$ benchmark in \Cref{eq:offline optimal benchmark} depends on the full sequence of values $\{\bm{v}_t\}_{t\in [T]}$ and consumptions $\{\bm{c}_t\}_{t\in [T]}$ that are \textit{(i)} not realized until the end of the game, and \textit{(ii)} not observable to the planner even in hindsight. This distinguishes it from typical online learning benchmarks, which commit to a fixed policy before the game (see \citealt{balseiro2020dual} for a related discussion).
Since it relies on information unavailable at decision time, it cannot be matched exactly by an online mechanism. Instead, we evaluate a mechanism $\bm M$ by its \emph{social welfare regret} (or \emph{regret} in short) w.r.t. the offline optimal benchmark $\{i_t^\ast\}_{t\in [T]}$:
\begin{equation}\label{eq:regret}
\mR_T(\bm{\pi}, \bm{M}) := \E_{\text{\Cref{alg:protocol}}} \left[ \sum_{t=1}^T \left(v_{t,i_t^\ast} - v_{t,i_t}\right) \right],
\end{equation}
where the expectation is taken over the randomness in the planner's mechanism $\bm M$, agent strategies $\bm \pi$, and the realizations of values and consumptions.
We also evaluate $\bm M$ by its cumulative \emph{constraint violations}:
\begin{equation}\label{eq:constr violation}
\mB_T(\bm{\pi}, \bm{M}) := \E_{\text{\Cref{alg:protocol}}} \Bigg[ \Big\lVert \Big( \sum_{t=1}^T (\bm{c}_{t,i_t} - \bm{\rho}) \Big)_+ \Big\rVert_1 \Bigg],
\end{equation}
where $(\cdot)_+$ is the coordinate-wise maximum with zero, i.e., $\bm x_+:=[\max(x_i,0)]_{i\in [d]}$.

We are now able to formalize the three design objectives informally mentioned in \Cref{sec:overview setup}: we aim to design a mechanism $\bm M$ for the planner, such that there exists a joint strategy $\bm \pi$ of agents ensuring:
\begin{enumerate}
\item \textbf{Efficiency.} The social welfare regret $\mR_T(\bm \pi,\bm M)$ should be as small as possible, ideally $\Otil(\sqrt T)$.
\item \textbf{Feasibility.} The constraint violation $\mB_T(\bm \pi,\bm M)$ should be zero, i.e., all constraints must be feasible.
\item \textbf{Incentive-Awareness.} This $\bm \pi$ should be a PBE under $\bm M$, as defined in \Cref{def:PBE}.
\end{enumerate}

Compared to the non-strategic resource allocation setting studied by \citet{balseiro2020dual}, the third objective of \emph{incentive-awareness} is new. Specifically, in their setting, agents \emph{always} truthfully report $u_{t,i} = v_{t,i}$ (denoted by $\bm \pi=\truth$). Under this truth-telling assumption, their mechanism $\bm M^{\truth}$ achieved efficiency $\mR_T(\truth, \bm M^{\truth}) = \Otil(\sqrt T)$ and feasibility $\mB_T(\truth,\bm M^{\truth})=0$.
However, as we numerically demonstrated in \Cref{fig:plot main}, should agents be strategic they will deviate from $\truth$ to an \emph{over-reporting} strategy. In other words, $\truth$ fails to be a PBE under their mechanism $\bm M^{\truth}$, i.e., $\bm M^{\truth}$ is \emph{not} incentive-aware.

In contrast, based on our proposed Incentive-Aware Primal-Dual (\algname) framework, our mechanism \mechO (see \Cref{sec:mechanism}) attains the three objectives \emph{at the same time}: it induces a \emph{PBE strategy} $\bm{\pi}$ such that no agent benefits from unilateral deviation---thus \emph{incentive-aware}---under which we still achieve \emph{efficiency} $\mR_T(\bm \pi,\mechO)=\Otil(\sqrt T)$ and \emph{feasibility} $\mB_T(\bm \pi,\mechO)=0$.

We conclude this section with a smoothness assumption on the (unknown) consumption distributions. The idea is to ensure that projected consumptions---i.e., linear combinations of the consumption vector---do not place excessive probability mass on any single value. This smoothness condition prevents pathological behaviors where a small change in an agent's report could drastically alter outcomes due to spiky distributions. Such assumptions are common in strategic settings to ensure robustness to perturbations, including in bilateral trades \citep{cesa2024regret}, first-price auctions \citep{cesa2024role}, second-price auctions with reserves \citep{golrezaei2021dynamic}, and smoothed revenue maximization \citep{durvasula2023smoothed}.

\begin{assumption}[Smooth Consumptions]\label{assump:smooth consumptions}
For any agent $i \in [K]$, the consumption distribution $\mathcal{C}_i$ satisfies the following \emph{smoothness assumption}: for all $\bm{\lambda} \in \bm{\Lambda} := \{ \bm{\lambda} \in \mathbb{R}^d \mid \lambda_j \in [0, \rho_j^{-1}] \}$, the density function of the random variable $\bm{\lambda}^\trans \bm{c}_i$ (i.e., projected consumption) is uniformly bounded by a universal constant $\epsilon_c>0$.
\end{assumption}

Similar to the $\gamma$ in \Cref{def:impatient}, we do not assume the constant $\epsilon_c$ to be known to the planner.

\section{Primal Design: Incentive-Aware Primal-Dual Framework}\label{sec:mechanism}
In this section, we present our Incentive-Aware Primal-Dual (\algname) framework. It is a generic algorithmic framework that makes efficient allocations while screening strategic agents; we detail the design in \Cref{sec:mechanism primal}.
In addition, the \algname framework requires a dual-update subroutine that sets dual variables based on the consumptions and budgets; we defer the design of such dual-update subroutines to \Cref{sec:mechanism dual}.
Before introducing \algname, we begin by dissecting the failure of standard primal-dual when facing strategic agents.

\subsection{Challenge in Strategic Settings: Direct and Indirect Impact of Misreports}\label{sec:vicious cycle}

\begin{figure}[htb]
\centering
\begin{tikzpicture}[
auto,
every node/.style={scale=0.7},
block/.style={rectangle, draw, fill=blue!10, text width=5cm, minimum height=1.5cm, align=center, rounded corners},
arrow/.style={->, thick}
]
\node[block, fill=black!5!white] (primal) {
\textbf{Primal} (Make Allocations) \\
{$i_t\gets \argmax_i (v_{t,i}-\bm \lambda_t^\trans \bm c_{t,i})$}
};
\node[block, fill=black!5!white, right=4cm of primal] (dual) {
\textbf{Dual} (Track Consumptions) \\
{$\bm \lambda_{t+1}\gets \text{DualUpd}(\bm\lambda_t,\bm c_{t,i_t})$}
};

\draw[arrow] (dual.175) to[bend right=10] node[midway, above] {Dual $\bm{\lambda}_t\in \mathbb R_{\ge 0}^d$} (primal.5);
\draw[arrow] (primal.355) to[bend right=10] node[midway, below] {Allocation $i_t\in [K]$} (dual.185);

\draw[->, thick] ([xshift=-3.5cm]primal.north) -- ([xshift=-3.5cm]primal.south) node[midway, right=0.3cm, rotate=270, anchor=center, align=center] {$t=1,2,\ldots,T$};
\end{tikzpicture}
\caption{Standard primal-dual framework for efficiency and feasibility.}
\label{fig:primal-dual}
\end{figure}

In non-strategic settings, the primal-dual framework is the prominent paradigm for online resource allocation. We illustrate its main idea in \Cref{fig:primal-dual}: Throughout the game, the framework maintains a dual variable $\bm \lambda_t\in \mathbb R_{\ge 0}^d$ that serves as a shadow price for cumulative consumptions. That is, for each consumption dimension $j\in [d]$, $\lambda_{t,j}$ increases if the budget is depleted, and decreases otherwise. It alternates between two components:
\begin{itemize}
\item \textbf{Primal Component.} Allocate to the agent maximizing the \emph{dual-adjusted value} $v_{t,i} - \bm \lambda_t^\trans \bm c_{t,i}$. This dual-adjustment component essentially penalizes agents who consume heavily on those scarce dimensions.
\item \textbf{Dual Component.} Find the new dual variable $\bm \lambda_{t+1}$ based on the realized consumption $\bm c_{t,i_t}$, aiming to ensure long-term feasibility without degrading the efficiency. The update rule $\bm \lambda_{t+1}\gets \text{DualUpd}(\bm \lambda_t,\bm c_{t,i_t})$ can be instantiated by various online learning algorithms (for instance, Online Gradient Descent).
\end{itemize}

\begin{figure}[htb]
\centering
\begin{tikzpicture}[every node/.style={scale=0.7, text width=5cm}, node/.style={rectangle, rounded corners, draw=black, fill=black!5!white, minimum height=1cm, text centered}, arrow/.style={->, draw=black!70!white, line width=0.8pt}]
\node[node] (manipulate) {Misreport $u_{t,i}\ne v_{t,i}$};
\node[node, below=0.7cm of manipulate] (alter) {Alter allocation $i_t$};
\draw[arrow] (manipulate) to node[auto, font=\bfseries] {Direct Impact} (alter);
\node[node, right=0.8cm of manipulate] (distort) {Distort dual $\bm \lambda_{t+1}$};
\draw[arrow] (manipulate) to (distort);
\node[node, right=0.8cm of distort] (consequence) {Affect future allocation $i_{t+1}$};
\draw[arrow] (distort) to (consequence);
\draw[arrow] (consequence.190) to[out=-130, in=-30, looseness=0.4] node[midway, below, font=\bfseries] {\color{black}Indirect Impact} (manipulate.350);
\end{tikzpicture}
\caption{Direct and indirect impact of a misreport in round $t$.}
\label{fig:vicious cycle}
\end{figure}

This design balances efficiency and feasibility when reports are truthful, ensuring $\Otil(\sqrt T)$ social welfare regret while adhering to constraints \citep[see, e.g., Theorem 1 of][]{balseiro2020dual}. However, when agents are strategic, the primal-dual feedback loop---i.e., the dual $\bm \lambda_{t+1}$ reacting immediately to the consumption $\bm c_{t,i_t}$---becomes a vulnerability. As shown in \Cref{fig:vicious cycle}, a strategic misreport $u_{t,i} \ne v_{t,i}$ creates \emph{two} impacts:
\begin{enumerate}
\item \textbf{Direct Impact.} It alters the current allocation $i_t$ by changing the $\argmax$. This is a classical mechanism design problem that can typically be resolved via appropriately designed allocation and payment rules.
\item \textbf{Indirect Impact.} The altered allocation $i_t$ leads to a distinct consumption $\bm c_{t,i_t}$, which in turn distorts the dual update $\bm \lambda_{t+1}$. This affects all future allocations and payments, and \emph{cannot} be isolated ex-post.
\end{enumerate}

Strategic agents can exploit this ``indirect impact'' to manipulate future prices in their favor---for example, by under-consuming in the current round, an agent can artificially lower the shadow prices for later rounds. Indeed, merely equipping the standard primal-dual method with VCG-style payments---which only addresses the direct impact---is insufficient for incentive-awareness; we direct the readers to \Cref{sec:numerical illustration setup} for a concrete example. This necessitates our \algname framework, which resolves \emph{both} the direct and indirect impacts.

\subsection{Algorithmic Design: Tame Incentives via Payments, Epoching, and Exploration}\label{sec:mechanism primal}

\begin{algorithm}[t!]
\caption{\textbf{I}ncentive-\textbf{A}ware \textbf{P}rimal-\textbf{D}ual (\algname) Framework}
\label{alg:mech}
\begin{algorithmic}[1]
\Require{Number of rounds $T$, number of agents $K$, number of constraints $d$, per-round budget $\bm{\rho} \in [0,1]^d$, an epoching scheme $\{\mathcal{E}_\ell\}_{\ell=1}^L$ such that $[T]=\bigcup_{\ell=1}^L \mE_\ell$, and a dual-update subroutine $\mA$.}
\Ensure Allocations and payments $(i_1,\bm p_1), (i_2, \bm p_2), \ldots, (i_T,\bm p_T)$, where $i_t = 0$ denotes no allocation.

\For{epoch $\ell = 1, 2, \ldots, L$}\Comment{Epoch-Based Lazy Updates}
\State Set epoch-$\ell$ dual variable $\bm{\lambda}_\ell \in \bm \Lambda$ via a \textit{dual-update subroutine} $\mA$, presented in \Cref{sec:mechanism dual}. \label{line:dual update}\label{line:call subroutine}\label{line:epoching}

\For{each round $t \in \mathcal{E}_\ell$}
\State Each agent $i \in [K]$ observes their value $v_{t,i} \sim \mathcal{V}_i$ and consumption vector $\bm{c}_t \sim \mathcal{C}$.
\State Agent reports $u_{t,i} \in [0,1]$ according to \Cref{alg:protocol}. $\bm{u}_t$ and $\bm{c}_t$ are revealed to the planner.

\If{round $t$ is selected for exploration (w.p. $1/|\mathcal{E}_\ell|$)} \Comment{Randomized Exploration}
\State Sample tentative agent and payment $(i,p)$ as
\algeq{$\displaystyle i\sim \Unif([K]),\quad p\sim \Unif([0,1]).$}
\State If $u_{t,i} \ge p$, set $i_t = i$ and $p_{t,i_t} = p$. Otherwise, set $i_t = 0$ (i.e., forfeit the resource).\label{line:exploration}
\Else \Comment{Dual-Adjusted Payments}
\State Compute dual-adjusted report for all agents as
\algeq{$\displaystyle \tilde{u}_{t,i} := u_{t,i} - \bm{\lambda}_\ell^\trans \bm{c}_{t,i},\quad \forall i\in [K].$}
\State Decide allocation and payment $(i_t,\bm p_t)$ for this round as \label{line:end of epoch game}
\algeq[line:primal allocation]{$\displaystyle  i_t = \argmax_{i=1,2,\ldots,K\vphantom{[K]}} \tilde{u}_{t,i};\quad p_{t,i_t} = \bm{\lambda}_t^\trans \bm{c}_{t,i_t} + \max_{j \ne i_t} \tilde{u}_{t,j};\quad p_{t,i}=0,\forall i\ne i_t.$}
\EndIf

\State If $\sum_{\tau < t} \bm{c}_{\tau, i_\tau} + \bm{c}_{t, i_t} \not\le T \bm{\rho}$, reject the allocation by setting $i_t = 0$ and $\bm p_t=\bm 0$.\label{line:safety}
\EndFor
\EndFor
\end{algorithmic}
\end{algorithm}

Our \textbf{I}ncentive-\textbf{A}ware \textbf{P}rimal-\textbf{D}ual (\algname) framework is presented in \Cref{alg:mech}. It addresses both the direct and indirect incentive effects through three algorithmic modifications: dual-adjusted payments to incentivize truthful reporting, epoch-based lazy updates to stabilize the evolution of the dual variables, and randomized exploration to deter misreporting. At a high level, \algname consists of a primal update, which selects an action in an incentive-aware manner given the current dual variable, and a dual update, which adjusts the shadow prices of consumptions over time. We leave the dual update as a flexible subroutine $\mA$ (see \Cref{line:call subroutine} of \Cref{alg:mech}) and later present two concrete approaches in \Cref{sec:mechanism dual}.

As in standard primal-dual methods, a $d$-dimensional dual variable $\bm \lambda_t\in \mathbb R_{\ge 0}^d$ is maintained as a vector of shadow prices for resource consumption. Given a dual variable $\bm \lambda_t$ decided by the subroutine $\mA$, standard primal-dual would directly maximize the dual-adjusted value of each action; see the left panel of \Cref{fig:primal-dual}. In contrast, \algname performs this maximization in an incentive-aware manner.
For notational simplicity, let
\begin{equation*}
\tilde v_{t,i}:=v_{t,i}-\bm \lambda_t^\trans \bm c_{t,i},~
\tilde u_{t,i}:=u_{t,i}-\bm \lambda_t^\trans \bm c_{t,i},\quad \forall i\in [K],
\end{equation*}
which we refer to as the dual-adjusted value and dual-adjusted report of agent $i$, respectively. We also define the dual-adjusted optimal action in round $t\in [T]$ as the maximizer of the dual-adjusted value, namely
\begin{equation}\label{eq:tilde i definition main}
\tilde i_t^\ast:=\argmax_{i\in [K]} \tilde v_{t,i}
=\argmax_{i\in [K]} \left(v_{t,i}-\bm \lambda_t^\trans \bm c_{t,i}\right).
\end{equation}

\subsubsection{Dual-Adjusted Allocations and Payments.}\label{sec:dual-adjusted payments}
We first address the direct incentives via payments. Given the objective of optimizing the dual-adjusted value, aside from exploration rounds, the planner should allocate the resource to the agent with the maximal \textit{dual-adjusted report} $\tilde u_{t,i}:=u_{t,i}-\bm \lambda_t^\trans \bm c_{t,i}$. To incentivize truthfulness, we design a payment rule inspired by the VCG mechanism \citep{vickrey1961counterspeculation,clarke1971multipart,groves1973incentives} and boosted second-price auctions \citep{golrezaei2021boosted}. The winner $i_t$ pays their consumption cost $\bm \lambda_t^\trans \bm c_{t,i_t}$ plus the second-highest dual-adjusted report $\max_{j\ne i_t} \tilde u_{t,j}$, as shown in \Cref{line:primal allocation} of \Cref{alg:mech}.

In \Cref{thm:informal IntraEpoch} from \Cref{sec:sketch PrimalAlloc}, we show that if agents are \emph{myopic}---that is, if they choose their reports to maximize only their current-round utilities---then this dual-adjusted payment rule makes truth-telling a dominant strategy within each round. This neutralizes the direct impact illustrated in \Cref{fig:vicious cycle}.

\subsubsection{Epoch-Based Lazy Updates.}\label{sec:epoch-based lazy updates}
We now turn to the indirect impact.
While misreporting does not yield immediate benefits thanks to \Cref{sec:dual-adjusted payments}, it may bring long-term gains via the indirect impact on future dual variables and, consequently, future allocations. Such an indirect channel may incentivize a non-myopic agent to misreport. We therefore weaken the immediate link between current actions and future duals: We divide the game horizon $[T]$ into $L$ epochs $\mE_1 \cup \cdots \cup \mE_L$, and update the dual variables based on the aggregate information of a full epoch (rather than instantaneously after every round).

Specifically, during each epoch $\mE_\ell$, we hold a constant dual variable $\bm \lambda_t\equiv \bm \lambda_\ell$. The update of $\bm \lambda_\ell$ occurs only at the boundary of two epochs. This ``lazy update'' creates a buffer: To influence the next dual $\bm \lambda_{\ell+1}$, an agent must sustain misreports over the current epoch. For $\gamma$-impatient agents (\Cref{def:impatient}), the immediate utility loss from sustained misreporting---detailed in \Cref{sec:random exploration}---is less attractive than discounted future benefits. Aggregating reports and consumptions over an epoch helps limit the long-term effect of misreports.

\subsubsection{Randomized Exploration.}\label{sec:random exploration}
So far, dual-adjusted payments ensure that misreporting is not directly profitable, while epoch-based lazy updates substantially limit its future influence. However, an agent may still attempt to manipulate if misreporting carries potential upside without meaningful downside. We introduce randomized exploration to create such a downside by making deviations from truthful reporting costly.

Specifically, for each round $t$ in epoch $\mE_\ell$, with a small probability, the planner initiates an exploration round, offering a random price $p\sim \Unif([0,1])$ to a randomly selected agent $i\sim \Unif([K])$ (see \Cref{line:exploration}).
Consequently, if an agent $i$ is being explored when over-reporting $u_{t,i}>v_{t,i}$, then they may pay a price higher than their true value; on the other hand, under-reporting $u_{t,i}<v_{t,i}$ means possibly losing a profitable allocation (see \Cref{thm:informal InterEpoch}).
In summary, this structure imposes an immediate utility loss when the report deviates from the true value, thereby further discouraging strategic behavior.
While randomized pricing has been explored in repeated second-price auctions \citep{amin2013learning, golrezaei2021dynamic,golrezaei2023incentive}, to our knowledge, this work is the first to leverage it for incentive-aware primal-dual learning.

\subsection{\algname Guarantee: Incentive-Awareness and Primal Efficiency}\label{sec:IAPD main theorem}
Through the components in \Cref{sec:mechanism primal}, our \algname framework ensures two properties: Agents make near-truthful reports; consequently, allocations nearly maximizes dual-adjusted values $\tilde v_{t,i}=v_{t,i}-\bm \lambda_t^\trans \bm c_{t,i}$.
The first property is precisely the incentive-awareness we defined in \Cref{sec:setup}, whereas the second is more subtle: 

The efficiency metric---the social welfare regret $\mR_T$ in \Cref{eq:regret}---compares the planner's allocation $i_t$ to the offline optimal allocation $i_t^\ast$, but the allocation rule of \algname only ensures that $i_t$ is close to the dual-adjusted optimal $\tilde i_t^\ast$ induced by the dual variable $\bm \lambda_t$; recall \Cref{eq:tilde i definition main,line:primal allocation}.
Therefore, the \algname framework alone cannot attain the efficiency defined in \Cref{sec:setup}. Instead, it only optimizes the \emph{primal efficiency} by minimizing the discrepancy between $i_t$ and $\tilde i_t^\ast$. The other part of efficiency, namely the \emph{dual efficiency} measuring the difference between $\tilde i_t^\ast$ and $i_t^\ast$, relies on the dual-update subroutine $\mA$, which we present in \Cref{sec:mechanism dual}.

A separate source of allocation inefficiency arises from the hard constraints on consumptions.
To satisfy constraints $\frac 1T\sum_{t=1}^T \bm c_{t,i_t}\le \bm \rho$ strictly, \Cref{line:safety} rejects an allocation $i_t$ once some constraint will be violated. Fortunately, this can also be handled by the dual-update subroutine $\mA$: Intuitively, properly set dual variables ensure that budgets almost never deplete.
To formalize, we decompose the regret $\mR_T$ in \Cref{eq:regret} as:
\begin{equation}\begin{aligned}
\mR_T&\le \mathbb{E}\Bigg [\underbrace{\sum_{t=1}^{\mT_v} (v_{t,\tilde i_t^\ast}-v_{t,i_t})}_{\primalreg}+\underbrace{\sum_{t=1}^T v_{t,i_t^\ast}-\sum_{t=1}^{\mT_v} v_{t,\tilde i_t^\ast}}_{\dualreg}\Bigg ], \label{eq:regret decomposition informal}
\end{aligned}\end{equation}
where $\mT_v:=\min\{t\in [T]\mid \sum_{{\tau}=1}^{t} \bm c_{{\tau},i_{\tau}} + \bm 1\not \le T\bm \rho\}\cup \{T+1\}$ is a stopping time (random variable), representing  the first round where  \Cref{line:safety} of the algorithm rejects an allocation.
The \algname framework then---in addition to incentive-awareness---controls the \primalreg term in \Cref{eq:regret decomposition informal}. We formalize this in \Cref{thm:informal IAPD main theorem}.
\begin{theorem}[\algname Framework]\label{thm:informal IAPD main theorem}\label{thm:informal InterEpoch}
Under the \algname framework in \Cref{alg:mech}, regardless of the dual-update subroutine $\mA$ and epoching scheme $[T]=\bigcup_{\ell=1}^L \mE_\ell$, there exists a PBE $\bm \pi$ (\Cref{def:PBE}) such that:
\begin{itemize}
\item \textbf{Incentive-Awareness.} For any epoch $\ell \in [L]$ with epoch length $\lvert \mE_\ell\rvert$, for any agent $i\in [K]$, with probability $1-\frac{1}{\lvert \mE_\ell\rvert}$, the number of $t\in \mE_\ell$'s where $|u_{t,i} - v_{t,i}| \ge \frac{1}{|\mE_\ell|}$ is no more than $\log_{\gamma^{-1}}\bigl (1+\O(K\lvert \mE_\ell\rvert^4)\bigr )=\Otil(1)$.
\item \textbf{Primal Efficiency.}
For any epoch $\ell\in [L]$, with probability $1-\frac{2}{\lvert \mE_\ell\rvert}$, aside from ``large misreport'' rounds in the previous bullet point, the number of $t\in \mE_\ell$'s where $i_t\ne \tilde i_t^\ast$ is bounded by $\O(K^2\log \lvert \mE_\ell\rvert)=\Otil(K^2)$.
\end{itemize}
Consequently, the \primalreg term in \Cref{eq:regret decomposition informal} satisfies $\E[\primalreg]=\E[\sum_{t=1}^T (v_{t,\tilde i_t^\ast}-v_{t,i_t})]\le \Otil(LK^2)$.
\end{theorem}

With the formal proof postponed to \Cref{sec:appendix PrimalAlloc}, we highlight the key proof techniques in \Cref{sec:sketch PrimalAlloc}.

\subsection{Proof Sketch of \Cref{thm:informal IAPD main theorem}}\label{sec:sketch PrimalAlloc}
We now provide a proof sketch for \Cref{thm:informal IAPD main theorem}. The argument proceeds in three steps. First, we isolate the direct impact of misreports by considering a nearsighted benchmark, in which agents optimize only their current-epoch utilities. In this benchmark, the dual variables are fixed throughout the epoch, and we show that dual-adjusted payments induce truthful reporting. Second, we return to the original farsighted setting and show that large deviations from truth-telling are unattractive: Epoch-based lazy updates delay the future benefits of manipulation, while randomized exploration imposes an immediate expected penalty on misreports. Third, we translate near-truthful reporting into a primal efficiency guarantee by showing that, except in rare events where two agents have close dual-adjusted values, the allocation induced by the reported dual-adjusted values coincides with the dual-adjusted optimal allocation.

\paragraph{Step 1: A nearsighted benchmark for the direct impact.}
To show that dual-adjusted pricing resolves the direct impact of misreports, we consider a hypothetical nearsighted game. Instead of optimizing their expected $\gamma$-discounted utility over the remainder of the game, agents optimize their current-epoch discounted utility (and ignore any indirect impact of their reports on future-epoch dual variables). Formally, for an agent $i\in [K]$ in round $t\in \mE_\ell$, we compare the following two objectives when deciding their round-$t$ report $u_t$:
\begin{equation}\begin{aligned}
\text{farsighted:} &\max_{u_t\in [0,1]} \mathbb{E}\Bigg[\sum_{\tau \ge t} \gamma^{\tau} (v_{\tau,i} - p_{\tau,i}) \1[i_\tau = i]\Bigg]; \\
\text{nearsighted:} &\max_{u_t\in [0,1]} \mathbb{E}\Bigg[\sum_{\substack{\tau\ge t\\ \tau \in \mE_\ell}} \gamma^{\tau} (v_{\tau,i} - p_{\tau,i}) \1[i_\tau = i]\Bigg].\label{eq:two agent models}
\end{aligned}\end{equation}
The nearsighted benchmark is used only as an intermediate device for our analysis: It removes the indirect effect through future dual updates and allows us to focus on the current-epoch allocation incentives.

Fix an epoch $\ell\in [L]$ and a dual variable $\bm \lambda_\ell$. Since the dual variable is held fixed within the epoch, the nearsighted benchmark induces an $\lvert \mE_\ell\rvert$-round resource allocation game. In this game, each agent maximizes their $\gamma$-discounted value-minus-payment, whereas the planner seeks to maximize the undiscounted dual-adjusted social welfare.
Thus, in addition to the usual tension between self-interested agents and a welfare-maximizing planner, there is also a mismatch induced by the dual adjustment and the agents' discounting. The dual-adjusted allocation and payment rule in \Cref{line:primal allocation} is designed precisely to address this mismatch.

\begin{theorem}[Direct Impact; \Cref{thm:IntraEpoch guarantee}]\label{thm:informal IntraEpoch}
For any epoch $\ell\in [L]$ and any fixed dual variable $\bm \lambda_\ell$, in the nearsighted game defined in \Cref{eq:two agent models}, the allocation and payment rule in \Cref{line:primal allocation} induces a truthful PBE.
\end{theorem}

\begin{proof}[Proof Sketch]
When agents are myopic and focus only on current-round utility, truth-telling is a dominant strategy (i.e., optimal regardless of the other agents' reports). Moreover, within an epoch, the allocation and payment rule in \Cref{line:primal allocation} does not depend on the history of reports or allocations in that epoch; the only history dependence is through the fixed dual variable $\bm \lambda_\ell$, which was determined before the epoch began. Therefore, if all other agents report truthfully---which is also history-independent---in the nearsighted game, agent $i$ also finds truthful reporting optimal in every round of the epoch. This establishes the truthful PBE.
\end{proof}

\paragraph{Step 2: Bounding large misreports by farsighted agents.}
We now return to the original farsighted setting, where $\gamma$-impatient agents (\Cref{def:impatient}) optimize over the remainder of the game. By \Cref{thm:informal IntraEpoch}, dual-adjusted payments remove the direct current-round benefit of misreporting. It remains to control the indirect incentive to misreport in order to influence future dual variables.

The two remaining components of \algname address this indirect incentive. Epoch-based lazy updates delay the effect of any misreport on future dual variables, while randomized exploration creates an immediate expected penalty when reports deviate from values. We compare these two forces below.

Consider any round $t\in \mE_\ell$ and any agent $i\in [K]$. Because of randomized exploration, with probability $\frac{1}{K\lvert \mE_\ell\rvert}$, the planner offers agent $i$ a random price $p\sim \Unif[0,1]$. If the agent over-reports, they risk paying a price above their true value. If the agent under-reports, they risk losing a profitable allocation. Integrating over the random price lowers bound the expected immediate penalty from a misreport: $\E[\text{Penalty}_{t,i}] \ge
\frac{(u_{t,i}-v_{t,i})^2}{2K\lvert \mE_\ell\rvert}$.

On the other hand, because the dual variable is updated only at epoch boundaries, a misreport in round $t\in \mE_\ell$ can affect allocations only in future epochs. Hence, any potential benefit is delayed and discounted. In particular, the gain from misreporting in round $t$ is upper bounded by $\E[\text{Gain}_{t,i}]
\le
\frac{\gamma^{\text{remaining rounds in }\mE
_{\ell}}}{1-\gamma}$.

Thus, the immediate penalty is quadratic in the magnitude of the misreport, whereas the future gain decays exponentially with the number of remaining rounds in the epoch. Comparing these two quantities shows that a rational agent will not choose a ``large misreport,'' in the sense that $\lvert u_{t,i}-v_{t,i}\rvert \ge \frac{1}{\lvert \mE_\ell\rvert}$, except possibly when the number of remaining rounds in the epoch is of order $\O(\log \lvert \mE_\ell\rvert)=\Otil(1)$. Formalizing this comparison, as in \Cref{lem:large misreport}, yields the incentive-awareness guarantee in \Cref{thm:informal IAPD main theorem}.

\paragraph{Step 3: From near-truthfulness to primal efficiency.}
It remains to show that near-truthful reporting implies rare misallocations. Consider a round $t\in \mE_\ell$ where no agent makes a large misreport, i.e., $\lvert u_{t,i}-v_{t,i}\rvert \le \frac{1}{\lvert \mE_\ell\rvert}$ for all $i\in [K]$.
Suppose that a misallocation occurs in the sense that $i_t\ne \tilde i_t^\ast$ (defined in \Cref{eq:tilde i definition main}). We know that \textit{(i)} from the definition of $\tilde i_t^\ast$ in \Cref{eq:tilde i definition main}, the dual-adjusted value of $\tilde i_t^\ast$ is at least that of $i_t$, i.e., $\tilde v_{t,\tilde i_t^\ast}\ge \tilde v_{t,i_t}$; and \textit{(ii)} from the definition of $i_t$ in \Cref{line:primal allocation}, the dual-adjusted report of $i_t$ is at least that of $\tilde i_t^\ast$, i.e., $\tilde u_{t,i_t}\ge \tilde u_{t,\tilde i_t^\ast}$. Since reports are within $\frac{1}{\lvert \mE_\ell\rvert}$ of true values for every agent, we must have $\lvert \tilde v_{t,i_t}-\tilde v_{t,\tilde i_t^\ast}\rvert
\le \frac{2}{\lvert \mE_\ell\rvert}$.
Therefore, a misallocation (when there is no large misreport) can only occur when the top two dual-adjusted values are close to each other. But this happens rarely under \Cref{assump:smooth consumptions}: Due to smooth consumption distributions,
\begin{equation}\label{eq:rank flip probability single-unit}
\Pr\left \{\min_{1\le i<j\le K}\lvert \tilde v_{t,i}-\tilde v_{t,j}\rvert\le \frac{2}{\lvert \mE_\ell\rvert}\right \}\le \frac{2\epsilon_c K^2}{\lvert \mE_\ell\rvert}.
\end{equation}

Azuma-Hoeffding inequality implies w.h.p. that the event in \Cref{eq:rank flip probability single-unit} occurs at most $\O(K^2\log \lvert \mE_\ell\rvert)=\Otil(K^2)$ times in epoch $\mE_\ell$ (\Cref{lem:misallocation}). Thus, aside from the large-misreport rounds controlled in Step 2 (and the exploration rounds), \algname allocation coincides with the dual-adjusted optimum $\tilde i_t^\ast$ in all but $\Otil(K^2)$ rounds.

Combining the bounds on large-misreport rounds, small-gap misallocation rounds, and exploration rounds, we conclude that in each epoch $\mE_\ell$, the number of rounds in which $i_t\ne \tilde i_t^\ast$ is at most $\Otil(K^2)$. Summing across the $L$ epochs yields
$
\E[\primalreg]=\Otil(LK^2).
$
This proves the primal-efficiency component of \Cref{thm:informal IAPD main theorem}.

To conclude, \Cref{thm:informal IAPD main theorem} shows that the combination of dual-adjusted payments, epoch-based lazy updates, and randomized exploration ensures two properties. First, there exists a PBE $\bm \pi$ under which agents make only $\Otil(1)$ large misreports per epoch. Second, the allocation chosen by \algname differs from the dual-adjusted optimum $\tilde i_t^\ast$ in only $\Otil(K^2)$ rounds per epoch. We now turn to the dual-update subroutine $\mA$ in \algname.

\section{Dual-Update Subroutines: Learning via Predictability and Fixed Points}\label{sec:mechanism dual}
In this section, we present two dual-update subroutines that one can use as the $\mA$ in \Cref{line:call subroutine} of \Cref{alg:mech}. Both subroutines build upon the reduction from \dualreg to online learning regret, but adopt different online learning methods.
Specifically, since \algname (\Cref{alg:mech}) ensures incentive-awareness and primal efficiency (recall \Cref{thm:informal IAPD main theorem}), we now focus on feasibility---ensuring $\frac 1T \sum_{t=1}^T \bm c_{t,i_t}\le \bm \rho$---and dual efficiency---minimizing the gap between $\tilde i_t^\ast$ (\Cref{eq:tilde i definition main}) and the offline optimum $i_t^\ast$ (\Cref{eq:offline optimal benchmark}).
Following standard online resource allocation analysis \citep{devanur2023online}, both objectives---as captured by \dualreg in \Cref{eq:regret decomposition informal}---reduce to the \emph{online learning regret} of $\bm \lambda_1,\bm \lambda_2,\ldots,\bm \lambda_L$ w.r.t. some hindsight optimal dual variable $\bm \lambda^\ast$:
\begin{theorem}[\dualreg as Online Learning Regret]\label{lem:DualUpd to regret informal}
Fix an epoching scheme $[T]=\bigcup_{\ell=1}^L \mE_\ell$.
Let $\bm \Lambda=\prod_{j=1}^d[0,\rho_j^{-1}]:=\{\bm \lambda\in \mathbb R_{\ge 0}^d\mid \lambda_j\le \rho_j^{-1},\forall j\in [d]\}$ be a box region in $\mathbb R_{\ge 0}^d$. Let $\mT_v:=\min\{t\in [T]\mid \sum_{{\tau}=1}^{t} \bm c_{{\tau},i_{\tau}} + \bm 1\not \le T\bm \rho\}\cup \{T+1\}$ be the (random) first round where the feasibility can possibly be violated. Let $\mathcal L_v$ be the (random) epoch that $\mT_v$ belongs to. Then the \dualreg defined in \Cref{eq:regret decomposition informal} has an expectation bounded by:
\begin{equation}\begin{aligned}\label{eq:online learning regret informal}
\E[\dualreg] = \E\left[\sum_{t=1}^T v_{t,i_t^\ast} - \sum_{t=1}^{\mT_v} v_{t,\tilde i_t^\ast} \right]\lesssim \E\left[\sup_{\bm \lambda^\ast \in \bm \Lambda} \sum_{\ell=1}^{\mathcal L_v} \sum_{t\in \mE_\ell} (\bm \rho - \bm c_{t,i_t})^\trans (\bm \lambda_\ell - \bm \lambda^\ast) \right]:=\mR_L^{\text{OL}},
\end{aligned}\end{equation}
where we recall that the dual-adjusted optimum $\tilde i_t^\ast  = \argmax_i (v_{t,i}-\bm \lambda_t^\trans \bm c_{t,i})$ is defined in \Cref{eq:tilde i definition main}, and the offline optimum $\{i_t^\ast\}_{t\in [T]}=\argmax_{i_1,i_2,\ldots,i_T} \sum_{t=1}^T v_{t,i_t}$ s.t. $\frac 1T \sum_{t=1}^T \bm c_{t,i_t}\le \bm \rho$ is defined in \Cref{eq:offline optimal benchmark}.
\end{theorem}
The proof mostly follows the standard primal-dual analysis in online resource allocation, except for epoching and strategic manipulation that are unique to our setting. We defer the formal analysis to \Cref{sec:DualUpd to regret}.

The RHS of \Cref{eq:online learning regret informal}, namely $\mR_L^{\text{OL}}$, is called the \emph{online learning regret}. This is because when defining the ``loss function'' for epoch $\ell$---which is a mapping from any dual $\bm \lambda\in \bm \Lambda$ to a real number---as a linear function
\begin{equation}\label{eq:online learning loss informal}
F_\ell(\bm{\lambda}) := \left(\sum_{t \in \mE_\ell} (\bm{\rho} - \bm{c}_{t, i_t})\right)^\trans \bm{\lambda},\quad \forall \ell \in [L],\bm \lambda\in \bm \Lambda,
\end{equation}
then we have $\mR_L^{\text{OL}}=\E[\sup_{\bm \lambda^\ast\in \bm \Lambda}\sum_{\ell=1}^{\mathcal L_v} (F_\ell(\bm \lambda_\ell)-F_\ell(\bm \lambda^\ast))]$, i.e., the expected difference between the total loss suffered by our duals $\{\bm \lambda_\ell\}_{\ell\in [L]}$ and a fixed hindsight optimal dual $\bm \lambda^\ast\in \bm \Lambda$.
The key challenge of minimizing $\mR_L^{\text{OL}}$ is that the loss function $F_\ell$ in \Cref{eq:online learning loss informal} is \emph{unknown} when deciding the dual $\bm \lambda_\ell$; however, the planner is still required to minimize $\sum_{\ell\in [L]} F_\ell(\bm \lambda_\ell)$. We demonstrate how this is made possible in subsequent sections.

\subsection{Baseline: Follow-the-Regularized-Leader and Price of Incentives}

\begin{algorithm}[t!]
\caption{Dual-Update Subroutine using Follow-the-Regularized-Leader (\ftrl)}\label{alg:lambda FTRL}
\begin{algorithmic}[1]
\Require Current epoch number $\ell$, learning rate $\eta_\ell>0$, a strongly convex regularizer $\Psi\colon \mathbb R_{\ge 0}^d\to \mathbb R$.
\State Define the dual decision region $\bm{\Lambda} := \prod_{j=1}^d [0,\rho_j^{-1}]$ according to \Cref{lem:DualUpd to regret informal}.\label{line:define dual region}
\State Solve the following optimization problem over $\bm \lambda_\ell\in \bm \Lambda$ and return $\bm \lambda_\ell$:
\algeq[eq:lambda FTRL]{$\displaystyle \bm{\lambda}_\ell = \argmin_{\bm{\lambda} \in \bm{\Lambda}} \Biggl(\sum_{\ell' < \ell} \sum_{\tau \in \mE_{\ell'}} (\bm{\rho} - \bm{c}_{\tau, i_\tau})^\trans \bm{\lambda} + \frac{1}{\eta_\ell} \Psi(\bm{\lambda})\Biggr).$}
\end{algorithmic}
\end{algorithm}

We present the first dual-update subroutine $\mA$, which applies a classical online learning algorithm---\textbf{F}ollow-\textbf{t}he-\textbf{R}egularized-\textbf{L}eader algorithm (\ftrl, \citealt{abernethy2008competing})---to decide the dual variables according to \Cref{eq:online learning loss informal}. Intuitively, \ftrl balances two objectives: fitting the historical data (by minimizing cumulative past losses) and maintaining stability in subsequent decisions (via a strictly convex regularization function $\Psi$). The pseudo-code of \ftrl is in \Cref{alg:lambda FTRL}, and it  enjoys the following online learning regret guarantee:
\begin{theorem}[\ftrl Subroutine]\label{thm:DualUpd guarantee FTRL informal}
For any epoching scheme $[T]=\bigcup_{\ell=1}^L \mE_\ell$, when setting $\Psi(\bm \lambda)=\frac 12 \lVert \bm \lambda\rVert_2^2$ as the Euclidean norm and $\eta_\ell=\Theta\Bigl(1\Big /\sqrt{\sum_{\ell'=1}^\ell \lvert \mE_{\ell'}\rvert^2} \Bigr )$, \ftrl in \Cref{alg:lambda FTRL} ensures $\mR_L^{\text{OL}}=\O\Bigl( \sqrt{\sum_{\ell=1}^L \lvert \mE_\ell\rvert^2} \Bigr )$.
\end{theorem}
\begin{proof}[Proof Sketch]
The formal proof can be found in \Cref{sec:DualUpd guarantee FTRL}.
Roughly, the online learning regret of \ftrl depends on the sum of (squared) gradient norms of loss functions \citep[Corollary 7.7]{orabona2019modern}:
\begin{equation}\label{eq:FTRL online learning regret}
\mR_L^{\text{OL}}\le \O(\eta_L^{-1})+\sum_{\ell=1}^L \eta_\ell \E\left [\lVert \nabla F_\ell\rVert_2^2\right ].
\end{equation}
(In the most general form, the gradient $\nabla F_\ell$ is evaluated at the chosen decision $\bm \lambda_\ell$. Nevertheless, our loss function $F_\ell$---defined in \Cref{eq:online learning loss informal}---is linear, whose gradient $\nabla F_\ell = \sum_{t \in \mE_\ell} (\bm \rho - \bm c_{t,i_t})$ is the same everywhere.)

In the worst case, all consumptions $\{\bm c_{t,i_t}\}_{t\in \mE_\ell}$ accumulate in similar directions, thus resulting in a gradient norm of order $\lVert \nabla F_\ell\rVert_2^2=\Theta(|\mE_\ell|^2)$.
The configuration of $\eta_\ell$ ensures that the first term in \Cref{eq:FTRL online learning regret} is $\O(\eta_L^{-1})=\O\Bigl (\sqrt{\sum_{\ell=1}^L \lvert \mE_\ell\rvert^2}\Bigr )$, and the second term is bounded as $\sum_{\ell=1}^L \Bigl (\lvert \mE_\ell\rvert^2\Big / \sqrt{\sum_{\ell'=1}^\ell \lvert \mE_{\ell'}\rvert^2}\Bigr )=\O\Bigl (\sqrt{\sum_{\ell=1}^L \lvert \mE_\ell\rvert^2}\Bigr )$.
\end{proof}

Using \ftrl as the dual-update subroutine $\mA$ in the \algname framework (\Cref{line:call subroutine} in \Cref{alg:mech}), we obtain a mechanism \mech. It ensures efficiency, feasibility, and incentive-awareness at the same time:
\begin{corollary}[\mech Mechanism]\label{thm:main theorem FTRL}
When setting $L=T^{2/3}$ and  $\lvert \mE_1\rvert=\lvert \mE_2\rvert=\cdots \lvert \mE_L\rvert = T^{1/3}$, the
\mech mechanism---plugging \ftrl (\Cref{alg:lambda FTRL}) into \algname (\Cref{alg:mech}) as the dual-update subroutine $\mA$ and configuring according to \Cref{thm:DualUpd guarantee FTRL informal}---induces a PBE $\bm \pi$ (\Cref{def:PBE}) such that:
\begin{enumerate}
\item \textbf{Efficiency.} The social welfare regret of $\bm \pi$ under \mech is sublinear in $T$, hence the planner's allocations converge to the offline optimum. Formally, we have $\mR_T(\bm \pi, \mech) = \Otil\big ((K^2+\sqrt d) T^{2/3}\big )$;
\item \textbf{Feasibility.} All constraints are satisfied with probability 1, i.e., $\mB_T(\bm \pi, \mech) = 0$; and
\item \textbf{Incentive-Awareness.} The expected number of rounds where $\max_{i\in [K]}\lvert u_{t,i}-v_{t,i}\rvert\ge T^{-1/3}$ is $\Otil(K T^{2/3})$.
\end{enumerate}
\end{corollary}

\label{sec:price of incentives}
The formal proof is in \Cref{sec:main theorem FTRL}. We will discuss the computational complexities of our mechanisms in \Cref{sec:numerical illustration setup}.
While \mech is the first to attain efficiency, feasibility, and incentive-awareness at the same time, its efficiency bound of $\mR_T=\Otil(T^{2/3})$ stands in stark contrast to the $\Theta(\sqrt{T})$ minimax rate achievable in non-strategic resource allocation \citep{li2023simple,arlotto2019uniformly}.
This sub-optimality gap is \emph{not} an artifact of our analysis, but a fundamental consequence of the epoch-based lazy updates.

Indeed, we identify an $\Omega(T^{2/3})$ ``price of incentives'' barrier:
As shown in \Cref{fig:lazy vs no lazy}, the freeze of dual variables for prolonged epochs---necessary to tame strategic behavior---fundamentally changes the hardness of learning. We present the following hardness result, adopted from \citet[Proposition 7]{chen2020minimax}:
\begin{proposition}[``Price of Incentives'']
Fix the number of rounds $T$ and number of epochs $L$. There is a family of linear loss functions, $f_t(\bm \lambda)=\bm g_t^\trans \bm \lambda$ where $\bm g_1,\bm g_2,\ldots,\bm g_T\in [0,1]^d$ (cf. the $F_\ell(\bm \lambda)$ in \Cref{eq:online learning loss informal}), such that any algorithm switching decision variables for $L$ times must suffer $\Omega(T/\sqrt{2L})$ online learning regret.
\end{proposition}
Without additional properties, for \emph{any} dual-update subroutine, the social welfare regret $\mR_T=\E[\primalreg+\dualreg]$ must be $\Omega\left(\min_{L} \bigl(L+(T/\sqrt{2L})\bigr)\right) = \Omega(T^{2/3})$.
Since epoching is essential to deter agents' incentives, we call this $\Omega(T^{2/3})$ barrier a ``price of incentives,'' which \mech mechanism cannot break.

\begin{figure}[t!]
\centering

\begin{tikzpicture}[
    transform shape,
    action/.style={rectangle, fill=black!5!white, thick, align=center, minimum width=2cm, minimum height=0.7cm},
    loss/.style={rectangle, fill=black!15!white, thick, align=center, minimum width=2cm, minimum height=0.7cm},
    pred/.style={diamond, aspect=1.5, draw=green!60!black, fill=green!20, thick, minimum size=7mm},
    flow/.style={-latex, thick, gray!80!black},
    limitation/.style={color=red!80!black, font=\small\itshape, align=center, text width=3cm}
]
\node[loss] (L11) at (-5, 0) {$\bm c_{{t-1},i_{t-1}}$};
\node[action] (A11) at (-2.5, 0) {$\bm \lambda_t$};
\draw[flow] (L11) -- (A11);
\node[loss] (L12) at (0,0) {$\bm c_{t,i_t}$};
\draw[flow] (A11) -- (L12);
\node[action] (A12) at (2.5,0) {$\bm \lambda_{t+1}$};
\draw[flow] (L12) -- (A12);
\node[text width=5cm, align=right] (T1) at (5.5,0) {\small\textbf{without Epoching: $\Theta(T^{1/2})$}};

\draw[dashed, gray] (1.25, 0.5) -- (1.25, -1.5);
\node[above] at (1.25, 0.6) {\small{\color{black!50!white}Start of Epoch $\mE_{\ell+1}$}};
\node[text width=5cm, align=right] (T2) at (5.5,-1) {\small\textbf{with Epoching: $\Omega(T^{2/3})$}};

\node[loss, text width=6.7cm] (L2) at (-2.5, -1) {$\bm c_{t-\lvert \mE_{\ell}\rvert,i_{t-\lvert \mE_{\ell}\rvert}},\ldots,\bm c_{t-1,i_{t-1}},\bm c_{t,i_t}$};
\node[action] (A2) at (2.5, -1) {$\bm \lambda_{\ell+1}$};
\draw[flow] (L2) -- (A2);
\end{tikzpicture}
\caption{``Price of Incentives'': Online Learning Regret with and without Epoching}
\label{fig:lazy vs no lazy}
\end{figure}

\subsection{Breaking the Barrier: Optimistic FTRL and Circular Dependency}\label{sec:preditability}
Fortunately, the $\Omega(T^{2/3})$ barrier only applies when loss functions are {arbitrary} linear functions, whereas our loss function $F_\ell$ (in \Cref{eq:online learning loss informal}) admits a rich structure: Its gradient, namely $\sum_{t\in \mE_\ell} (\bm \rho-\bm c_{t,i_t})$, is determined by the epoch-$\ell$ allocations $\{i_t\}_{t\in \mE_\ell}$.
Such allocations---thanks to the \algname guarantee in \Cref{thm:informal IAPD main theorem}---almost always equal the dual-adjusted optima $\tilde i_t^\ast=\argmax_i(v_{t,i}-\bm \lambda_\ell^\trans \bm c_{t,i})$, which are i.i.d. conditional on the epoch-$\ell$ dual variable $\bm \lambda_\ell$ (because both $\bm v_t\sim \mV$ and $\bm c_t\sim \mC$ are i.i.d. for all $t$).
We present the following observation:

\begin{proposition}[Loss Structure]\label{claim:almost iid}
Fix an epoch $\ell\in [L]$ and dual $\bm \lambda_\ell \in \bm \Lambda$. The gradient of the loss function $F_\ell$ in \Cref{eq:online learning loss informal}, namely $\nabla F_\ell=\sum_{t \in \mE_\ell} (\bm \rho - \bm c_{t,i_t})$, is $\Otil(K)$-close (in $L_2$-norm) to the sum of $\lvert \mE_\ell\rvert$ i.i.d. random variables $\sum_{t\in \mE_\ell}(\bm \rho-\bm c_{t,\tilde i_t^\ast})$, where  $\tilde i_t^\ast=\argmax_i(v_{t,i}-\bm \lambda_\ell^\trans \bm c_{t,i})$ is defined in \Cref{eq:tilde i definition main}. Formally, we have
\begin{equation*}
\left \lVert \nabla F_\ell - \sum_{t \in \mE_\ell} (\bm \rho - \bm c_{t,\tilde i_t^\ast})\right \rVert_2=\Otil(K).
\end{equation*}
\end{proposition}
\Cref{claim:almost iid} is a combination of \Cref{lem:stability term 1,lem:stability term 2,lem:stability term 3}, all presented in \Cref{sec:appendix DualUpd}.
Should we be able to predict this i.i.d. sum $\sum_{t\in \mE_\ell}(\bm \rho-\bm c_{t,\tilde i_t^\ast})$ in advance, we could boost the performance via \textbf{O}ptimistic \textbf{FTRL} (\oftrl, \citealt{rakhlin2013online}): Suppose that some predicted gradient $\tilde{\bm g}_\ell$ is available at the beginning of each epoch $\ell$, \oftrl ensures online learning regret \citep[Theorem 7.39]{orabona2019modern}:
\begin{equation}\label{eq:O-FTRL online learning regret}
\mR_L^{\text{OL}}\le \O(\eta_L^{-1})+\sum_{\ell=1}^L \eta_\ell \E\left [\lVert \nabla F_\ell-\tilde{\bm g}_\ell\rVert_2^2\right ].
\end{equation}
Compared to the $\E[\lVert \nabla F_\ell\rVert_2^2]=\Theta(\lvert \mE_\ell\rvert^2)$ term in \Cref{eq:FTRL online learning regret}, the $\E[\lVert \nabla F_\ell - \tilde{\bm g}_\ell \rVert_2^2]$ term in \Cref{eq:O-FTRL-FP online learning regret} can be much smaller: Let $\bm g_\ell:=\sum_{t\in \mE_\ell}(\bm \rho-\bm c_{t,\tilde i_t^\ast})$. Observe that $\E[\lVert \nabla F_\ell - \tilde{\bm g}_\ell \rVert_2^2]\lesssim \E\left [\left \lVert \nabla F_\ell - \bm g_\ell\right \rVert_2^2\right ]+\E\left [\lVert \bm g_\ell-\E[\bm g_\ell]\rVert_2^2\right ]+\lVert \bm g_\ell-\tilde{\bm g}_\ell\rVert_2^2$.
The first term is no more than $\Otil(K^2)$ due to \Cref{claim:almost iid}; the second term is of order $\O(\lvert \mE_\ell\rvert)$ due to the i.i.d. nature of $\{\tilde i_t^\ast\}_{t\in \mE_\ell}$; and the third term---as we soon justify---can also be of order $\Otil(\lvert \mE_\ell\rvert)$.
Consequently, $\E[\lVert \nabla F_\ell - \tilde{\bm g}_\ell \rVert_2^2]=\O(\lvert \mE_\ell\rvert)+\Otil(K^2)$. Properly setting $\{\eta_\ell\}_{\ell\in [L]}$ thus gives $\Otil_T(\sqrt T)$ regret.

The next step is crafting a prediction $\tilde{\bm g}_\ell$ that is close to $\bm g_\ell=\sum_{t\in \mE_\ell}(\bm \rho-\bm c_{t,\tilde i_t^\ast})$. Again thanks to \Cref{thm:informal IAPD main theorem}, agents' (most) historical reports are close to their true values. Therefore, suppose that the to-be-decided dual variable $\bm \lambda_\ell$ is \emph{known} (which is unrealistic but helps understand the construction), we can define $\tilde{\bm g}_\ell$ as:
\begin{itemize}
\item For each historical (i.e., in previous epochs) round $\tau$, let $\tilde i_\tau(\bm \lambda_\ell):=\argmax_i (u_{\tau,i}-\bm \lambda_\ell^\trans \bm c_{\tau,i})$ where $\bm u_{\tau},\bm c_{\tau}$ are the reports and consumptions in round $\tau$. We call $\tilde i_\tau(\bm \lambda_\ell)$ a \emph{historical counterfactual allocation}, because it is the allocation as-if we were using the new dual variable $\bm \lambda_\ell$ in the historical round $\tau$.
\item Since $\bm u_\tau\approx \bm v_{\tau}$, all $\tilde i_\tau(\bm \lambda_\ell)$'s are almost identically distributed as $\tilde i_t^\ast=\argmax_i (v_{t,i}-\bm \lambda_\ell^\trans \bm c_{t,i})$. Meanwhile, the induced consumptions $\bm c_{\tau,\tilde i_\tau(\bm \lambda_\ell)}$'s are also almost identically distributed as $\bm c_{t,\tilde i_t^\ast}$. Hence, averaging these \emph{historical counterfactual consumptions} gives a good prediction of the gradient $\nabla F_\ell\approx \sum_{t\in \mE_\ell} (\bm \rho-\bm c_{t,\tilde i_t^\ast})$.
\end{itemize}
We formalize this prediction---denoted by $\tilde{\bm g}_\ell(\bm \lambda_\ell)$---in \Cref{eq:lambda O-FTRL-FP}. It only remains to lift the assumption that $\bm \lambda_\ell$ is known before crafting $\tilde{\bm g}_\ell$: This is unrealistic because \oftrl precisely uses the predicted gradient $\tilde{\bm g}_\ell$ to decide the dual variable $\bm \lambda_\ell$, hence this induces a \emph{circular dependency} that \oftrl cannot accommodate.

To fix this, we propose an online learning algorithm called \oftrl with \textbf{F}ixed \textbf{P}oints (\oftrlfp). The pseudo-code is given in \Cref{alg:lambda O-FTRL-FP}. With a more technical exposition postponed to \Cref{sec:circular dependency}, on a high level, our \oftrlfp algorithm formulates each dual update---searching for the $(\bm \lambda_\ell, \tilde{\bm g}_\ell)$ pair---as a fixed-point problem to avoid the circular dependency. We first state the main guarantee of \Cref{alg:lambda O-FTRL-FP}:
\begin{theorem}[\oftrlfp Subroutine]\label{thm:DualUpd guarantee O-FTRL-FP informal}
For any epoching scheme $[T]=\bigcup_{\ell=1}^L \mE_\ell$, when setting $\Psi(\bm \lambda)=\frac 12 \lVert \bm \lambda\rVert_2^2$ and $\eta_\ell=\Theta\Bigl(1\Big /\sqrt{\sum_{\ell'=1}^\ell \lvert \mE_{\ell'}\rvert} \Bigr )$, \oftrlfp in \Cref{alg:lambda O-FTRL-FP} ensures $\mR_L^{\text{OL}}=\O\Bigl(\sqrt{\sum_{\ell=1}^L \lvert \mE_{\ell}\rvert} \Bigr )=\O(\sqrt T)$.
\end{theorem}
With the formal proof deferred to \Cref{sec:appendix DualUpd}, we sketch the key technical steps in \Cref{sec:sketch O-FTRL-FP}.
Plugging the \oftrlfp subroutine (\Cref{alg:lambda O-FTRL-FP}) into \algname framework (\Cref{alg:mech}), we obtain \mechO mechanism. It enjoys a much sharper regret bound compared to \mech (cf. \Cref{thm:main theorem FTRL}):
\begin{corollary}[\mechO Mechanism]\label{thm:main theorem O-FTRL-FP} When setting $L=\sqrt T$ and  $\lvert \mE_1\rvert=\lvert \mE_2\rvert=\cdots \lvert \mE_L\rvert = \sqrt T$, the
\mechO mechanism---obtained by plugging \oftrlfp (\Cref{alg:lambda O-FTRL-FP}) into \algname (\Cref{alg:mech}) and configuring according to \Cref{thm:DualUpd guarantee O-FTRL-FP informal}---induces a PBE $\bm \pi$ (\Cref{def:PBE}) such that:
\begin{enumerate}
\item \textbf{Efficiency:} The regret bound is near-optimal in $T$, i.e., $\mR_T(\bm \pi, \mechO) = \Otil\big ((d+K^2) \sqrt T\big )$;
\item \textbf{Feasibility.} All constraints are satisfied with probability 1, i.e., $\mB_T(\bm \pi, \mechO) = 0$; and
\item \textbf{Incentive-Awareness.} The number of rounds where $\max_{i\in [K]}\lvert u_{t,i}-v_{t,i}\rvert\ge 1/\sqrt t$ is at most $\Otil(K \sqrt T)$.
\end{enumerate}
\end{corollary}

The $\Otil(\sqrt T)$ regret attained by \mechO is near-optimal (up to polylog) in $T$: Even in non-strategic resource allocation, an $\Omega(\sqrt T)$ lower bound is unavoidable \citep[Lemma 1]{arlotto2019uniformly}. Therefore, our result in \Cref{thm:main theorem O-FTRL-FP} successfully bypasses the $\Omega(T^{2/3})$ price of incentives: By exploiting the predictability due to incentive-awareness and new online learning algorithms, one can recover the near-optimal $\Otil(\sqrt{T})$ regret in the presence of strategic agents---thus virtually paying no price of incentives.

\begin{algorithm}[t!]
\caption{Dual-Update Subroutine using Optimistic FTRL with Fixed Points (\oftrlfp)}\label{alg:lambda O-FTRL-FP}
\begin{algorithmic}[1]
\Require Epoch $\ell$, learning rate $\eta_\ell>0$, regularizer $\Psi\colon \mathbb R_{\ge 0}^d\to \mathbb R$, and historical reports and consumptions.
\State $\bm \lambda_1=\bm 0$. For $\ell\ge 2$, solve the following (approximate) fixed point for $(\bm \lambda_\ell,\tilde{\bm g}_\ell)\in \bm \Lambda\times \mathbb R^d$ and return $\bm \lambda_\ell$:
\algeq[eq:lambda O-FTRL-FP]{$\displaystyle \begin{aligned}
\bm{\lambda}_\ell
\approx &\argmin_{\bm{\lambda} \in \bm{\Lambda}}
\Biggl (
\sum_{\ell' < \ell} \sum_{\tau \in \mE_{\ell'}}
(\bm{\rho} - \bm{c}_{\tau, i_\tau})^\trans \bm{\lambda}
+ \tilde{\bm g}_\ell(\bm{\lambda}_\ell)^\trans \bm{\lambda}
+ \frac{1}{\eta_\ell} \Psi(\bm{\lambda})
\Biggr),\\
\text{where }
&\tilde{\bm g}_\ell(\bm{\lambda}_\ell)
:=
\frac{\lvert \mE_\ell\rvert}
{\sum_{\ell'<\ell} \lvert \mE_{\ell'}\rvert}
\sum_{\ell'<\ell}
\sum_{\tau \in \mE_{\ell'}}
\left(
\bm \rho-\bm c_{\tau,\tilde i_\tau(\bm \lambda_\ell)}
\right),\\
&\tilde i_\tau(\bm \lambda_\ell)
:=
\argmax_{i \in \{0\} \cup [K]}
\left(
u_{\tau,i} - \bm \lambda_\ell^\trans \bm{c}_{\tau,i}
\right),
\quad
\forall \ell'<\ell,\ \tau\in \mE_{\ell'}.\end{aligned}$}
\end{algorithmic}
\end{algorithm}

\subsection{Tackling Endogeneity and Discontinuity: Optimistic FTRL with Fixed Points}\label{sec:circular dependency}
We first detail our \oftrlfp algorithm, which tackles \emph{two} key challenges: endogeneity and discontinuity.

\subsubsection{Endogeneity.}
As sketched in \Cref{sec:preditability}, we face a circular dependency between deciding $\bm \lambda_\ell$ and predicting $\tilde{\bm g}_\ell$. We call this challenge \emph{endogeneity} of predictions. We highlight that this issue is not specifically due to our prediction procedure, which calculates historical counterfactual allocations and consumptions:
\begin{enumerate}
\item \textbf{Prediction requires Allocation.} In order to predict $\nabla F_\ell\approx \sum_{t\in \mE_\ell} (\bm \rho - \bm c_{t, \tilde i_t^\ast})$ accurately (see \Cref{claim:almost iid}), we need to approximate the distribution of dual-adjusted allocations $\tilde i_t^\ast:=\argmax_i (v_{t,i}-\bm \lambda_\ell^\trans \bm c_{t,i})$, $\forall t\in \mE_\ell$.
\item \textbf{Allocation requires Decision.} By definition, the allocations $\tilde i_t^\ast$ depend on the to-be-decided dual $\bm \lambda_\ell$.
\item \textbf{Decision requires Prediction.} However, the idea of \oftrl is exactly to decide $\bm \lambda_\ell$ as the minimizer of a function that includes the prediction $\tilde{\bm g}_\ell$. Thus, in \oftrl, deciding $\bm \lambda_\ell$ requires a good $\tilde{\bm g}_\ell$ in advance.
\end{enumerate}

This endogeneity challenge is unique to our setting: In standard online learning, predictions are typically exogenous and only depend on the history, as in the case faced by \citet{rakhlin2013online}.
Instead of first predicting $\tilde{\bm g}_\ell$ and then deciding $\bm \lambda_\ell$ based on $\tilde{\bm g}_\ell$, we find a fixed point such that $(\bm \lambda_\ell,\tilde{\bm g}_\ell)$ is self-consistent: The dual $\bm \lambda_\ell$ minimizes the regularized cumulative historical loss plus the predicted loss $\tilde{\bm g}_\ell(\bm \lambda_\ell)^\trans \bm \lambda_\ell$ (cf. \Cref{eq:lambda FTRL}), while the prediction $\tilde{\bm g}_\ell(\bm \lambda_\ell)$ is the historical counterfactual consumption that would be induced by the dual $\bm \lambda_\ell$. This allows the prediction $\tilde{\bm g}_\ell$ to be a function of the decision $\bm \lambda_\ell$, but yet another issue arises.

\subsubsection{Discontinuity.}
Since $\tilde{\bm g}_\ell(\bm \lambda_\ell)$ averages historical counterfactual consumptions---which involves discrete argmax operations---it is discontinuous.
Therefore, \Cref{eq:lambda O-FTRL-FP} may not admit an exact fixed point; instead, we can only aim for an approximate solution.
As another key feature of our \oftrlfp algorithm, we prove that the following weaker condition on $\tilde{\bm g}_\ell(\bm \lambda_\ell)$ suffices to ensure approximate fixed points, which in turn give an online learning regret bound (see \Cref{eq:O-FTRL-FP online learning regret}) similar to that of \oftrl (see \Cref{eq:O-FTRL online learning regret}).
\begin{definition}[$(\epsilon, L)$-Gradient-Stability]\label{def:gradient stability}
A predicted gradient function $\tilde{\bm g}_\ell(\bm \lambda_\ell)$ is said to be $(\epsilon_\ell, L_\ell)$-gradient-stable, if for any two $\bm \lambda_\ell^1,\bm \lambda_\ell^2 \in \bm \Lambda$ with $\lVert \bm \lambda_\ell^1-\bm \lambda_\ell^2 \rVert_2 < \epsilon_r$, we always have $\lVert \tilde{\bm g}_\ell(\bm \lambda_\ell^1)-\tilde{\bm g}_\ell(\bm \lambda_\ell^2)\rVert_2\le L_\ell$.
\end{definition}

Once \Cref{def:gradient stability} holds, we prove that our \oftrlfp ensures the following online learning regret (a generalized version of \oftrlfp, as well as its performance guarantee, is in \Cref{sec:O-FTRL lemma}; we expect it to be of independent interest to online learning problems with endogenous and discontinuous predictions):
\begin{equation}\begin{aligned}\label{eq:O-FTRL-FP online learning regret}
\mR_L^{\text{OL}}&\le \O(\eta_L^{-1})+\sum_{\ell=1}^L \eta_\ell \E\left [\lVert \nabla F_\ell-\tilde{\bm g}_\ell(\bm \lambda_\ell)\rVert_2^2\right ]+\sum_{\ell=1}^L \O\bigl (\eta_\ell L_\ell+\eta_\ell^2 L_\ell^2\bigr ).
\end{aligned}\end{equation}
Compared to the \oftrl bound in \Cref{eq:O-FTRL online learning regret}, the regret overhead of \oftrlfp---due to the endogeneity and discontinuity in predictions---is $\O(\eta_\ell L_\ell+\eta_\ell^2 L_\ell^2)$. In our case, this term is also of order $\Otil(\sqrt T)$.
Another drawback of \oftrlfp, however, is the computational cost: Both \ftrl and \oftrl solves convex optimization problems, while our \oftrlfp requires an approximate fixed-point oracle. While a computationally efficient implementation remains open, we present a practical approximation in \Cref{sec:computational cost of O-FTRL-FP}.

\subsection{Proof Sketch of \Cref{thm:DualUpd guarantee O-FTRL-FP informal}}\label{sec:sketch O-FTRL-FP}
\paragraph{Step 1: Ensuring gradient stability.} We first show the gradient estimator $\tilde{\bm g}_\ell(\bm \lambda_\ell)$ in \Cref{eq:lambda O-FTRL-FP} satisfies the $(\epsilon_\ell, L_\ell)$-gradient stability condition in \Cref{def:gradient stability}, which is non-trivial due to the discrete argmax operations therein.
In \Cref{lem:approximate continuity of predictions}, we bridge this gap by exploiting the \textit{smooth consumption} property in \Cref{assump:smooth consumptions}: We prove that for sufficiently small perturbations around $\bm \lambda_\ell$, the probability of having a different historical counterfactual allocation $\tilde i_\tau(\bm \lambda_\ell)$ is small, thereby bounding the variation in the predicted gradient $\tilde{\bm g}_\ell(\bm \lambda_\ell)$. Consequently, the $\O(\eta_\ell L_\ell+\eta_\ell^2 L_\ell^2)$ term in \Cref{eq:O-FTRL-FP online learning regret}---capturing the discontinuity of $\tilde{\bm g}_\ell(\bm \lambda_\ell)$---is small.

\paragraph{Step 2: Resolving statistical dependency.} We now bound the prediction error $\E[\lVert \nabla F_\ell - \tilde{\bm g}_\ell(\bm \lambda_\ell) \rVert_2^2]$ in \Cref{eq:O-FTRL-FP online learning regret}.
Unlike statistical learning where samples are i.i.d., our ``samples''---historical reports---are not: The dual $\bm \lambda_\ell$ depends on the historical reports, which were generated by agents reacting to previous duals.
This breaks standard concentration inequalities. To overcome this, we develop a \textit{uniform convergence} in \Cref{lem:stability term 2}: We show that, under any possible dual $\bm \lambda \in \bm \Lambda$, the empirical average of historical counterfactual consumptions converge to its expectation uniformly, ensuring a small prediction error even for history-dependent $\bm \lambda_\ell$'s.

\paragraph{Step 3: Handling bias from strategic corruption.} Finally, our prediction $\tilde{\bm g}_\ell$ relies on potentially untruthful historical reports $u_{\tau,i}$ (where $\tau$ is from a previous epoch $\ell'<\ell$). This introduces a strategic bias into both the current gradient estimation and the historical data.
We control these two biases in \Cref{lem:stability term 1,lem:stability term 3} by leveraging the incentive guarantees of \algname (\Cref{thm:informal InterEpoch}): Although agents are strategic, the mechanism ensures that large deviations from truth-telling are rare. We rigorously propagate these ``rare'' deviation bounds through the learning dynamics to show that their cumulative impact on the regret is negligible.

We now arrive at the desired conclusion: Instead of \ftrl whose regret scales as $\Otil\Big (\sqrt{\sum_{\ell=1}^L \lvert \mE_\ell\rvert^2}\Big )=\Otil(T^{2/3})$---due to the magnitude term in \Cref{eq:FTRL online learning regret}---by successfully exploiting optimistic predictions that are both endogenous and discontinuous, our \oftrlfp replaces the $\O(\lvert \mE_\ell\rvert^2)$ magnitude with an $\O(\lvert \mE_\ell\rvert)$ variance (cf. \Cref{eq:O-FTRL-FP online learning regret}). Balancing $\O(\eta_L^{-1})$ with $\sum_{\ell=1}^L \eta_\ell \O(\lvert \mE_\ell\rvert)$ thus gives $\mR_L^{\text{OL}}=\Otil\Bigl (\sqrt{\sum_{\ell=1}^L \lvert \mE_\ell\rvert}\Bigr)=\Otil(\sqrt T)$.

\section{Extension: Multi-Unit Resource Allocation with Multi-Unit Demands}\label{sec:extensions}\label{sec:multi-unit}
Thus far, we focus on the dynamic allocation of a single resource over time. In practical settings like cloud computing platforms or ride-shares, the planner often possesses multiple identical units of resources to fulfill the demand of many users. In this section, we show that our \algname framework and mechanisms seamlessly extend to such multi-unit resource allocation settings; we also allow agents to have multi-unit demands.

\paragraph{Setup.}
In each round $t\in [T]$, the planner has $N$ (which is fixed and known) identical resources available. Each agent $i\in [K]$ has a private {marginal value vector} $\bm v_{t,i}\in [0,1]^N$ in each round $t\in [T]$, such that being allocated $I_i\in [N]$ resources generates a total value of $\sum_{n=1}^{I_i} v_{t,i,n}$. The vector $\bm v_{t,i}$ is an i.i.d. sample from a fixed but unknown $N$-dimensional distribution $\mV_i$ such that the marginal values are diminishing (i.e., $v_{t,i,1}\ge \cdots\ge v_{t,i,N}$) almost surely \citep[see, e.g.,][]{engelbrecht1998multi}. Allocating one unit of resource to agent $i$ incurs a consumption of $\bm c_{t,i}\in [0,1]^d$, which is i.i.d. from a $\mC_i$ satisfying \Cref{assump:smooth consumptions}.

Therefore, in each round $t$, the planner decides an allocation vector $I_t$ from $\mathcal I:=\{I\in \mathbb N^K\mid \sum_{i=1}^K I_i\le N\}$. This generates a social value of $v_{t,I_t}:=\sum_{i=1}^K \sum_{n=1}^{I_{t,i}} v_{t,i,n}$, and the total consumption to the planner is $\bm c_{t,I_t}:=\sum_{i=1}^K I_{t,i}\cdot \bm c_{t,i}$.
We still write the planner's per-round budget as $\bm \rho$, but now $\bm\rho\in [0,N]^d$. The social welfare regret $\mR_T$, analogous to \Cref{eq:regret}, is $\mR_T(\bm \pi,\bm M):=\E_{\bm \pi,\bm M} \left [\sum_{t=1}^T (v_{t,I_t^\ast}-v_{t,I_t})\right ]$ where $\{I_t^\ast\}_{t\in [T]}$ is the offline optimal benchmark defined as $\{I_t^\ast\}_{t\in [T]}=\argmax_{I_1,I_2,\ldots,I_T\in \mathcal I} \sum_{t=1}^T v_{t,I_t}$ s.t. $\frac 1{T}\sum_{t=1}^T \bm c_{t,I_t} \le \bm \rho$.

\paragraph{Multi-Unit \algname Framework.}
In the single-unit \algname framework (\Cref{alg:mech}), we used the dual-adjusted second-price auction for allocation and payment. Specifically, in \Cref{line:primal allocation} of \Cref{alg:mech}, we set
$
i_t = \argmax_{i\in [K]} \tilde{u}_{t,i}$ and $p_{t,i_t} = \bm{\lambda}_t^\trans \bm{c}_{t,i_t} + \max_{j \ne i_t} \tilde{u}_{t,j}
$ (recall $\tilde u_{t,i}:=u_{t,i}-\bm \lambda_t^\trans \bm c_{t,i}$ is the dual-adjusted report).
In the multi-unit \algname framework, given the reports from agents, namely $\bm u_{t,i}=(u_{t,i,1},u_{t,i,2},\ldots,u_{t,i,N})$, we analogously decide the allocation $I_t\in \mathcal I$ to maximize the (reported) dual-adjusted value:
\begin{equation}\label{eq:allocation multi-unit multi-demand}
I_t=\argmax_{I\in \mathcal I} \left (u_{t,I}-\bm \lambda_t^\trans \bm c_{t,I}\right ),\quad u_{t,I}:=\sum_{i=1}^K \sum_{n=1}^{I_i} u_{t,i,n}.
\end{equation}
Since marginal values are diminishing, without loss of generality assume the reports satisfy $u_{t,i,n}\ge u_{t,i,n+1}$ for all $1\le n<N$ (if not, the planner can always redefine monotonized reports as $u_{t,i,n}':=\min_{n'=1}^n u_{t,i,n'}$).
\Cref{eq:allocation multi-unit multi-demand} is therefore exactly solvable by greedy: let $\tilde u_{t,i,n}:=u_{t,i,n}-\bm \lambda_t^\trans \bm c_{t,i}$, then we have $\tilde u_{t,i,n}\ge \tilde u_{t,i,n+1}$ for all $1\le n<N$. Therefore, the top-$N$ among all $\{\tilde u_{t,i,n}\}_{i\in [K],n\in [N]}$'s naturally correspond to a valid $I_t\in \mathcal I$.

The payment component is slightly more complicated.
With multi-demand, uniform-price payment rules (e.g., second- or $(N+1)$-th price auction) fail to be incentive-compatible because an agent can win the items and---at the same time---alter their payment \citep{ausubel2014demand}. We equip \algname with VCG pricing:
\begin{equation}\label{eq:payment multi-unit multi-demand}
\begin{aligned}p_{t,i}&=\bm \lambda_t^\trans (I_{t,i} \bm c_{t,i})+\max_{J\in \mathcal I,J_i=0} \left (u_{t,J}-\bm \lambda_t^\trans \bm c_{t,J}\right )-\sum_{j\ne i} \left (\sum_{n=1}^{I_{t,j}} u_{t,j,n}-\bm \lambda_t^\trans \bigl( I_{t,j}\cdot \bm c_{t,j}\bigr)\right ),~\forall i\in [K].\end{aligned}
\end{equation}
In words, the payment from agent $i$ consists of two parts: the first part translates dual-adjusted values to actual ones via consumptions, and the second part is others' optimal dual-adjusted welfare when agent $i$ is absent.
Similar to \Cref{eq:allocation multi-unit multi-demand}, the $\max_J$ term in \Cref{eq:payment multi-unit multi-demand} is also solvable by a greedy over all $\{\tilde u_{t,j,n}\}_{j\ne i,n\in [N]}$'s.

The randomized exploration component (\Cref{line:exploration}) is changed to the following: in each exploration round $t$ (happening uniformly w.p. $\frac{1}{\lvert \mE_\ell\rvert}$), we sample an agent $i\sim \Unif[K]$, a random index $n\sim \Unif[N]$, and a uniform price $p\sim \Unif[0,1]$. We allocate one resource (and forfeit the remaining) to agent $i$ should their $n$-th report, namely $u_{t,i,n}$, be at least $p$, where the payment is $p_{t,i}=p$. Otherwise, we forfeit all $N$ items.
(We only use one resource for randomized exploration in order to keep the \algname analysis in \Cref{sec:sketch PrimalAlloc} mostly unchanged. Using multiple resources for exploration only gives logarithmic improvements in the regret.)

\paragraph{Dual-Update Subroutines.}
We also give multi-unit analogs of the dual-update subroutines in \Cref{sec:mechanism dual}. For the \ftrl-based \Cref{alg:lambda FTRL}, we change the per-round consumption vector from $\bm c_{\tau,i_\tau}$ to $\bm c_{\tau,I_\tau}$:
\begin{equation}\label{eq:lambda FTRL multi-unit}
\bm{\lambda}_\ell = \argmin_{\bm{\lambda} \in \bm{\Lambda}} \Biggl(\sum_{\ell' < \ell} \sum_{\tau \in \mE_{\ell'}} \left (\bm{\rho} - \bm{c}_{\tau, I_\tau}\right )^\trans \bm{\lambda} + \frac{1}{\eta_\ell} \Psi(\bm{\lambda})\Biggr).
\end{equation}
This is because the \dualreg to online learning reduction in \Cref{lem:DualUpd to regret informal}---which only uses the linear programming formulation of $\{I_t^\ast\}_{t\in [T]}$---still holds in multi-unit setups.
Plugging \Cref{eq:lambda FTRL multi-unit} into the multi-unit \algname in \Cref{eq:allocation multi-unit multi-demand,eq:payment multi-unit multi-demand} gives the multi-unit \mech mechanism; its guarantee is in \Cref{thm:multi-unit}.

For the \oftrlfp-based \Cref{alg:lambda O-FTRL-FP}, we now solve for an approximate fixed point of the following system. Here, in addition to changing the per-round consumption from $\bm c_{\tau,i_\tau}$ to $\bm c_{\tau,I_\tau}$, we also changed the historical counterfactual allocation from $\tilde i_\tau(\bm \lambda_\ell)$ to $\tilde I_\tau(\bm \lambda_\ell)$ due to the difference between \Cref{line:primal allocation,eq:allocation multi-unit multi-demand}:
\begin{equation}\label{eq:lambda O-FTRL-FP multi-unit}
\begin{aligned}
\bm{\lambda}_\ell
\approx &\argmin_{\bm{\lambda} \in \bm{\Lambda}} \Biggl (\sum_{\ell' < \ell} \sum_{\tau \in \mE_{\ell'}} (\bm{\rho} - \bm{c}_{\tau, I_\tau})^\trans \bm{\lambda} + \tilde{\bm g}_\ell(\bm \lambda_\ell)^\trans \bm \lambda +  \frac{1}{\eta_\ell} \Psi(\bm{\lambda})\Biggr ),\\
\text{where }
\tilde{\bm g}_\ell(\bm \lambda_\ell)&:=\frac{\lvert \mE_\ell\rvert}{\sum_{\ell'<\ell} \lvert \mE_{\ell'}\rvert} \sum_{\ell'<\ell,\tau \in \mE_{\ell'}} \left (\bm \rho-\bm c_{\tau,\tilde I_\tau(\bm \lambda_\ell)}\right ),\\
\tilde I_\tau(\bm \lambda_\ell) &:= \argmax_{I\in \mathcal I} \bigl(u_{\tau,I}-\bm \lambda_\ell^\trans \bm c_{\tau,I} \bigr ),\quad \forall \ell'<\ell,\ \tau\in \mE_{\ell'}.\end{aligned}
\end{equation}

Similar to \Cref{eq:allocation multi-unit multi-demand,eq:payment multi-unit multi-demand}, $\tilde I_\tau(\bm \lambda_\ell)$ is also equivalent to the top-$N$ of $\{u_{\tau,i,n}-\bm \lambda_\ell^\trans \bm c_{\tau,i,n}\}_{i\in [K],n\in [N]}$.

\paragraph{Guarantees.}
We give the following theorem, whose proof is given in \Cref{sec:appendix extensions} and sketched below.
\begin{theorem}[Multi-Unit \mech and \mechO]\label{thm:multi-demand}
The multi-unit \mech mechanism ensures $\mR_T=\Otil\big ((K^2 N^2+\sqrt d N) T^{2/3}\big )$, and multi-unit \mechO ensures $\mR_T=\Otil\big ((K^2 N^2+d N) \sqrt T\big )$. Both mechanisms remain feasible and incentive-aware as in \Cref{thm:main theorem FTRL,thm:main theorem O-FTRL-FP}.
\end{theorem}

As a special case, if each agent only demands one resource but has no use for more than one (i.e., $v_{t,i,n}=0$ for all $n\ge 2$), the allocation in \Cref{eq:allocation multi-unit multi-demand}, payment in \Cref{eq:payment multi-unit multi-demand}, and \oftrlfp in \Cref{eq:lambda O-FTRL-FP multi-unit} now only involves top-$N$ of $K$ values. This not only speeds up computation, but also gives the following refined regret bound:
\begin{theorem}[Multi-Unit \mech and \mechO with Unit Demand]\label{thm:multi-unit}
With single-unit demand, the multi-unit \mech mechanism ensures $\mR_T=\Otil\big ((K^2+\sqrt d N) T^{2/3}\big )$, and multi-unit \mechO ensures $\mR_T\big ((K^2+d N) \sqrt T\big )$. Both mechanisms remain feasible and incentive-aware.
\end{theorem}
\begin{proof}[Proof Sketch of \Cref{thm:multi-demand,thm:multi-unit}]
Since the allocation and payment rule in \Cref{eq:allocation multi-unit multi-demand,eq:payment multi-unit multi-demand} comes from VCG, should agents be myopic, they would find truth-telling dominant; that is, \Cref{thm:informal IntraEpoch} still holds. Therefore, for $\gamma$-impatient agents in \Cref{def:impatient}, due to lazy updates and randomized exploration, unilateral deviation from truth-telling is unprofitable unless the current epoch is about to end. Hence there exists an almost-truthful PBE, and consequently the allocations attain primal efficiency. This recovers \Cref{thm:informal IAPD main theorem}.

On the dual side, the online learning loss in \Cref{lem:DualUpd to regret informal} becomes $F_\ell(\bm \lambda)=\big (\sum_{t\in \mE_\ell} (\bm \rho-\sum_{i_t\in I_t}\bm c_{t,i_t})\big )^\trans \bm \lambda$.
The multi-unit \mech bound comes from \Cref{eq:FTRL online learning regret}, but noticing that $\lVert \nabla F_\ell\rVert_2^2$ can now be of order $\O(N^2 \lvert \mE_\ell\rvert^2)$ due to the inflated magnitudes of $\bm \rho$ and $\bm c_{t,I_t}$. For \mechO, as sketched in \Cref{sec:circular dependency}, we need to ensure $(\epsilon,L)$-gradient stability (\Cref{def:gradient stability}); this requires that under small perturbations in $\bm \lambda_\ell$, with high probability $\tilde I_\ell(\bm \lambda_\ell)$ remains unchanged.
Although the candidate allocation actions, namely $I_t\in \mathcal I$, are exponentially more than that in single-unit cases, due to the top-$N$ structure in \Cref{eq:lambda O-FTRL-FP multi-unit}, we still manage to adopt an argument similar to that of \Cref{thm:DualUpd guarantee O-FTRL-FP informal}. See \Cref{sec:appendix extensions} for the full arguments. 
\end{proof}

\section{Numerical Studies}\label{sec:numerical illustration setup}
We evaluate our mechanisms against strategic agents across several resource-allocation settings.
In all experiments, we simulate strategic agents (who are $\gamma=0.9$-impatient; \Cref{def:impatient}) via continuous Q-learning \citep{watkins1992q}: In round $t$, agent $i$ observes state $s_{t,i}:=(v_{t,i},\frac tT,\bm \lambda_t)$ and decides their report $u_{t,i}\in [0,1]$ $\epsilon$-greedily over a ``Q-function'' $Q_i(s_{t,i},u_{t,i})$, defined as the expected $\gamma$-discounted utility of reporting $u_{t,i}$ at $s_{t,i}$.
For each resource-allocation setting considered below, we repeat the same $T=1000$-round game for $1000$ training episodes, during which agents update their Q-functions (a small neural network; see \Cref{sec:appendix numerical}). All metrics are averaged over $5$ independent trials to measure the standard error.

We conduct three experiments: comparison of \mech with non-strategic benchmarks, the role of \algname components towards mitigating the indirect impact, and the trade-off between \ftrl and \oftrlfp.

\subsection{Incentive-Aware Resource Allocation versus Non-Strategic Ones}\label{sec:incentive-aware vs non-strategic}

\paragraph{Setup.}
We consider five synthetic settings. Online packing is the base setup, and others modify one component of it: the feasible allocation actions, the consumption structure, the correlation between values and consumptions, or the units of resources that can be allocated in each round. This allows us to evaluate the mechanisms across distinct sources of difficulty. Full parameter specifications are provided in \Cref{sec:appendix numerical}.

\begin{enumerate}
\item ``Online packing'' is adapted from the online packing LP literature \citep{kesselheim2018primal}. In each round, $K=5$ agents have random values and consume random amounts of resources along all $d=5$ dimensions. As each allocation simultaneously consumes on several dimensions and the per-round budget $\bm \rho$ is relatively tight, the planner must track scarcity across multiple coupled constraints via dual updates.

\item ``Random availability'' is motivated by stochastic online matching, where the set of feasible allocations changes randomly over time \citep{feldman2009online,manshadi2012online}. 
Specifically, each agent is associated with a distinct resource dimension and is independently available in each round. An available agent has a positive value and consumes one unit of its corresponding resource, whereas an unavailable agent has zero value or consumption. Relative to online packing, this removes dense multi-dimensional consumptions and instead introduces random variation in which allocations are feasible. It tests whether the mechanism can respond effectively when allocation opportunities change unpredictably over time.

\item ``Fairness constraint'' changes the consumption structure to $d=2$: One agent consumes one dimension, while all other agents consume the other. Thus persistently allocating to one side depletes its corresponding capacity, raises its dual price, and pushes future allocations to the other side. We use this stylized structure to make the intertemporal role of the dual variables particularly transparent and to analyze the indirect impact in \Cref{fig:vicious cycle}: Strategic reports can distort future duals via current misreports; see \Cref{example:affect the future}.

\item ``Correlated value-cost'' retains the online-packing structure but introduces positive correlation between values and consumptions (i.e., higher-value requests consume more), motivated by the stochastic resource-consumption setting \citep{jiang2020online}. While our theoretical model assumes independent $v_{t,i}$ and $\bm c_{t,i}$, as a robustness check, this setting tests the mechanism when independence assumptions are violated.

\item ``Multi-unit demand'' moves beyond the single-unit resource allocation model. Up to three identical units can be allocated in each round, and each agent has diminishing marginal values across units, following the multi-unit valuation model of \citet{robu2013online}, but we omit the stochastic arrivals and departures therein to remain consistent with our fixed-horizon setting. We apply the multi-unit multi-demand extension of \mech defined in \Cref{eq:allocation multi-unit multi-demand,eq:payment multi-unit multi-demand,eq:lambda FTRL multi-unit}. This setting evaluates whether the mechanisms continue to perform well when both allocations and agents' demands are multi-unit, as analyzed in \Cref{sec:multi-unit}.
\end{enumerate}

\paragraph{Mechanisms.}
We benchmark our incentive-aware mechanism \mech (\Cref{thm:main theorem FTRL}) against two state-of-the-art non-strategic resource-allocation algorithms: \texttt{DMD} \citep{balseiro2020dual} and \texttt{Simple} \citep{li2023simple}. Both algorithms' primal parts allocate to the agent maximizing dual-adjusted report $u_{t,i}-\bm \lambda_t^\trans \bm c_{t,i}$, but have different dual-update subroutines (recall \Cref{fig:primal-dual}). These non-strategic algorithms do not require or utilize payments, so their revenues are naturally zero. We present their pseudo-codes in \Cref{sec:appendix numerical}.

\paragraph{Metrics.}
We report five metrics for each mechanism. The first two evaluates the planner: ``{SW}'' is the social welfare $\sum_t v_{t,i_t}$ normalized by the offline optimum $\sum_t v_{t,i_t^\ast}$, and ``{Rev}'' is the revenue $\sum_t p_{t,i_t}$ normalized by that under the offline optimum (that is, $\sum_t p_{t,i_t^\ast}$). The third one is for agents' utilities: ``{MinUtil}'' is the minimum ratio of an agent's utility when all agents are strategic, to the same agent's utility when all agents are truthful (both under the same mechanism). Our last two metrics measure the incentive-awareness and primal efficiency (see the two bullet points of \Cref{thm:informal IAPD main theorem}): ``{RankCorr}'' is the Spearman rank correlation between agents' values and reports, and ``{PrimalAlloc}'' is the fraction of rounds where there's no misallocation due to misreports, namely $\frac 1T \sum_t \1[i_t=\tilde i_t^\ast]$ with the dual-adjusted optimum $\tilde i_t^\ast=\argmax_i (v_{t,i}-\bm \lambda_t^\trans \bm c_{t,i})$ in \Cref{eq:tilde i definition main}.

\begin{table}[t]
\centering
\small
\caption{\mech versus non-strategic algorithms. Bold are best or within 1\%. Stderr in parentheses.}
\label{table:incentive-aware vs non-strategic}
\resizebox{\textwidth}{!}{%
\begin{tabular}{llccccc}
\hline
Setting & Mechanism & SW $\uparrow$ & Rev $\uparrow$ & MinUtil $\uparrow$ & RankCorr $\uparrow$ & PrimalAlloc $\uparrow$ \\
\hline
\multirow{3}{*}{online packing} & \texttt{DMD} \citep{balseiro2020dual} & 67.8\% {\scriptsize ($\pm$0.0)} & 0.0\% {\scriptsize ($\pm$0.0)} & 61.2\% {\scriptsize ($\pm$0.3)} & 0.058 {\scriptsize ($\pm$0.001)} & 42.7\% {\scriptsize ($\pm$0.0)} \\
\cdashline{2-7}
 & \texttt{Simple} \citep{li2023simple} & 67.0\% {\scriptsize ($\pm$0.0)} & 0.0\% {\scriptsize ($\pm$0.0)} & 51.7\% {\scriptsize ($\pm$0.4)} & -0.004 {\scriptsize ($\pm$0.001)} & 40.6\% {\scriptsize ($\pm$0.0)} \\
\cdashline{2-7}
 & \cellcolor{black!7}\mech (ours) & \cellcolor{black!7} \textbf{85.0\% {\scriptsize ($\pm$0.1)}} & \cellcolor{black!7} \textbf{64.8\% {\scriptsize ($\pm$0.1)}} & \cellcolor{black!7} \textbf{91.5\% {\scriptsize ($\pm$0.6)}} & \cellcolor{black!7} \textbf{0.754 {\scriptsize ($\pm$0.002)}} & \cellcolor{black!7} \textbf{68.4\% {\scriptsize ($\pm$0.1)}} \\
\hline
\multirow{3}{*}{random availability} & \texttt{DMD} \citep{balseiro2020dual} & 58.6\% {\scriptsize ($\pm$0.1)} & 0.0\% {\scriptsize ($\pm$0.0)} & 12.2\% {\scriptsize ($\pm$0.2)} & 0.253 {\scriptsize ($\pm$0.002)} & 25.9\% {\scriptsize ($\pm$0.1)} \\
\cdashline{2-7}
 & \texttt{Simple} \citep{li2023simple} & 59.0\% {\scriptsize ($\pm$0.1)} & 0.0\% {\scriptsize ($\pm$0.0)} & 17.6\% {\scriptsize ($\pm$0.2)} & 0.231 {\scriptsize ($\pm$0.001)} & 26.6\% {\scriptsize ($\pm$0.1)} \\
\cdashline{2-7}
 & \cellcolor{black!7}\mech (ours) & \cellcolor{black!7} \textbf{85.5\% {\scriptsize ($\pm$0.1)}} & \cellcolor{black!7} \textbf{68.7\% {\scriptsize ($\pm$0.1)}} & \cellcolor{black!7} \textbf{66.9\% {\scriptsize ($\pm$0.6)}} & \cellcolor{black!7} \textbf{0.808 {\scriptsize ($\pm$0.002)}} & \cellcolor{black!7} \textbf{64.4\% {\scriptsize ($\pm$0.2)}} \\
\hline
\multirow{3}{*}{fairness constraint} & \texttt{DMD} \citep{balseiro2020dual} & 80.5\% {\scriptsize ($\pm$0.0)} & 0.0\% {\scriptsize ($\pm$0.0)} & 5.4\% {\scriptsize ($\pm$0.2)} & 0.008 {\scriptsize ($\pm$0.001)} & 35.3\% {\scriptsize ($\pm$0.0)} \\
\cdashline{2-7}
 & \texttt{Simple} \citep{li2023simple} & 82.0\% {\scriptsize ($\pm$0.0)} & 0.0\% {\scriptsize ($\pm$0.0)} & 4.7\% {\scriptsize ($\pm$0.2)} & 0.037 {\scriptsize ($\pm$0.001)} & 33.3\% {\scriptsize ($\pm$0.0)} \\
\cdashline{2-7}
 & \cellcolor{black!7}\mech (ours) & \cellcolor{black!7} \textbf{89.3\% {\scriptsize ($\pm$0.0)}} & \cellcolor{black!7} \textbf{71.6\% {\scriptsize ($\pm$0.1)}} & \cellcolor{black!7} \textbf{48.6\% {\scriptsize ($\pm$0.4)}} & \cellcolor{black!7} \textbf{0.341 {\scriptsize ($\pm$0.002)}} & \cellcolor{black!7} \textbf{53.5\% {\scriptsize ($\pm$0.1)}} \\
\hline
\multirow{3}{*}{correlated value-cost} & \texttt{DMD} \citep{balseiro2020dual} & 53.0\% {\scriptsize ($\pm$0.0)} & 0.0\% {\scriptsize ($\pm$0.0)} & 46.6\% {\scriptsize ($\pm$0.2)} & -0.096 {\scriptsize ($\pm$0.001)} & 11.3\% {\scriptsize ($\pm$0.0)} \\
\cdashline{2-7}
 & \texttt{Simple} \citep{li2023simple} & 52.1\% {\scriptsize ($\pm$0.0)} & 0.0\% {\scriptsize ($\pm$0.0)} & \textbf{52.6\% {\scriptsize ($\pm$0.1)}} & -0.090 {\scriptsize ($\pm$0.002)} & 7.5\% {\scriptsize ($\pm$0.0)} \\
\cdashline{2-7}
 & \cellcolor{black!7}\mech (ours) & \cellcolor{black!7} \textbf{81.3\% {\scriptsize ($\pm$0.1)}} & \cellcolor{black!7} \textbf{72.9\% {\scriptsize ($\pm$0.1)}} & \cellcolor{black!7} 46.1\% {\scriptsize ($\pm$1.2)} & \cellcolor{black!7} \textbf{0.691 {\scriptsize ($\pm$0.002)}} & \cellcolor{black!7} \textbf{60.7\% {\scriptsize ($\pm$0.2)}} \\
\hline
\multirow{3}{*}{multi-unit demand} & \texttt{DMD} \citep{balseiro2020dual} & 69.0\% {\scriptsize ($\pm$0.0)} & 0.0\% {\scriptsize ($\pm$0.0)} & 64.0\% {\scriptsize ($\pm$0.4)} & 0.402 {\scriptsize ($\pm$0.001)} & 11.2\% {\scriptsize ($\pm$0.0)} \\
\cdashline{2-7}
 & \texttt{Simple} \citep{li2023simple} & 69.8\% {\scriptsize ($\pm$0.0)} & 0.0\% {\scriptsize ($\pm$0.0)} & 60.4\% {\scriptsize ($\pm$0.4)} & 0.357 {\scriptsize ($\pm$0.001)} & 11.6\% {\scriptsize ($\pm$0.0)} \\
\cdashline{2-7}
 & \cellcolor{black!7}\mech (ours) & \cellcolor{black!7} \textbf{88.9\% {\scriptsize ($\pm$0.1)}} & \cellcolor{black!7} \textbf{66.6\% {\scriptsize ($\pm$0.1)}} & \cellcolor{black!7} \textbf{90.2\% {\scriptsize ($\pm$0.5)}} & \cellcolor{black!7} \textbf{0.699 {\scriptsize ($\pm$0.001)}} & \cellcolor{black!7} \textbf{47.6\% {\scriptsize ($\pm$0.2)}} \\
\hline
\end{tabular}}
\end{table}

\paragraph{Results.}
\Cref{table:incentive-aware vs non-strategic} shows that our incentive-aware resource allocation mechanism, \mech, substantially outperforms non-strategic primal-dual benchmarks. Across all five settings, \mech is the only mechanism that consistently attains high SW, Rev, RankCorr, and PrimalAlloc (MinUtil is also the highest except for one setup). For example, in online packing, relative to the best non-strategic benchmark, \mech improves SW from $67.8\%$ to $85.0\%$, Rev from $0.0\%$ to $64.8\%$, RankCorr from $0.058$ to $0.754$, and PrimalAlloc from $42.7\%$ to $68.4\%$. The same pattern persists in the other settings. The visualization in \Cref{fig:plot main} uses the correlated value-cost setup: It plots value-report scatters, SW across episodes, and misallocation due to misreports rate ($1-\text{PrimalAlloc}$) under \texttt{DMD} \citep{balseiro2020dual} and \mech.

\subsection{Role of Indirect Impact: Epoch-Based Lazy Updates and Randomized Exploration}
\label{example:affect the future}

\Cref{table:incentive-aware vs non-strategic} compares \mech with allocation algorithms that were designed for non-strategic reports. We now exhibit the more specific issue identified in \Cref{sec:mechanism primal}: Even with dual-adjusted payments---which deters the direct benefit of misreporting---strategic agents may still utilize the \emph{indirect impact} in \Cref{fig:vicious cycle} to influence future duals and allocations.
We consider the fairness-constraint setting, where one single agent consumes dimension 1 and all others consume dimension 2. Thus if an agent strategically under-reports in a round $t\in \mE_\ell$, the allocation $i_t$ immediately shifts toward the other side; the planner hence observes an imbalanced consumption $\sum_{t\in \mE_\ell} \bm c_{t,i_t}$. Therefore, $\bm \lambda_{\ell+1}$ penalizes the other side, which makes future allocations favorable to the misreporting agent. The ablation below tests the impact of such indirect manipulation.

\paragraph{Mechanisms.}
In \Cref{fig:fairness_dual_manipulation}, we additionally implement two mechanisms: ``{PayOnly}'' is primal-dual equipped with dual-adjusted payments (\Cref{sec:dual-adjusted payments}) but not epoch-based lazy updates (\Cref{sec:epoch-based lazy updates}) or randomized exploration (\Cref{sec:random exploration}). ``{NoExplore}'' additionally keeps epoch-based lazy updates (but not exploration).

\paragraph{Results.}
The comparison in \Cref{fig:fairness_dual_manipulation} reveals a trade-off: PayOnly achieves the highest SW (92.6\%), followed by NoExplore (90.7\%) and \mech (89.3\%). We argue this is non-surprising: Compared to PayOnly, \mech adds lazy updates---slowing down the adjustment of dual variables---and randomized exploration that deliberately departing from the welfare-maximizing allocation.
But these additional components are pivotal to tame strategic behavior. Indeed, \mech is better on the remaining report- and allocation-related metrics: PayOnly has a much lower RankCorr ($0.304$ versus $0.341$) and PrimalAlloc ($45.9\%$ versus $53.5\%$). Adding lazy updates alone removes much of the allocation-inefficiency: NoExplore reaches almost the same PrimalAlloc as \mech but has a lower RankCorr ($0.320$), i.e., agents remain less truthful. In other words, lazy updates limit agents' ability to manipulate future dual variables and allocations, while randomized exploration strengthens incentives for truthful reporting; both agree well with our \Cref{sec:mechanism primal}.

\begin{figure}[t]
\centering
\begin{tikzpicture}
\begin{groupplot}[
    group style={group size=3 by 1,horizontal sep=1.5cm},
    width=0.23\textwidth,
    height=0.24\textwidth,
    ymajorgrids=true,
    grid style={black!10},
    axis line style={black!60},
    tick style={black!60},
    tick label style={black,font=\tiny},
    label style={black},
    title style={black},
    every axis plot/.append style={draw=black},
]
\nextgroupplot[
    title={SW/OPT $\uparrow$},
    symbolic x coords={PaidPD,NoExplor,IAPD},
    xtick={PaidPD,NoExplor,IAPD},
    xticklabels={PayOnly,NoExplore,\mech},
    x tick label style={font=\tiny,rotate=28,anchor=east},
    ymin=0,
    ymax=109.2608,
    yticklabel={\pgfmathprintnumber{\tick}\%},
    nodes near coords,
    every node near coord/.append style={font=\scriptsize,text=black},
    point meta=explicit symbolic,
]
\addplot[mark=none,ybar,bar width=7pt,draw=black,fill=black!15,error bars/.cd,y dir=both,y explicit] coordinates {(PaidPD,92.5514) +- (0,0.0425) [\textcolor{black}{$92.6\%$}]};
\addplot[mark=none,ybar,bar width=7pt,draw=black,fill=black!40,error bars/.cd,y dir=both,y explicit] coordinates {(NoExplor,90.6626) +- (0,0.0377) [\textcolor{black}{$90.7\%$}]};
\addplot[mark=none,ybar,bar width=7pt,draw=black,fill=black!65,postaction={pattern=north east lines},error bars/.cd,y dir=both,y explicit] coordinates {(IAPD,89.2701) +- (0,0.0347) [\textcolor{black}{$89.3\%$}]};

\nextgroupplot[
    title={RankCorr $\uparrow$},
    symbolic x coords={PaidPD,NoExplor,IAPD},
    xtick={PaidPD,NoExplor,IAPD},
    xticklabels={PayOnly,NoExplore,\mech},
    x tick label style={font=\tiny,rotate=28,anchor=east},
    ymin=0,
    ymax=0.4054,
    nodes near coords,
    every node near coord/.append style={font=\scriptsize,text=black},
    point meta=explicit symbolic,
]
\addplot[mark=none,ybar,bar width=7pt,draw=black,fill=black!15,error bars/.cd,y dir=both,y explicit] coordinates {(PaidPD,0.3037) +- (0,0.0023) [\textcolor{black}{$0.304$}]};
\addplot[mark=none,ybar,bar width=7pt,draw=black,fill=black!40,error bars/.cd,y dir=both,y explicit] coordinates {(NoExplor,0.3203) +- (0,0.0019) [\textcolor{black}{$0.320$}]};
\addplot[mark=none,ybar,bar width=7pt,draw=black,fill=black!65,postaction={pattern=north east lines},error bars/.cd,y dir=both,y explicit] coordinates {(IAPD,0.3414) +- (0,0.0021) [\textcolor{black}{$0.341$}]};

\nextgroupplot[
    title={PrimalAlloc $\uparrow$},
    symbolic x coords={PaidPD,NoExplor,IAPD},
    xtick={PaidPD,NoExplor,IAPD},
    xticklabels={PayOnly,NoExplore,\mech},
    x tick label style={font=\tiny,rotate=28,anchor=east},
    ymin=0,
    ymax=63.3247,
    yticklabel={\pgfmathprintnumber{\tick}\%},
    nodes near coords,
    every node near coord/.append style={font=\scriptsize,text=black},
    point meta=explicit symbolic,
]
\addplot[mark=none,ybar,bar width=7pt,draw=black,fill=black!15,error bars/.cd,y dir=both,y explicit] coordinates {(PaidPD,45.8643) +- (0,0.1671) [\textcolor{black}{$45.9\%$}]};
\addplot[mark=none,ybar,bar width=7pt,draw=black,fill=black!40,error bars/.cd,y dir=both,y explicit] coordinates {(NoExplor,53.4603) +- (0,0.1594) [\textcolor{black}{$53.5\%$}]};
\addplot[mark=none,ybar,bar width=7pt,draw=black,fill=black!65,postaction={pattern=north east lines},error bars/.cd,y dir=both,y explicit] coordinates {(IAPD,53.5467) +- (0,0.1183) [\textcolor{black}{$53.5\%$}]};
\end{groupplot}
\end{tikzpicture}
\caption{Strategic manipulation in fairness constraint ($T=1000$, 1000 episodes). ``PayOnly'' removes both epoching and exploration from \mech, and ``NoExplore'' only removes exploration. Error bars are stderr.}
\label{fig:fairness_dual_manipulation}\vspace{-10pt}
\end{figure}

\subsection{Trade-Off between Efficiency and Computational Costs in FTRL and O-FTRL-FP}\label{sec:computational cost of O-FTRL-FP}

We now compare \mech to \mechO.
The latter has a much sharper regret bound (from $\Otil(T^{2/3})$ to $\Otil(\sqrt T)$), but is computationally more expensive: Each \ftrl update aggregates historical consumptions and performs a closed-form projection; in contrast, \oftrlfp needs to calculate how a dual $\bm \lambda_\ell$ would change historical counterfactual allocations, and solve an induced (approximate) fixed-point problem.
We provide an damped iterative approximation of the fixed-point problem \Cref{eq:lambda O-FTRL-FP} in \Cref{sec:appendix numerical}: Starting from the previous dual $\bm \lambda_{\ell-1}$, it performs fixed-point iterations of \Cref{eq:lambda O-FTRL-FP}, until either the fixed-point residual (the $L_2$ difference between the LHS and the RHS) is small enough or the number of iterations reaches 100.

\paragraph{Setup.}
Using the damped iterative approximation of \oftrlfp, we implement an approximated version of \mechO and compare it against \mech in \Cref{tab:oftrlfp_comparison}. In addition to the main performance metrics of SW, Rev, RankCorr, and PrimalAlloc (see \Cref{sec:incentive-aware vs non-strategic}), we additionally report computational metrics: ``FPResidual'' is the fixed-point residual of the found dual $\bm \lambda_\ell$; ``FPIteration'' is the average number of fixed-point iterates conducted; and ``Runtime'' is the per-update running time of the dual-update subroutine.

\paragraph{Results.}
From \Cref{tab:oftrlfp_comparison}, we observe that \mechO improves SW over \mech in all settings (which one exception discussed later). This agrees with our theoretical analysis that \oftrlfp attains lower online learning regret $\mR_L^{\text{OL}}$, and consequently, lower social welfare regret $\mR_T$ (cf. \Cref{thm:main theorem FTRL,thm:main theorem O-FTRL-FP}).
On the other hand, both mechanisms exhibit similar RankCorr and PrimalAlloc; this is because the primal component that is in charge of incentive-awareness and \primalreg, i.e., the \algname framework, is shared.

The only setting where \oftrlfp performs worse than \ftrl (in terms of SW) is the fairness constraint setting. The reason can be seen from the FPResidual column: After an average of 70.1 fixed-point iterations, the fixed-point residual remains as large as 0.059, substantially larger than its counterpart in other settings.
This is due to the structure of fairness constraint: Since there are only two dimensions and their consumptions are usually balanced, a small change in the dimension-2 dual variable $\bm \lambda_\ell$ can switch many historical counterfactual allocations from one side to the other, making the predicted gradient $\tilde{\bm g}_\ell(\bm\lambda_\ell)$ very sensitive to $\bm\lambda_\ell$.
Indeed, following the $(\epsilon,L)$-gradient-stability notion in \Cref{def:gradient stability}, we perturb each epoch's dual variable $\bm\lambda_\ell$ in a uniformly random direction with $\|\Delta\bm\lambda\|_2=0.05$, and measure the resulting change in the predicted gradient, i.e., $\lVert \tilde{\bm g}_\ell(\bm\lambda_\ell+\Delta \bm \lambda)-\tilde{\bm g}_\ell(\bm\lambda_\ell)\rVert_2$. The average change is $4.73$ in fairness constraint, compared with $3.34$ in correlated value-cost, $1.11$ in online packing, and $1.18$ in random availability.
Thus, the same structure that helped isolate the indirect impact in \Cref{example:affect the future} also makes fixed-point approximations harder.

Overall, \Cref{tab:oftrlfp_comparison} confirms another trade-off: when good approximate fixed points are found, \oftrlfp improves performance over \ftrl; but this comes with extra computational costs and is sometimes hard.

\begin{table}[t]
\centering
\caption{\mech versus approximate \mechO. Bold are best or within 1\%. Stderr in parentheses.}
\label{tab:oftrlfp_comparison}
\resizebox{\textwidth}{!}{%
\begin{tabular}{llccccccrl}
\hline
Setting & Dual-Update & SW $\uparrow$ & Rev $\uparrow$ & RankCorr $\uparrow$ & PrimalAlloc $\uparrow$ & FPResidual $\downarrow$ & FPIteration & \multicolumn{2}{c}{Runtime} \\
\hline
\multirow{2}{*}{online packing} & \ftrl & 85.0\% {\scriptsize ($\pm$0.09)} & 64.8\% {\scriptsize ($\pm$0.08)} & \textbf{0.754 {\scriptsize ($\pm$0.002)}} & \textbf{68.4\% {\scriptsize ($\pm$0.11)}} & \multicolumn{2}{c}{ } & 0.85 & ms \\
\cdashline{2-10}
 & \cellcolor{black!7}\oftrlfp & \cellcolor{black!7}\textbf{87.4\% {\scriptsize ($\pm$0.09)}} & \cellcolor{black!7}\textbf{69.1\% {\scriptsize ($\pm$0.07)}} & \cellcolor{black!7}\textbf{0.751 {\scriptsize ($\pm$0.002)}} & \cellcolor{black!7}\textbf{68.3\% {\scriptsize ($\pm$0.11)}} & \cellcolor{black!7}0.006 {\scriptsize ($\pm$0.000)} & \cellcolor{black!7}68.7  {\scriptsize ($\pm$0.1)} & \cellcolor{black!7}116.5 & \cellcolor{black!7}ms \\
\hline
\multirow{2}{*}{random availability} & \ftrl & 85.5\% {\scriptsize ($\pm$0.12)} & 68.7\% {\scriptsize ($\pm$0.06)} & \textbf{0.808 {\scriptsize ($\pm$0.002)}} & \textbf{64.4\% {\scriptsize ($\pm$0.16)}} & \multicolumn{2}{c}{ } & 0.95 & ms \\
\cdashline{2-10}
 & \cellcolor{black!7}\oftrlfp & \cellcolor{black!7}\textbf{89.5\% {\scriptsize ($\pm$0.13)}} & \cellcolor{black!7}\textbf{73.0\% {\scriptsize ($\pm$0.06)}} & \cellcolor{black!7}\textbf{0.803 {\scriptsize ($\pm$0.002)}} & \cellcolor{black!7}\textbf{64.8\% {\scriptsize ($\pm$0.16)}} & \cellcolor{black!7}0.001 {\scriptsize ($\pm$0.000)} & \cellcolor{black!7}5.5  {\scriptsize ($\pm$0.1)} & \cellcolor{black!7}31.1 & \cellcolor{black!7}ms \\
\hline
\multirow{2}{*}{fairness constraint} & \ftrl & \textbf{89.3\% {\scriptsize ($\pm$0.03)}} & \textbf{71.6\% {\scriptsize ($\pm$0.07)}} & \textbf{0.341 {\scriptsize ($\pm$0.002)}} & \textbf{53.5\% {\scriptsize ($\pm$0.12)}} & \multicolumn{2}{c}{ } & 0.95 & ms \\
\cdashline{2-10}
 & \cellcolor{black!7}\oftrlfp & \cellcolor{black!7}85.6\% {\scriptsize ($\pm$0.07)} & \cellcolor{black!7}62.9\% {\scriptsize ($\pm$0.22)} & \cellcolor{black!7}0.198 {\scriptsize ($\pm$0.003)} & \cellcolor{black!7}\textbf{53.8\% {\scriptsize ($\pm$0.14)}} & \cellcolor{black!7}0.059 {\scriptsize ($\pm$0.001)} & \cellcolor{black!7}70.1  {\scriptsize ($\pm$0.1)} & \cellcolor{black!7}135.8 & \cellcolor{black!7}ms \\
\hline
\multirow{2}{*}{correlated value-cost} & \ftrl & 81.3\% {\scriptsize ($\pm$0.09)} & 72.9\% {\scriptsize ($\pm$0.13)} & 0.691 {\scriptsize ($\pm$0.002)} & \textbf{60.7\% {\scriptsize ($\pm$0.15)}} & \multicolumn{2}{c}{ } & 0.87 & ms \\
\cdashline{2-10}
 & \cellcolor{black!7}\oftrlfp & \cellcolor{black!7}\textbf{85.0\% {\scriptsize ($\pm$0.11)}} & \cellcolor{black!7}\textbf{77.7\% {\scriptsize ($\pm$0.13)}} & \cellcolor{black!7}\textbf{0.706 {\scriptsize ($\pm$0.002)}} & \cellcolor{black!7}55.7\% {\scriptsize ($\pm$0.15)} & \cellcolor{black!7}0.006 {\scriptsize ($\pm$0.000)} & \cellcolor{black!7}72.6  {\scriptsize ($\pm$0.1)} & \cellcolor{black!7}172.6 & \cellcolor{black!7}ms \\
\hline
\end{tabular}}
\end{table}

\section{Conclusion and Future Directions}
We investigate the dynamic allocation of resources to strategic agents under long-term constraints. Standard primal-dual, though effective in non-strategic settings, is vulnerable to strategic manipulation. We design a novel \algname framework, using payments, epoching, and exploration to attain incentive-awareness and primal efficiency.
When equipped with \ftrl, the resulting mechanism attains sub-linear $\Otil(T^{2/3})$ regret, satisfies all constraints, and admits a near-truthful PBE. This is the first mechanism that can attain efficiency, feasibility, and incentive-awareness at the same time in general constrained online resource allocation problems.

Identifying an $\Omega(T^{2/3})$ barrier as a price of incentives, we design a novel online learning algorithm, \oftrlfp. It not only exploits the predictability provided by incentive-awareness, but also formulates a fixed-point problem to accommodate endogenous and discontinuous predictions and can be of independent interest. We then boost the regret to $\Otil(\sqrt T)$, near-optimal even in non-strategic online resource allocation.

Several promising research directions remain open. Our framework shows how monetary transfers can be used to achieve incentive-awareness in dynamic allocation problems with strategic agents. A natural next step is to extend these ideas to important non-monetary domains---such as organ matching, school admissions, and food bank distribution---where transfers are unavailable or undesirable. Developing analogous incentive-aware mechanisms for such settings would further broaden the scope of robust dynamic mechanism design.

\bibliography{references}

@article{blanchard2024near,
  title={Near-Optimal Mechanisms for Resource Allocation Without Monetary Transfers},
  author={Blanchard, Moise and Jaillet, Patrick},
  journal={arXiv preprint arXiv:2408.10066},
  year={2024}
}

@inproceedings{gorokh2019remarkable,
  title={The Remarkable Robustness of the Repeated Fisher Market},
  author={Gorokh, Artur and Banerjee, Siddhartha and Iyer, Krishnamurthy},
  booktitle={Proceedings of the 22nd ACM Conference on Economics and Computation},
  pages={562--562},
  year={2021}
}

@article{amin2013learning,
  title={Learning prices for repeated auctions with strategic buyers},
  author={Amin, Kareem and Rostamizadeh, Afshin and Syed, Umar},
  journal={Advances in Neural Information Processing Systems},
  volume={26},
  year={2013}
}

@article{golrezaei2014real,
  title={Real-time optimization of personalized assortments},
  author={Golrezaei, Negin and Nazerzadeh, Hamid and Rusmevichientong, Paat},
  journal={Management Science},
  volume={60},
  number={6},
  pages={1532--1551},
  year={2014},
  publisher={INFORMS}
}

@inproceedings{golrezaei2021boosted,
  title={Boosted second price auctions: Revenue optimization for heterogeneous bidders},
  author={Golrezaei, Negin and Lin, Max and Mirrokni, Vahab and Nazerzadeh, Hamid},
  booktitle={Proceedings of the 27th ACM SIGKDD Conference on Knowledge Discovery \& Data Mining},
  pages={447--457},
  year={2021}
}

@inproceedings{berriaud2024spend,
  title={To Spend or to Gain: Online Learning in Repeated Karma Auctions},
  author={Berriaud, Damien and Elokda, Ezzat and Jalota, Devansh and Frazzoli, Emilio and Pavone, Marco and D{\"o}rfler, Florian},
  booktitle={Proceedings of the 24th International Conference on Autonomous Agents and Multiagent Systems},
  pages={289--297},
  year={2025}
}

@article{balseiro2020dual,
  title={The Best of Many Worlds: Dual Mirror Descent for Online Allocation Problems},
  author={Balseiro, Santiago R and Lu, Haihao and Mirrokni, Vahab},
  journal={Operations Research},
  volume={71},
  number={1},
  pages={101--119},
  year={2023},
  publisher={INFORMS}
}

@article{balseiro2019learning,
  title={Learning in repeated auctions with budgets: Regret minimization and equilibrium},
  author={Balseiro, Santiago R and Gur, Yonatan},
  journal={Management Science},
  volume={65},
  number={9},
  pages={3952--3968},
  year={2019},
  publisher={INFORMS}
}

@article{arlotto2019uniformly,
  title={Uniformly bounded regret in the multisecretary problem},
  author={Arlotto, Alessandro and Gurvich, Itai},
  journal={Stochastic Systems},
  volume={9},
  number={3},
  pages={231--260},
  year={2019},
  publisher={INFORMS}
}

@article{golrezaei2021dynamic,
  title={Dynamic Incentive-Aware Learning: Robust Pricing in Contextual Auctions},
  author={Golrezaei, Negin and Javanmard, Adel and Mirrokni, Vahab},
  journal={Operations Research},
  volume={69},
  number={1},
  pages={297--314},
  year={2021},
  publisher={INFORMS}
}

@inproceedings{golrezaei2023incentive,
  title={Incentive-aware contextual pricing with non-parametric market noise},
  author={Golrezaei, Negin and Jaillet, Patrick and Liang, Jason Cheuk Nam},
  booktitle={International Conference on Artificial Intelligence and Statistics},
  pages={9331--9361},
  year={2023},
  organization={PMLR}
}

@article{gorokh2021monetary,
  title={From monetary to nonmonetary mechanism design via artificial currencies},
  author={Gorokh, Artur and Banerjee, Siddhartha and Iyer, Krishnamurthy},
  journal={Mathematics of Operations Research},
  volume={46},
  number={3},
  pages={835--855},
  year={2021},
  publisher={INFORMS}
}

@inproceedings{rakhlin2013online,
  title={Online learning with predictable sequences},
  author={Rakhlin, Alexander and Sridharan, Karthik},
  booktitle={Conference on Learning Theory},
  pages={993--1019},
  year={2013},
  organization={PMLR}
}

@inproceedings{dai2025non,
  title={Non-Monetary Mechanism Design without Distributional Information: Using Scarce Audits Wisely},
  author={Dai, Yan and Blanchard, Mo{\"i}se and Jaillet, Patrick},
  booktitle={The Thirty Eighth Annual Conference on Learning Theory},
  pages={1366--1367},
  year={2025},
  organization={PMLR}
}

@article{koufogiannakis2014nearly,
  title={A nearly linear-time PTAS for explicit fractional packing and covering linear programs},
  author={Koufogiannakis, Christos and Young, Neal E},
  journal={Algorithmica},
  volume={70},
  pages={648--674},
  year={2014},
  publisher={Springer}
}

@article{cesa2024regret,
  title={Regret analysis of bilateral trade with a smoothed adversary},
  author={Cesa-Bianchi, Nicol{\`o} and Cesari, Tommaso and Colomboni, Roberto and Fusco, Federico and Leonardi, Stefano},
  journal={Journal of Machine Learning Research},
  volume={25},
  number={234},
  pages={1--36},
  year={2024}
}

@inproceedings{cesa2024role,
  title={The role of transparency in repeated first-price auctions with unknown valuations},
  author={Cesa-Bianchi, Nicol{\`o} and Cesari, Tommaso and Colomboni, Roberto and Fusco, Federico and Leonardi, Stefano},
  booktitle={Proceedings of the 56th Annual ACM Symposium on Theory of Computing},
  pages={225--236},
  year={2024}
}

@inproceedings{durvasula2023smoothed,
  title={Smoothed analysis of online non-parametric auctions},
  author={Durvasula, Naveen and Haghtalab, Nika and Zampetakis, Manolis},
  booktitle={Proceedings of the 24th ACM Conference on Economics and Computation},
  pages={540--560},
  year={2023}
}

@inproceedings{devanur2009adwords,
  title={The adwords problem: online keyword matching with budgeted bidders under random permutations},
  author={Devanur, Nikhil R and Hayes, Thomas P},
  booktitle={Proceedings of the 10th ACM Conference on Electronic Commerce},
  pages={71--78},
  year={2009}
}

@article{molinaro2014geometry,
  title={The geometry of online packing linear programs},
  author={Molinaro, Marco and Ravi, Ramamoorthi},
  journal={Mathematics of Operations Research},
  volume={39},
  number={1},
  pages={46--59},
  year={2014},
  publisher={INFORMS}
}

@article{yin2022online,
  title={Online allocation and learning in the presence of strategic agents},
  author={Yin, Steven and Agrawal, Shipra and Zeevi, Assaf},
  journal={Advances in Neural Information Processing Systems},
  volume={35},
  pages={6333--6344},
  year={2022}
}

@article{orabona2019modern,
  title={A modern introduction to online learning},
  author={Orabona, Francesco},
  journal={arXiv preprint arXiv:1912.13213},
  year={2019}
}

@inproceedings{zimmert2022return,
  title={Return of the bias: Almost minimax optimal high probability bounds for adversarial linear bandits},
  author={Zimmert, Julian and Lattimore, Tor},
  booktitle={Conference on Learning Theory},
  pages={3285--3312},
  year={2022},
  organization={PMLR}
}

@inproceedings{kohler2017sub,
  title={Sub-sampled cubic regularization for non-convex optimization},
  author={Kohler, Jonas Moritz and Lucchi, Aurelien},
  booktitle={International Conference on Machine Learning},
  pages={1895--1904},
  year={2017},
  organization={PMLR}
}

@article{duchi2011adaptive,
  title={Adaptive subgradient methods for online learning and stochastic optimization.},
  author={Duchi, John and Hazan, Elad and Singer, Yoram},
  journal={Journal of Machine Learning Research},
  volume={12},
  number={7},
  year={2011}
}

@article{s75,
title={Strategy-Proofness and {A}rrow's Conditions: {E}xistence and Correspondence Theorems for Voting Procedures and Social Welfare Functions},
author={Satterthwaite, N.A.},
journal={Journal of Economic Theory},
volume=10,
number=2,
pages={187-217},
year= 1975}

@article{gi73,
title={Manipulation of Voting Schemes: A General Result},
author={Gibbard, A.},
journal={Econometrica},
volume=41,
number=4,
pages={587-601},
year= 1973}

@article{clarke1971multipart,
  title={Multipart pricing of public goods},
  author={Clarke, Edward H},
  journal={Public Choice},
  pages={17--33},
  year={1971},
  publisher={JSTOR}
}

@article{groves1973incentives,
  title={Incentives in teams},
  author={Groves, Theodore},
  journal={Econometrica: Journal of the Econometric Society},
  pages={617--631},
  year={1973},
  publisher={JSTOR}
}

@article{vickrey1961counterspeculation,
  title={Counterspeculation, auctions, and competitive sealed tenders},
  author={Vickrey, William},
  journal={The Journal of Finance},
  volume={16},
  number={1},
  pages={8--37},
  year={1961},
  publisher={JSTOR}
}

@article{arrow1950difficulty,
  title={A difficulty in the concept of social welfare},
  author={Arrow, Kenneth J},
  journal={Journal of Political Economy},
  volume={58},
  number={4},
  pages={328--346},
  year={1950},
  publisher={The University of Chicago Press}
}

@article{balseiro2019multiagent,
  title={Multiagent mechanism design without money},
  author={Balseiro, Santiago R and Gurkan, Huseyin and Sun, Peng},
  journal={Operations Research},
  volume={67},
  number={5},
  pages={1417--1436},
  year={2019},
  publisher={INFORMS}
}

@inproceedings{banerjee2023robust,
  title={Robust Pseudo-Markets for Reusable Public Resources},
  author={Banerjee, Siddhartha and Fikioris, Giannis and Tardos, Eva},
  booktitle={Proceedings of the 24th ACM Conference on Economics and Computation},
  pages={241--241},
  year={2023}
}

@inproceedings{fikioris2025beyond,
  title={Beyond worst-case online allocation via dynamic max-min fairness},
  author={Fikioris, Giannis and Banerjee, Siddhartha and Tardos, Eva},
  booktitle={Proceedings of the 26th ACM Conference on Economics and Computation},
  pages={92--92},
  year={2025}
}

@article{onyeze2025allocating,
  title={Allocating public goods via dynamic max-min fairness: Long-run behavior and competitive equilibria},
  author={Onyeze, Chido and Banerjee, Siddhartha and Fikioris, Giannis and Tardos, {\'E}va},
  journal={Proceedings of the ACM on Measurement and Analysis of Computing Systems},
  volume={9},
  number={1},
  pages={1--45},
  year={2025},
  publisher={ACM New York, NY, USA}
}

@article{miralles2012cardinal,
  title={Cardinal Bayesian allocation mechanisms without transfers},
  author={Miralles, Antonio},
  journal={Journal of Economic Theory},
  volume={147},
  number={1},
  pages={179--206},
  year={2012},
  publisher={Elsevier}
}

@inproceedings{cole2013positive,
  title={Positive results for mechanism design without money},
  author={Cole, Richard and Gkatzelis, Vasilis and Goel, Gagan},
  booktitle={Proceedings of the 2013 International Conference on Autonomous Agents and Multi-Agent Systems},
  pages={1165--1166},
  year={2013}
}

@inproceedings{guo2010strategy,
  title={Strategy-proof allocation of multiple items between two agents without payments or priors},
  author={Guo, Mingyu and Conitzer, Vincent},
  booktitle={Proceedings of the 9th International Conference on Autonomous Agents and Multiagent Systems},
  pages={881--888},
  year={2010}
}

@inproceedings{han2011strategy,
  title={On strategy-proof allocation without payments or priors},
  author={Han, Li and Su, Chunzhi and Tang, Linpeng and Zhang, Hongyang},
  booktitle={International Workshop on Internet and Network Economics},
  pages={182--193},
  year={2011},
  organization={Springer}
}

@inproceedings{hartline2008optimal,
  title={Optimal mechanism design and money burning},
  author={Hartline, Jason D and Roughgarden, Tim},
  booktitle={Proceedings of the fortieth annual ACM symposium on Theory of computing},
  pages={75--84},
  year={2008}
}

@article{hoppe2009theory,
  title={The theory of assortative matching based on costly signals},
  author={Hoppe, Heidrun C and Moldovanu, Benny and Sela, Aner},
  journal={The Review of Economic Studies},
  volume={76},
  number={1},
  pages={253--281},
  year={2009},
  publisher={Wiley-Blackwell}
}

@article{condorelli2012money,
  title={What money can't buy: Efficient mechanism design with costly signals},
  author={Condorelli, Daniele},
  journal={Games and Economic Behavior},
  volume={75},
  number={2},
  pages={613--624},
  year={2012},
  publisher={Elsevier}
}

@inproceedings{feldman2010online,
  title={Online stochastic packing applied to display ad allocation},
  author={Feldman, Jon and Henzinger, Monika and Korula, Nitish and Mirrokni, Vahab S and Stein, Cliff},
  booktitle={European Symposium on Algorithms},
  pages={182--194},
  year={2010},
  organization={Springer}
}

@article{agrawal2014dynamic,
  title={A dynamic near-optimal algorithm for online linear programming},
  author={Agrawal, Shipra and Wang, Zizhuo and Ye, Yinyu},
  journal={Operations Research},
  volume={62},
  number={4},
  pages={876--890},
  year={2014},
  publisher={INFORMS}
}

@article{jalota2024catch,
  title={Catch Me If You Can: Combatting Fraud in Artificial Currency Based Government Benefits Programs},
  author={Jalota, Devansh and Tsao, Matthew and Pavone, Marco},
  journal={arXiv preprint arXiv:2402.16162},
  year={2024}
}

@article{kesselheim2018primal,
  title={Primal beats dual on online packing LPs in the random-order model},
  author={Kesselheim, Thomas and Radke, Klaus and Tonnis, Andreas and Vocking, Berthold},
  journal={SIAM Journal on Computing},
  volume={47},
  number={5},
  pages={1939--1964},
  year={2018},
  publisher={SIAM}
}

@article{gupta2016experts,
  title={How the experts algorithm can help solve lps online},
  author={Gupta, Anupam and Molinaro, Marco},
  journal={Mathematics of Operations Research},
  volume={41},
  number={4},
  pages={1404--1431},
  year={2016},
  publisher={INFORMS}
}

@article{li2023simple,
  title={Simple and fast algorithm for binary integer and online linear programming},
  author={Li, Xiaocheng and Sun, Chunlin and Ye, Yinyu},
  journal={Mathematical Programming},
  volume={200},
  number={2},
  pages={831--875},
  year={2023},
  publisher={Springer}
}

@article{sun2020near,
  title={Near-optimal primal-dual algorithms for quantity-based network revenue management},
  author={Sun, Rui and Wang, Xinshang and Zhou, Zijie},
  journal={arXiv preprint arXiv:2011.06327},
  year={2020}
}

@article{jiang2020online,
  title={Online resource allocation with stochastic resource consumption},
  author={Jiang, Jiashuo and Zhang, Jiawei},
  journal={arXiv preprint arXiv:2012.07933},
  year={2020}
}

@book{williams1991probability,
  title={Probability with martingales},
  author={Williams, David},
  year={1991},
  publisher={Cambridge university press}
}

@inproceedings{castiglioni2022online,
  title={Online learning with knapsacks: the best of both worlds},
  author={Castiglioni, Matteo and Celli, Andrea and Kroer, Christian},
  booktitle={International Conference on Machine Learning},
  pages={2767--2783},
  year={2022},
  organization={PMLR}
}

@book{bertsekas2003convex,
  title={Convex analysis and optimization},
  author={Bertsekas, Dimitri and Nedic, Angelia and Ozdaglar, Asuman},
  volume={1},
  year={2003},
  publisher={Athena Scientific}
}

@book{munkres2000topology,
  title={Topology},
  author={Munkres, J.R.},
  series={Featured Titles for Topology},
  year={2000},
  publisher={Prentice Hall, Incorporated}
}

@article{amin2014repeated,
  title={Repeated contextual auctions with strategic buyers},
  author={Amin, Kareem and Rostamizadeh, Afshin and Syed, Umar},
  journal={Advances in Neural Information Processing Systems},
  volume={27},
  year={2014}
}

@inproceedings{kanoria2014dynamic,
  title={Dynamic Reserve Prices for Repeated Auctions: Learning from Bids},
  author={Kanoria, Yash and Nazerzadeh, Hamid},
  booktitle={Web and Internet Economics: 10th International Conference},
  volume={8877},
  pages={232},
  year={2014},
  organization={Springer}
}

@article{badanidiyuru2018bandits,
  title={Bandits with knapsacks},
  author={Badanidiyuru, Ashwinkumar and Kleinberg, Robert and Slivkins, Aleksandrs},
  journal={Journal of the ACM (JACM)},
  volume={65},
  number={3},
  pages={1--55},
  year={2018},
  publisher={ACM New York, NY, USA}
}

@inproceedings{sankararaman2018combinatorial,
  title={Combinatorial semi-bandits with knapsacks},
  author={Sankararaman, Karthik Abinav and Slivkins, Aleksandrs},
  booktitle={International Conference on Artificial Intelligence and Statistics},
  pages={1760--1770},
  year={2018},
  organization={PMLR}
}

@article{immorlica2022adversarial,
  title={Adversarial bandits with knapsacks},
  author={Immorlica, Nicole and Sankararaman, Karthik and Schapire, Robert and Slivkins, Aleksandrs},
  journal={Journal of the ACM},
  volume={69},
  number={6},
  pages={1--47},
  year={2022},
  publisher={ACM New York, NY}
}

@article{golrezaei2020no,
  title={No-regret learning in price competitions under consumer reference effects},
  author={Golrezaei, Negin and Jaillet, Patrick and Liang, Jason Cheuk Nam},
  journal={Advances in Neural Information Processing Systems},
  volume={33},
  pages={21416--21427},
  year={2020}
}

@article{galgana2025learning,
  title={Learning in Repeated Multiunit Pay-as-Bid Auctions},
  author={Galgana, Rigel and Golrezaei, Negin},
  journal={Manufacturing \& Service Operations Management},
  volume={27},
  number={1},
  pages={200--229},
  year={2025},
  publisher={INFORMS}
}

@article{watkins1992q,
  title={Q-learning},
  author={Watkins, Christopher JCH and Dayan, Peter},
  journal={Machine Learning},
  volume={8},
  pages={279--292},
  year={1992},
  publisher={Springer}
}

@article{mehrotra2020model,
  title={A model of supply-chain decisions for resource sharing with an application to ventilator allocation to combat COVID-19},
  author={Mehrotra, Sanjay and Rahimian, Hamed and Barah, Masoud and Luo, Fengqiao and Schantz, Karolina},
  journal={Naval Research Logistics (NRL)},
  volume={67},
  number={5},
  pages={303--320},
  year={2020},
  publisher={Wiley Online Library}
}

@article{alban2022resource,
  title={Resource allocation with sigmoidal demands: Mobile healthcare units and service adoption},
  author={Alban, Andres and Blaettchen, Philippe and De Vries, Harwin and Van Wassenhove, Luk N},
  journal={Manufacturing \& Service Operations Management},
  volume={24},
  number={6},
  pages={2944--2961},
  year={2022},
  publisher={INFORMS}
}

@article{perez2022dynamic,
  title={Dynamic resource allocation in the cloud with near-optimal efficiency},
  author={Perez-Salazar, Sebastian and Menache, Ishai and Singh, Mohit and Toriello, Alejandro},
  journal={Operations Research},
  volume={70},
  number={4},
  pages={2517--2537},
  year={2022},
  publisher={INFORMS}
}

@article{chen2023cloud,
  title={Cloud computing value chains: Research from the operations management perspective},
  author={Chen, Shi and Moinzadeh, Kamran and Song, Jing-Sheng and Zhong, Yuan},
  journal={Manufacturing \& Service Operations Management},
  volume={25},
  number={4},
  pages={1338--1356},
  year={2023},
  publisher={INFORMS}
}

@article{benjaafar2020operations,
  title={Operations management in the age of the sharing economy: What is old and what is new?},
  author={Benjaafar, Saif and Hu, Ming},
  journal={Manufacturing \& Service Operations Management},
  volume={22},
  number={1},
  pages={93--101},
  year={2020},
  publisher={INFORMS}
}

@inproceedings{abernethy2008competing,
  title={Competing in the dark: An efficient algorithm for bandit linear optimization},
  author={Abernethy, Jacob and Hazan, Elad and Rakhlin, Alexander},
  booktitle={21st Annual Conference on Learning Theory, COLT 2008},
  year={2008}
}

@article{chen2020minimax,
  title={Minimax regret of switching-constrained online convex optimization: No phase transition},
  author={Chen, Lin and Yu, Qian and Lawrence, Hannah and Karbasi, Amin},
  journal={Advances in Neural Information Processing Systems},
  volume={33},
  pages={3477--3486},
  year={2020}
}

@inproceedings{dai2024refined,
  title={Refined sample complexity for markov games with independent linear function approximation},
  author={Dai, Yan and Cui, Qiwen and Du, Simon S},
  booktitle={The Thirty Seventh Annual Conference on Learning Theory},
  pages={1260--1261},
  year={2024},
  organization={PMLR}
}

@article{galgana2026moneyless,
  title={Moneyless Resource Allocation with Heterogeneous Durations},
  author={Galgana, Rigel and Golrezaei, Negin and Teng, Yifeng},
  journal={Available at SSRN 6222979},
  year={2026}
}

@article{ben2014optimal,
  title={Optimal allocation with costly verification},
  author={Ben-Porath, Elchanan and Dekel, Eddie and Lipman, Barton L},
  journal={American Economic Review},
  volume={104},
  number={12},
  pages={3779--3813},
  year={2014}
}

@article{wilson1987game,
  title={Game-theoretic analyses of trading processes},
  author={Wilson, Robert},
  journal={Advances in Economic Theory},
  pages={33--70},
  year={1987},
  publisher={Cambridge University Press}
}

@article{bergemann2005robust,
  title={Robust mechanism design},
  author={Bergemann, Dirk and Morris, Stephen},
  journal={Econometrica},
  pages={1771--1813},
  year={2005},
  publisher={JSTOR}
}

@book{borgers2015introduction,
  title={An introduction to the theory of mechanism design},
  author={B{\"o}rgers, Tilman},
  year={2015},
  publisher={Oxford university press}
}

@article{engelbrecht1998multi,
  title={Multi-unit auctions with uniform prices},
  author={Engelbrecht-Wiggans, Richard and Kahn, Charles M},
  journal={Economic Theory},
  volume={12},
  number={2},
  pages={227--258},
  year={1998},
  publisher={Springer}
}

@article{ausubel2014demand,
  title={Demand reduction and inefficiency in multi-unit auctions},
  author={Ausubel, Lawrence M and Cramton, Peter and Pycia, Marek and Rostek, Marzena and Weretka, Marek},
  journal={The Review of Economic Studies},
  volume={81},
  number={4},
  pages={1366--1400},
  year={2014},
  publisher={Oxford University Press}
}

@book{wainwright2019high,
  title={High-dimensional statistics: A non-asymptotic viewpoint},
  author={Wainwright, Martin J},
  volume={48},
  year={2019},
  publisher={Cambridge university press}
}

@article{vera2021online,
  title={Online allocation and pricing: Constant regret via bellman inequalities},
  author={Vera, Alberto and Banerjee, Siddhartha and Gurvich, Itai},
  journal={Operations Research},
  volume={69},
  number={3},
  pages={821--840},
  year={2021},
  publisher={INFORMS}
}

@article{banerjee2025good,
  title={Good prophets know when the end is near},
  author={Banerjee, Siddhartha and Freund, Daniel},
  journal={Management Science},
  volume={71},
  number={6},
  pages={4877--4894},
  year={2025},
  publisher={INFORMS}
}

@inproceedings{kleinberg2005multiple,
  title={A multiple-choice secretary algorithm with applications to online auctions},
  author={Kleinberg, Robert},
  booktitle={Proceedings of the Sixteenth Annual ACM-SIAM symposium on Discrete Algorithms},
  pages={630--631},
  year={2005}
}

@inproceedings{babaioff2007matroids,
  title={Matroids, secretary problems, and online mechanisms},
  author={Babaioff, Moshe and Immorlica, Nicole and Kleinberg, Robert},
  booktitle={Symposium on Discrete Algorithms (SODA'07)},
  pages={434--443},
  year={2007}
}

@article{vera2025dynamic,
  title={Dynamic resource allocation: The geometry and robustness of constant regret},
  author={Vera, Alberto and Arlotto, Alessandro and Gurvich, Itai and Levin, Eli},
  journal={Mathematics of Operations Research},
  volume={50},
  number={4},
  pages={2834--2872},
  year={2025},
  publisher={INFORMS}
}

@inproceedings{castiglioni2024online,
  title={Online Learning under Budget and ROI Constraints via Weak Adaptivity},
  author={Castiglioni, Matteo and Celli, Andrea and Kroer, Christian},
  booktitle={International Conference on Machine Learning},
  pages={5792--5816},
  year={2024},
  organization={PMLR}
}

@inproceedings{agrawal2024dynamic,
  title={Dynamic Pricing and Learning with Long-term Reference Effects},
  author={Agrawal, Shipra and Tang, Wei},
  booktitle={Proceedings of the 25th ACM Conference on Economics and Computation},
  pages={72--72},
  year={2024}
}

@incollection{parkes2007online,
  author    = {Parkes, David C.},
  title     = {Online Mechanisms},
  booktitle = {Algorithmic Game Theory},
  editor    = {Nisan, Noam and Roughgarden, Tim and Tardos, {\'E}va and Vazirani, Vijay V.},
  year      = {2007},
  pages     = {411--439},
  publisher = {Cambridge University Press},
  address   = {Cambridge}
}

@inproceedings{deng2023multi,
  title={Multi-channel autobidding with budget and roi constraints},
  author={Deng, Yuan and Golrezaei, Negin and Jaillet, Patrick and Liang, Jason Cheuk Nam and Mirrokni, Vahab},
  booktitle={International Conference on Machine Learning},
  pages={7617--7644},
  year={2023},
  organization={PMLR}
}

@inproceedings{golrezaei2023pricing,
  title={Pricing against a budget and {ROI} constrained buyer},
  author={Golrezaei, Negin and Jaillet, Patrick and Liang, Jason Cheuk Nam and Mirrokni, Vahab},
  booktitle={International Conference on Artificial Intelligence and Statistics},
  pages={9282--9307},
  year={2023},
  organization={PMLR}
}

@article{cesa2014regret,
  title={Regret minimization for reserve prices in second-price auctions},
  author={Cesa-Bianchi, Nicolo and Gentile, Claudio and Mansour, Yishay},
  journal={IEEE Transactions on Information Theory},
  volume={61},
  number={1},
  pages={549--564},
  year={2014},
  publisher={IEEE}
}

@inproceedings{mohri2014learning,
  title={Learning theory and algorithms for revenue optimization in second-price auctions with reserve},
  author={Mohri, Mehryar and Medina, Andres Mu{\~n}oz},
  booktitle={Proceedings of the 31st International Conference on International Conference on Machine Learning-Volume 32},
  pages={I--262},
  year={2014}
}

@inproceedings{braverman2018selling,
  title={Selling to a no-regret buyer},
  author={Braverman, Mark and Mao, Jieming and Schneider, Jon and Weinberg, Matt},
  booktitle={Proceedings of the 2018 ACM Conference on Economics and Computation},
  pages={523--538},
  year={2018}
}

@article{nedelec2022learning,
  title={Learning in repeated auctions},
  author={Nedelec, Thomas and Calauz{\`e}nes, Cl{\'e}ment and Karoui, Noureddine El and Perchet, Vianney},
  journal={Foundations and Trends in Machine Learning},
  volume={15},
  number={3},
  pages={176--334},
  year={2022},
  publisher={Emerald Publishing Limited}
}

@inproceedings{ghuge2025single,
  title={Single-sample and robust online resource allocation},
  author={Ghuge, Rohan and Singla, Sahil and Wang, Yifan},
  booktitle={Proceedings of the 57th Annual ACM Symposium on Theory of Computing},
  pages={1442--1453},
  year={2025}
}

@incollection{devanur2023online,
  author    = {Devanur, Nikhil R. and Mehta, Aranyak},
  title     = {Online Matching in Advertisement Auctions},
  booktitle = {Online and Matching-Based Market Design},
  editor    = {Echenique, Federico and Immorlica, Nicole and Vazirani, Vijay V.},
  publisher = {Cambridge University Press},
  year      = {2023},
  pages     = {130--154},
  chapter   = {6}
}

@incollection{gupta2020random,
  author    = {Gupta, Anupam and Singla, Sahil},
  title     = {Random-Order Models},
  booktitle = {Beyond the Worst-Case Analysis of Algorithms},
  editor    = {Roughgarden, Tim},
  chapter   = {11},
  publisher = {Cambridge University Press},
  year      = {2020},
  pages     = {234--258}
}

@inproceedings{kingma2015adam,
  title={Adam: A Method for Stochastic Optimization},
  author={Kingma, Diederik P. and Ba, Jimmy},
  booktitle={International Conference on Learning Representations},
  year={2015}
}

@inproceedings{feldman2009online,
  title={Online stochastic matching: Beating 1-1/e},
  author={Feldman, Jon and Mehta, Aranyak and Mirrokni, Vahab and Muthukrishnan, Shan},
  booktitle={2009 50th Annual IEEE Symposium on Foundations of Computer Science},
  pages={117--126},
  year={2009},
  organization={IEEE}
}

@article{manshadi2012online,
  title={Online stochastic matching: Online actions based on offline statistics},
  author={Manshadi, Vahideh H and Gharan, Shayan Oveis and Saberi, Amin},
  journal={Mathematics of Operations Research},
  volume={37},
  number={4},
  pages={559--573},
  year={2012},
  publisher={INFORMS}
}

@article{robu2013online,
  title={An online mechanism for multi-unit demand and its application to plug-in hybrid electric vehicle charging},
  author={Robu, Valentin and Gerding, Enrico H and Stein, Sebastian and Parkes, David C and Rogers, Alex and Jennings, Nick R},
  journal={Journal of Artificial Intelligence Research},
  volume={48},
  pages={175--230},
  year={2013}
}

\onecolumn
\appendix
\renewcommand{\appendixpagename}{\centering \LARGE Technical Appendices}
\appendixpage

\startcontents[section]
\printcontents[section]{l}{1}{\setcounter{tocdepth}{2}}

\section{Proof of Main Theorems}\label{sec:appendix main}
\subsection{Guarantee of \mech}\label{sec:main theorem FTRL}
\begin{theorem}[Guarantee of \mech; Restatement of \Cref{thm:main theorem FTRL}]\label{thm:main theorem FTRL formal}
In the \algname framework (\Cref{alg:mech}), let the dual-update subroutine be \ftrl (\Cref{alg:lambda FTRL}) with (assume without loss of generality that $T$ is a perfect cube, since rounding errors only affect the regret by a negligible lower-order term)
\begin{equation*}
L=T^{2/3},~ \mE_\ell=\left [(\ell-1) \frac T L+1,\ell \frac TL\right ],~\eta_\ell=\frac{\lVert \bm \rho^{-1}\rVert_2}{\sqrt{2d}}\left (\sum_{\ell'=1}^{\ell} \lvert \mE_\ell\rvert^2\right )^{-1/2},~\Psi(\bm \lambda)=\frac 12 \lVert \bm \lambda\rVert_2^2.
\end{equation*}
Under \Cref{def:impatient,assump:smooth consumptions}, \mech induces an agents' PBE $\bm \pi^\ast$, with low social welfare regret
\begin{equation*}\begin{aligned}
\mR_T(\bm \pi^\ast,\mech)&\le T^{2/3}\left (4K^2 \epsilon_c+5\log T^{1/3}+\log_{\gamma^{-1}} \big (1+4(1+\lVert \bm \rho^{-1}\rVert_1) K T^{4/3}\big )+\sqrt{2d}\right )\lVert \bm \rho^{-1}\rVert_1\\
&=\Otil \left (T^{2/3} (K^2+\sqrt d )\right )\quad \text{ (when treating $\epsilon_c,\gamma,\lVert \bm \rho^{-1}\rVert$ as constants)},
\end{aligned}\end{equation*}
zero constraint violation $\mB_T(\bm \pi^\ast,\mech)=0$, and a small expected number of ``large misreports'' $\E\left [\sum_{t=1}^T \1\left [\max_{i\in [K]} \lvert u_{t,i}-v_{t,i}\rvert\ge T^{-1/3}\right ]\right ]\le T^{2/3} \left (K \log_{\gamma^{-1}} (1+4(1+\lVert \bm \rho^{-1}\rVert_1) K T^{4/3})+1\right )=\Otil(KT^{2/3})$.
\end{theorem}
\begin{proof}
Recall from \Cref{eq:tilde i definition main} that $\tilde i_t^\ast:=\argmax_{i\in \{0\} \cup [K]} \left (v_{t,i}-\bm \lambda_t^\trans \bm c_{t,i}\right )$ for all $t\in [T]$, and recall from \Cref{eq:regret decomposition informal} that $\mT_v$ is the stopping time defined as $\mT_v:=\min\{t\in [T]\mid \sum_{{\tau}=1}^{t} \bm c_{{\tau},i_{\tau}} + \bm 1\not \le T\bm \rho\}\cup \{T+1\}$.
Thus in all rounds $t\le \mT_v$, \Cref{line:safety} never rejects the suggested allocation $i_t$. This allows us to decompose the regret $\mR_T$:
\begin{equation}\begin{aligned}\label{eq:regret decomposition}
\mR_T=\E\left [\sum_{t=1}^T (v_{t,i_t^\ast}-v_{t,i_t})\right ] \le \E\Bigg [\underbrace{\sum_{t=1}^{\mT_v} (v_{t,\tilde i_t^\ast}-v_{t,i_t})}_{\primalreg}+\underbrace{\sum_{t=1}^T v_{t,i_t^\ast}-\sum_{t=1}^{\mT_v} v_{t,\tilde i_t^\ast}}_{\dualreg}\Bigg ].
\end{aligned}\end{equation}

From \Cref{thm:InterEpoch guarantee} presented later in \Cref{sec:appendix PrimalAlloc}, there exists a PBE $\bm \pi^\ast$ such that
\begin{equation*}\begin{aligned}
\E[\primalreg]\le \sum_{\ell=1}^L (N_\ell+3),
\end{aligned}\end{equation*}
where $N_\ell$ and $M_\ell$ are two quantities frequently used throughout the proof, defined as follows:
\begin{equation}\begin{aligned}\label{eq:N and M}
N_\ell&=1+K\log_{\gamma^{-1}} (1+4(1+\lVert \bm \rho^{-1}\rVert_1)K\lvert \mE_\ell\rvert^4)+4K^2 \epsilon_c +5\log \lvert \mE_\ell\rvert,\\
M_\ell&=\ell \log_{\gamma^{-1}} \left (1+4(1+\lVert \bm \rho^{-1}\rVert_1)K\lvert \mE_\ell\rvert^3\cdot 24 \ell dT\right )+4 \ell \epsilon_c \lVert \bm \rho^{-1}\rVert_1+4\left (\log(dT) + \sum_{j=1}^d \log \frac{d K^2 \epsilon_c T}{\rho_j \epsilon}\right ).
\end{aligned}\end{equation}
When focusing on the polynomial dependencies alone, we have $N_\ell=\Otil(K^2)$ and $M_\ell=\Otil(\ell)$.

From \Cref{thm:DualUpd guarantee FTRL} in \Cref{sec:appendix DualUpd}, when setting $\Psi(\bm \lambda)=\frac 12 \lVert \bm \lambda\rVert_2^2$ and $\eta_\ell=\frac{\lVert \bm \rho^{-1}\rVert_2}{\sqrt{2d}}\left (\sum_{\ell'=1}^{\ell} \lvert \mE_\ell\rvert^2\right )^{-1/2}$ for all $\ell\in [L]$, the $\E[\dualreg]$ term under the same $\bm \pi^\ast$ is bounded by
\begin{equation*}
\E[\dualreg]\le \sqrt{2d}\lVert \bm \rho^{-1}\rVert_2 \cdot \sqrt{\sum_{\ell=1}^L \lvert \mE_\ell\rvert^2}+ \left (1+\max_{1\le \ell\le L} \lvert \mE_\ell\rvert+\sum_{\ell=1}^L (N_\ell+3) \right )\lVert \bm \rho^{-1}\rVert_1.
\end{equation*}

Setting $L=T^{2/3}$ and $\lvert \mE_\ell\rvert=T^{1/3}$ for all $\ell\in [L]$ ensures $\sqrt{\sum_{\ell=1}^L \lvert \mE_\ell\rvert^2}=\sqrt{T^{2/3}\times T^{2/3}}=T^{2/3}$. Therefore
\begin{equation*}\begin{aligned}
\mR_T(\bm \pi^\ast,\mech)\le \sum_{\ell=1}^L (N_\ell+3)+\sqrt{2d}\lVert \bm \rho^{-1}\rVert_2 \cdot T^{2/3}+ \left (1+T^{1/3}+\sum_{\ell=1}^L (N_\ell+3)\right ) \lVert \bm \rho^{-1}\rVert_1.
\end{aligned}\end{equation*}
By definition of $N_\ell$ in \Cref{eq:N and M}, we derive the claimed regret bound of $\mR_T(\bm \pi^\ast,\mech)$  (where we used $\lVert \bm \rho^{-1}\rVert_2\le \lVert \bm \rho^{-1}\rVert_1$).
The constraint violation bound $\mB_T(\bm \pi^\ast,\mech)=0$ follows from \Cref{line:safety} of \Cref{alg:mech}, and the large misreports bound follows from \Cref{lem:large misreport} (where $T^{-1/3}=\frac{1}{T/L}=\frac{1}{\lvert \mE_\ell\rvert}$).
\end{proof}

\subsection{Main Theorem of \mechO}
\begin{theorem}[Main Theorem of \mechO; Restatement of \Cref{thm:main theorem O-FTRL-FP}]\label{thm:main theorem O-FTRL-FP formal}
In the \algname framework (\Cref{alg:mech}), let the dual-update subroutine be \oftrlfp (\Cref{alg:lambda O-FTRL-FP}) with (assume without loss of generality that $T=\sum_{\ell=1}^L \ell$ for some $L$, i.e., $L=\frac{-1+\sqrt{1+8T}}{2}\in \mathbb Z$, since rounding errors are lower-order)
\begin{equation*}
L=\frac{-1+\sqrt{1+8T}}{2},~ \mE_\ell=\left [\sum_{\ell'=1}^{\ell-1} \ell'+1,\sum_{\ell'=1}^{\ell} \ell' \right ],~\eta_\ell=\frac{\lVert \bm \rho^{-1}\rVert_2}{2d}\left (\sum_{\ell'=1}^{\ell} \lvert \mE_\ell\rvert^2\right )^{-1/2},~\Psi(\bm \lambda)=\frac 12 \lVert \bm \lambda\rVert_2^2.
\end{equation*}
Under \Cref{def:impatient,assump:smooth consumptions}, \mechO induces an agents' PBE $\bm \pi^\ast$, with low social welfare regret
\begin{equation*}\begin{aligned}
\mR_T(\bm \pi^\ast,\mechO)
&\le 2 \sqrt{T} \left (3d + 12 \log(dTL) + 12 \sum_{j=1}^d \log \frac{dK^2 \epsilon_c T}{T\rho_j}+18\right )\lVert \bm \rho^{-1}\rVert_1+d\\
&\quad +\sum_{\ell=2}^L \left (\frac{4 N_\ell^2}{\ell}+\frac{16 d M_\ell^2}{\ell (\ell-1)^2}\right ) + \left (3+L+\sum_{\ell=1}^L (N_\ell+3)\right )(1+\lVert \bm \rho^{-1}\rVert_1)\\
&=\Otil\left (\sqrt T \big (d+K^2\big )\right )\quad \text{ (when treating $\epsilon_c,\gamma,\lVert \bm \rho^{-1}\rVert$ as constants)},
\end{aligned}\end{equation*}
zero constraint violation $\mB_T(\bm \pi^\ast,\mechO)=0$, and a small expected number of ``large misreports'' $\E\left [\sum_{t=1}^T \1\left [\max_{i\in [K]} \lvert u_{t,i}-v_{t,i}\rvert\ge t^{-1/2}\right ]\right ]\le T^{1/2} \left (K \log_{\gamma^{-1}} (1+4(1+\lVert \bm \rho^{-1}\rVert_1) K T^2)+1\right )=\Otil(K \sqrt T)$.
\end{theorem}
\begin{proof}
The proof follows the same structure as the previous one, but with a more challenging treatment of $\E[\dualreg]$.
Specifically, we still decompose $\mR_T$ into \primalreg and \dualreg as in \Cref{eq:regret decomposition}, and we still use \Cref{thm:InterEpoch guarantee} from \Cref{sec:appendix PrimalAlloc} for the \primalreg term. It ensures the existence of a PBE $\bm \pi^\ast$ such that $\E[\primalreg]\le \sum_{\ell=1}^L (N_\ell+3)$ with $N_\ell$ defined in \Cref{eq:N and M}.
For the \dualreg term, \Cref{thm:DualUpd guarantee O-FTRL-FP} in \Cref{sec:appendix DualUpd} proves that when $\Psi(\bm \lambda)=\frac 12 \lVert \bm \lambda\rVert_2^2$ and $\eta_\ell=\frac{\lVert \bm \rho^{-1}\rVert_2}{2d}(\sum_{\ell'=1}^\ell \lvert \mE_{\ell'}\rvert)^{-1/2}$ for all $\ell$,
\begin{equation*}\begin{aligned}
\E[\dualreg]&\le 2 \sqrt{T} \left (3d + 12 \log(dTL) + 12 \sum_{j=1}^d \log \frac{dK^2 \epsilon_c T}{T\rho_j}+18\right )\lVert \bm \rho^{-1}\rVert_1+d \lvert \mE_1\rvert^{1.5} \\
&\quad +\sum_{\ell=2}^L \frac{N_\ell^2/d+\lvert \mE_\ell\rvert^2 M_\ell^2/(\sum_{\ell'=1}^{\ell-1} \lvert \mE_{\ell'}\rvert)^2}{(\sum_{\ell'=1}^\ell \lvert \mE_{\ell'}\rvert)^{1/2}} \lVert \bm \rho^{-1}\rVert_1 + \left (3+\max_{\ell\in [L]} \lvert \mE_\ell\rvert+\sum_{\ell=1}^L (N_\ell+3)\right )\lVert \bm \rho^{-1}\rVert_1,
\end{aligned}\end{equation*}
where $N_\ell$ and $M_\ell$ are defined in \Cref{eq:N and M} for all $\ell\in [L]$. Therefore, under our linear epoching scheme where $\lvert \mE_\ell\rvert=\ell$ for all $\ell$, we have $\lvert \mE_1\rvert^{1.5}=1$, $(\sum_{\ell'=1}^\ell \lvert \mE_{\ell'}\rvert)^{-1/2}\ge 2/\ell$, $\lvert \mE_\ell\rvert / (\sum_{\ell'=1}^{\ell-1}\lvert \mE_{\ell'}\rvert)=2/(\ell-1)$, and $\max_{\ell\in [L]} \lvert \mE_\ell\rvert=L$.
Putting the primal and dual bounds together, we therefore get
\begin{equation*}\begin{aligned}
&\quad \mR_T(\bm \pi^\ast,\mechO)\le \E[\primalreg]+\E[\dualreg]\\
&\le 2 \sqrt{T} \left (3d + 12 \log(dTL) + 12 \sum_{j=1}^d \log \frac{dK^2 \epsilon_c T}{T\rho_j}+18\right )\lVert \bm \rho^{-1}\rVert_1+d\\
&\quad +\sum_{\ell=2}^L \left (\frac{4 N_\ell^2}{\ell}+\frac{16 d M_\ell^2}{\ell (\ell-1)^2}\right ) + \left (3+L+\sum_{\ell=1}^L (N_\ell+3)\right )(1+\lVert \bm \rho^{-1}\rVert_1)=\Otil\left (\sqrt T \big (d+K^2\big )\right ),
\end{aligned}\end{equation*}
where the last step uses $L=\Theta(\sqrt T)$, $N_\ell=\Otil(K^2)$, $M_\ell=\Otil(\ell)$, and $\sum_{\ell=1}^L \ell^{-1}=\Theta(\log L)$.
The constraint violation bound $\mB_T(\bm \pi^\ast,\mech)=0$ follows from \Cref{line:safety} of \Cref{alg:mech}, and the large misreports bound follows from \Cref{lem:large misreport}: for each epoch $\mE_\ell$, the rounds $t\in \mE_\ell$ satisfy $t^{-1/2}\ge \ell^{-1}$.
\end{proof}

\subsection{O-FTRL-FP Algorithm for Endogenous and Discontinuous Predictions}\label{sec:O-FTRL lemma}
\begin{definition}[\oftrlfp Algorithm; Generalized \Cref{alg:lambda O-FTRL-FP}]\label{def:O-FTRL-FP}
Consider a $R$-round online learning over convex and compact $\mathcal X\subseteq \mathbb R^d$ with (unknown) convex losses $\{f_r\colon \mathcal X \to \mathbb{R}\}_{r\in [R]}$.
We have access to a \emph{prediction} $\tilde f_r\colon \mathcal X\times \mathcal X\to \mathbb R$ before each round $r$, such that should we play $\bm x_r\in \mathcal X$ in round $r$, the true loss $f_r(\cdot)$ would be close to $\tilde f_r(\bm x_r,\cdot)$.
For each round $r$, define the \textit{ideal update map} $\Phi_r\colon \mathcal X\to \mathcal X$ as:
\begin{equation}\label{eq:ideal update map}
\Phi_r(\bm y):=\argmin_{\bm x\in \mathcal X} \left( \sum_{\mathfrak r < r} f_{\mathfrak r}(\bm x) + \tilde f_r(\bm y, \bm x) + \frac{\Psi(\bm x)}{\eta_r} \right),
\end{equation}
where $\Psi\colon \mathcal X\to \mathbb R$ is a strongly-convex regularizer.
In each round, \oftrlfp determines its decision $\bm x_r\in \mathcal X$ via an \textit{approximate fixed-point oracle} for $\Phi_r$, such that (where $\delta_r$ is called the \textit{fixed-point residual})
\begin{equation}\label{eq:approx fixed point}
\left \lVert \bm x_r - \Phi_r(\bm x_r) \right \rVert \le \delta_r.
\end{equation}
\end{definition}

We remark that while our \oftrlfp is inspired by the \texttt{FTRL-FB} algorithm \citep{zimmert2022return}---which also uses fixed-point oracles to conduct implicit updates---we additionally resolve the extra difficulty of \textit{discontinuous predictions}. Consequently, an exact fixed point of $\Phi_r$ may not exist. Instead, as we only require a weaker gradient-stability condition (recall \Cref{def:gradient stability}) to ensure the existence of approximate fixed points.
We state the \oftrlfp guarantee in \Cref{lem:O-FTRL lemma}, which can be of independent interest.
\begin{theorem}[\oftrlfp Guarantee; Generalized \Cref{eq:O-FTRL-FP online learning regret}]\label{lem:O-FTRL lemma}
Consider a convex and compact set $\mathcal X\subseteq \mathbb R^d$ equipped with some norm $\lVert \cdot \rVert$. Let $\Psi\colon \mathcal X\to \mathbb R$ be $1$-strongly-convex w.r.t. $\lVert \cdot \rVert$ such that $\min_{\bm x\in \mathcal X} \Psi(\bm x)=0$.
Let $\{\tilde f_r\colon \mathcal X\times \mathcal X\to \mathbb R\}_{r\in [R]}$ be a sequence of prediction functions satisfying the following conditions:
\begin{enumerate}
\item \textbf{Measurability:} $\tilde f_r$ is $\mF_{r-1}$-measurable, where $\mF_r=\sigma(f_1,\ldots,f_r)$ is the filtration induced by losses;
\item \textbf{Regularity:} For any fixed $\bm y\in \mathcal X$, $\tilde f_r(\bm y,\cdot)\colon \mathcal X\to \mathbb R$ is convex, differentiable, and $L_f$-smooth; and
\item \textbf{Gradient-Stability:} Each prediction function $\tilde f_r$ is $(\epsilon_r,L_r)$-gradient-stable (\Cref{def:gradient stability}), i.e.,
\begin{equation*}
\sup_{\bm x\in \mathcal X}\left \lVert \nabla_2 \tilde f_r(\bm y_1,\bm x)-\nabla_2 \tilde f_r(\bm y_2,\bm x)\right \rVert_\ast\le L_r,\quad \forall \bm y_1,\bm y_2\in \mathcal X\text{ s.t. }\lVert \bm y_1-\bm y_2\rVert<\epsilon_r,
\end{equation*}
where $\nabla_2$ is the gradient with respect to the second argument, and $\lVert \cdot \rVert_\ast$ is the dual norm of $\lVert \cdot \rVert$.
\end{enumerate}
Then for any continuous, differentiable, convex, $L$-Lipschitz, and $L_f$-smooth loss functions $f_1,f_2,\ldots,f_R$ and any non-increasing learning rates $\eta_1\ge \eta_2\ge \cdots \ge \eta_R>0$, the \oftrlfp in \Cref{def:O-FTRL-FP} ensures:
\begin{itemize}
\item \textbf{Existence of Approximate Fixed-Points.} For every round $r\in [R]$, there exists a valid update $\bm x_r \in \mathcal X$ such that $\lVert \bm x_r-\Phi_r(\bm x_r)\rVert\le \eta_r L_r$ (i.e., satisfying \Cref{eq:approx fixed point} with a residual $\delta_r$ upper bounded by $\eta_r L_r$).
\item \textbf{Regret Bound.} The resulting sequence $\{\bm x_r\}_{r\in [R]}$ satisfies the following regret bound for any $\bm x^\ast \in \mathcal X$:
\begin{equation*}
\sum_{r=1}^R \left (f_r(\bm x_r)-f_r(\bm x^\ast)\right )\le \frac{\Psi(\bm x^\ast)}{\eta_R} + \frac 12\sum_{r=1}^R \eta_r \left \lVert \nabla f_r(\bm x_r)-\nabla_2 \tilde f_r(\bm x_r, \bm x_r)\right \rVert_\ast^2 + \sum_{r=1}^R \left (L\eta_r L_r+2L_f^2\eta_r^2 L_r^2\right ).
\end{equation*}
\end{itemize}
\end{theorem}
\begin{proof}
Fix a round $r\in [R]$. We first argue that the $\Phi_r$ in \Cref{eq:ideal update map} is well-defined: Since $\Psi(\bm x)$ is $1$-strongly convex and all loss functions are convex, the \ftrl objective $F_r(\bm x) := \sum_{\mathfrak r=1}^{r-1} f_{\mathfrak r}(\bm x) + \frac{1}{\eta_r} \Psi(\bm x)$ is $\eta_r^{-1}$-strongly convex w.r.t. $\lVert \cdot \rVert$.
For any fixed $\bm y \in \mathcal X$, since $\tilde f_r(\bm y,\cdot)$ is also convex, $\Phi_r(\bm y)=\argmin_{\bm x\in \mathcal X}(F_r(\cdot) + \tilde f_r(\bm y, \cdot))$ is the $\argmin$ of a strongly convex function, and thus well-defined and unique.

\paragraph{Step 1: Approximate-Continuity of $\Phi_r$.}
Fix $\bm y_1, \bm y_2 \in \mathcal X$ with $\lVert \bm y_1 - \bm y_2 \rVert < \epsilon_r$. Let $\bm x_1 = \Phi_r(\bm y_1)$ and $\bm x_2 = \Phi_r(\bm y_2)$. The first-order optimality conditions of the $\argmin$'s imply:
\begin{equation*}
\begin{aligned}
\langle \nabla F(\bm x_1) + \nabla_2 \tilde f_r(\bm y_1, \bm x_1), \bm x - \bm x_1 \rangle \ge 0, \quad \forall \bm x \in \mathcal X;\\
\langle \nabla F(\bm x_2) + \nabla_2 \tilde f_r(\bm y_2, \bm x_2), \bm x - \bm x_2 \rangle \ge 0, \quad \forall \bm x \in \mathcal X.\end{aligned}
\end{equation*}
Summing the optimality conditions at $\bm x = \bm x_2$ (for $\bm x_1$) and $\bm x = \bm x_1$ (for $\bm x_2$), and utilizing the $\eta_r^{-1}$-strong convexity of $F(\bm x) + \tilde f_r(\bm y, \bm x)$ \citep[Exercise 1.10]{bertsekas2003convex}, we obtain stability bounds:
\begin{equation*}\begin{aligned}
\eta_r^{-1} \lVert \bm x_1 - \bm x_2 \rVert^2
&\le \langle \nabla_2 \tilde f_r(\bm y_1, \bm x_2) - \nabla_2 \tilde f_r(\bm y_2, \bm x_2), \bm x_2 - \bm x_1 \rangle \\
&\le \lVert \bm x_1 - \bm x_2 \rVert \cdot \lVert \nabla_2 \tilde f_r(\bm y_1, \bm x_2) - \nabla_2 \tilde f_r(\bm y_2, \bm x_2) \rVert_\ast,
\end{aligned}\end{equation*}
where the last step uses the Cauchy-Schwarz inequality.
Applying the Gradient Stability condition (Condition 3), we have $\lVert \nabla_2 \tilde f_r(\bm y_1, \bm x_2) - \nabla_2 \tilde f_r(\bm y_2, \bm x_2) \rVert_\ast \le L_r$. Consequently:
\begin{equation}\label{eq:closeness of argmin}
\lVert \Phi_r(\bm y_1) - \Phi_r(\bm y_2) \rVert \le \eta_r L_r, \quad \forall \lVert \bm y_1 - \bm y_2 \rVert < \epsilon_r.
\end{equation}

\paragraph{Step 2: Continuous Approximation of $\Phi_r$.}
Since $\mathcal X$ is a compact subset of a finite-dimensional normed vector space, it admits a finite cover. Let $\{\mathcal B_i\}_{i=1}^M$ be a finite collection of open balls of radius $\frac{\epsilon_r}{2}$ centered at points $\{\tilde{\bm y}_i \in \mathcal X\}_{i=1}^M$ that cover $\mathcal X$. Let $\{\phi_i\}_{i=1}^M$ be a continuous partition of unity subordinate to this cover \citep[Theorem 36.1]{munkres2000topology}, satisfying $\sum_{i=1}^M \phi_i(\bm y) = 1$ and $\phi_i(\bm y) = 0$ if $\bm y \notin \mathcal B_i$.
Define $\tilde{\Phi_r}: \mathcal X \to \mathcal X$ as
\begin{equation*}
\tilde{\Phi_r}(\bm y) := \sum_{i=1}^M \phi_i(\bm y) \Phi_r(\tilde{\bm y}_i),\quad \forall \bm y\in \mathcal X,
\end{equation*}
which is a smoothened version of $\Phi_r$.
Since $\mathcal X$ is convex and $\Phi_r(\tilde{\bm y}_i) \in \mathcal X$, the convex combination $\tilde{\Phi_r}(\bm y)$ remains in $\mathcal X$. Moreover, $\tilde{\Phi_r}$ is continuous as it is a linear combination of continuous weight functions. By Brouwer's Fixed Point Theorem \citep[Theorem 55.6]{munkres2000topology}, there exists $\bm y^\ast \in \mathcal X$ such that $\tilde{\Phi_r}(\bm y^\ast) = \bm y^\ast$.

We show $\bm y^\ast$ is an approximate fixed point of $\Phi_r$. Note that for any $i$ where $\phi_i(\bm y^\ast) > 0$, we must have $\bm y^\ast \in \mathcal B_i$, implying $\lVert \bm y^\ast - \tilde{\bm y}_i \rVert < \frac{\epsilon_r}{2} + \frac{\epsilon_r}{2} = \epsilon_r$. By \Cref{eq:closeness of argmin}, $\lVert \Phi_r(\bm y^\ast) - \Phi_r(\tilde{\bm y}_i) \rVert \le \eta_r L_r$. Thus:
\begin{equation}\label{eq:approximate fixed point}
\lVert \bm y^\ast - \Phi_r(\bm y^\ast) \rVert
= \left\lVert \sum_{i=1}^M \phi_i(\bm y^\ast) (\Phi_r(\tilde{\bm y}_i) - \Phi_r(\bm y^\ast)) \right\rVert
\le \sum_{i=1}^M \phi_i(\bm y^\ast) \cdot \eta_r L_r = \eta_r L_r.
\end{equation}

\paragraph{Step 3: Regret Analysis.}
With $\bm x_r$ being the $\bm y^\ast$ derived in \Cref{eq:approximate fixed point}, let $\hat f_r(\cdot) := \tilde f_r(\bm x_r, \cdot)$. Since $\tilde f_r$ is $\mF_{r-1}$-measurable and $\bm x_r$ is derived deterministically from history, $\hat f_r$ serves as a valid prediction for the \oftrl algorithm \citep{rakhlin2013online}. Let $\bm x_r^\ast = \Phi_r(\bm x_r)$ be the exact minimizer given prediction $\bm x_r$.
Applying the standard \oftrl guarantee \citep[see, e.g.,][Theorem 7.39]{orabona2019modern} to the sequence $\bm x_r^\ast$:
\begin{equation*}
\sum_{r=1}^R (f_r(\bm x_r^\ast) - f_r(\bm x^\ast)) \le \frac{\Psi(\bm x^\ast)}{\eta_R} + \frac{1}{2} \sum_{r=1}^R \eta_r \lVert \nabla f_r(\bm x_r^\ast) - \nabla \hat f_r(\bm x_r^\ast) \rVert_\ast^2.
\end{equation*}

It only remains to bridge the gap between $\bm x_r$ and $\bm x_r^\ast$ w.r.t. $f_r$ and $\hat f_r$. We already showed $\lVert \bm x_r - \bm x_r^\ast \rVert \le \eta_r L_r$. Using Lipschitzness ($L$) and smoothness ($L_f$) of the relevant functions, $f_r(\bm x_r) - f_r(\bm x_r^\ast) \le L \lVert \bm x_r - \bm x_r^\ast \rVert \le L \eta_r L_r$, $\lVert \nabla f_r(\bm x_r) - \nabla f_r(\bm x_r^\ast) \rVert_\ast \le L_f \eta_r L_r$, and $\lVert \nabla \hat f_r(\bm x_r) - \nabla \hat f_r(\bm x_r^\ast) \rVert_\ast \le L_f \eta_r L_r$.
Substituting these into the \oftrl regret bound yields the guarantee of \oftrlfp.
\end{proof}

\section{\primalreg: Misallocations Due to Agents' Strategic Behavior}\label{sec:appendix PrimalAlloc}
\subsection{Direct Impact: Misreporting for Current-Epoch Allocations}\label{sec:appendix IntraEpoch}

\begin{lemma}[Truthful Myopic Agents]\label{lem:2nd price auction variant}
Suppose that agents are myopic. That is, consider a one-shot allocation with $K$ agents, where each agent $i \in [K]$ has a private value $v_i $ and public consumption $c_i=\bm \lambda_t^\trans \bm c_{t,i}$. Agents submit reports $u_i \in [0,1]$, and the planner allocates the item to one agent $i_t \in [K]$ at payment $p_{i_t}$. The utility of the selected agent is $v_{i_t} - p_{i_t}$, while all others have zero utility. Then the following mechanism
\begin{equation*}
i_t = \argmax_{i \in \{0\} \cup [K]} (u_i - c_i), \quad p_{i_t} = c_{i_t} + \max_{j \ne i_t} (u_j - c_j),
\end{equation*}
is Dominant-Strategy Incentive-Compatible (DSIC) and efficient. Specifically,
\begin{enumerate}
\item to maximize the expected utility $\E[(v_i-p_i) \cdot \1[i_t = i]]$, truthfully reporting $u_i = v_i$ is weakly dominant;
\item should all agents be truthful ($u_i = v_i$), the allocation maximizes the value-minus-consumption $v_{i_t} - c_{i_t}$.
\end{enumerate}
\end{lemma}
\begin{proof}
Let $\tilde{u}_i := u_i - c_i$ be each agent's report-minus-consumption, and let $\tilde{u}_i^\ast := v_i - c_i$ be the truthful report-minus-consumption.
Fix any agent $i \in [K]$ and suppose all other agents’ reports $\{u_j\}_{j \ne i}$ are fixed. We evaluate the utility $\mathcal{U}_i=(v_i-p_i) \1[i_t=i]$ that agent $i$ obtains under various reporting strategies:
\begin{enumerate}
\item Truthfully reporting $u_i = v_i$ so that $\tilde{u}_i = \tilde{u}_i^\ast$. If $\tilde{u}_i^\ast > \max_{j \ne i} \tilde{u}_j$, then $i_t = i$, and $\mathcal{U}_i = v_i - c_i - \max_{j \ne i} \tilde{u}_j = \tilde{u}_i^\ast - \max_{j \ne i} \tilde{u}_j$. Otherwise, i.e., when $\tilde{u}_i^\ast < \max_{j \ne i} \tilde{u}_j$, we have $i_t \ne i$ and $\mathcal{U}_i = 0$.
\item Over-reporting $u_i > v_i$ so that $\tilde{u}_i > \tilde{u}_i^\ast$. If either $\tilde u_i>\tilde u_i^\ast>\max_{j\ne i}\tilde u_j$ or $\max_{j\ne i} \tilde u_j>\tilde u_i>\tilde u_i^\ast$, agent $i$ secures the same utility as if they were truthful (in the former case $\mathcal U_i=\tilde u_i^\ast-\max_{j\ne i}\tilde u_j$, and in the latter case $\mathcal U_i=0$). But if $\tilde u_i>\max_{j\ne i} \tilde u_j>\tilde u_i^\ast$, $\mathcal U_i=\tilde u_i^\ast-\max_{j\ne i} \tilde u_j<0$, which is worse than being truthful.
\item Under-reporting $u_i < v_i$ so that $\tilde{u}_i < \tilde{u}_i^\ast$. Similarly, if $\tilde u_i^\ast>\tilde u_i>\max_{j\ne i}\tilde u_j$ or $\max_{j\ne i}\tilde u_j>\tilde u_i^\ast>\tilde u_i$, agent $i$ secures the same utility as if they were truthful. But when $\tilde u_i^\ast>\max_{j\ne i}\tilde u_j>\tilde u_i$, $\mathcal U_i=0$, which always never exceeds their utility should they reported truthfully: in which case they get $\tilde u_i^\ast-\max_{j\ne i} \tilde u_j\ge 0$.
\end{enumerate}

In all cases, deviating from truth-telling does not improve agent $i$'s utility, and may strictly reduce it. Hence, truthful reporting is a weakly dominant strategy, proving claim \textit{(i)}. When all agents report truthfully ($u_i = v_i$), the planner allocates the resource to $i_t = \arg\max_i (v_i - c_i)$. This gives claim \textit{(ii)}.
\end{proof}

\begin{theorem}[Direct Impact; Formal \Cref{thm:informal IntraEpoch}]\label{thm:IntraEpoch guarantee}
Fix an epoch $\mathcal{E}_\ell \subseteq [T]$ and variable dual $\bm{\lambda}_\ell \in \bm{\Lambda}$. Suppose that all agents adopt the nearsighted model of \Cref{eq:two agent models}, and the planner uses the dual-adjusted allocation and payment rule in \Cref{line:primal allocation}. Then truthtul reporting constitutes a Perfect Bayesian Equilibrium (PBE). Under this PBE, the planner always allocates to the agent maximizing the dual-adjusted value, namely in each round $t\in \mE_\ell$ the allocated agent is $\tilde{i}_t^\ast := \arg\max_{i \in [K]} \left( v_{t,i} - \bm{\lambda}_\ell^\top \bm{c}_{t,i} \right)$.
\end{theorem}
\begin{proof}
For any round $t\in \mE_\ell$, the planner's allocation and payment is ``history-independent,'' in the sense that they only depend on current reports $\bm{u}_t^h$, current consumptions $\bm{c}_t$, and dual variable $\bm{\lambda}_\ell$.
We now consider a fixed agent $i\in [K]$. Since all opponents $j\ne i$ are reporting $u_{t,j}=v_{t,j}$, their reports are also ``history-independent''. Therefore, agent $i$'s expected utility under any strategy only depends on the current round but not the history, and the unilateral deviation to a ``history-dependent'' strategy is unprofitable.

Hence we only need to consider agent $i$'s unilateral deviation to history-independent strategies. This isolates each round $t\in \mE_\ell$, allowing us to apply \Cref{lem:2nd price auction variant}: Given dual $\bm{\lambda}_\ell$, truthful reporting maximizes an agent's current-round expected utility regardless of opponents' actions. Therefore agent $i$'s unilateral deviation from truth-telling is unprofitable, and truthful reporting is a PBE.
Under this PBE, the planner allocates to the agent with maximal $v_{t,i} - \bm{\lambda}_\ell^\trans \bm{c}_{t,i}$, i.e., $\tilde{i}_t^\ast$. This completes the proof.
\end{proof}

\subsection{Indirect Impact: Misreporting for Future Dual Variables}\label{sec:appendix InterEpoch}

\begin{lemma}[Large Misreport Happens Rarely]\label{lem:large misreport}
For epoch $\ell\in [L]$ with fixed dual $\bm \lambda_\ell$, for $\gamma$-discounted agents (the ``farsighted'' model in \Cref{eq:two agent models}), there exists a PBE $\bm \pi$ for this epoch such that
\begin{equation}\label{eq:large misreport}
\Pr\left \{\sum_{t\in \mE_\ell} \1\left [\lvert u_{t,i}-v_{t,i}\rvert\ge \frac{1}{\lvert \mE_\ell\rvert}\right ]\le \log_{\gamma^{-1}} (1+4(1+\lVert \bm \rho^{-1}\rVert_1)K\lvert \mE_\ell\rvert^4)\right \}\ge 1-\frac{1}{\lvert \mE_\ell\rvert},\quad \forall i\in [K],
\end{equation}
where $\{u_{t,i}\}_{t\in \mE_\ell}$ are the reports made by agent $i$ under $\bm \pi$.
\end{lemma}
\begin{proof}
As in \Cref{thm:IntraEpoch guarantee}, we restrict our attention to the subspace of ``history-independent'' strategies where for each round $t\in \mE_\ell$, agent $i\in [K]$ can only decide their report $u_{t,i}$ based on current-epoch dual variable $\bm \lambda_\ell$, round number $t$, current-round value $v_{t,i}$, and current-round consumptions $\bm c_t$.

We claim there is a history-independent PBE $\bm \pi$. Indeed, consider a game where only history-independent strategies are allowed, and consider a PBE $\bm \pi$ therein. Using same arguments from \Cref{thm:IntraEpoch guarantee}, unilaterally deviating to a history-dependent strategy is unprofitable as all opponents' strategies $\bm \pi_{-i}$ and the mechanism are history-independent. Hence, $\bm \pi$ remains a PBE in the original game where history dependency is allowed.

To prove \Cref{eq:large misreport} for this PBE $\bm \pi$, consider the unilateral deviation of any agent $i\in [K]$ to the truth-telling policy, i.e., $\bm \pi^i:=\truth_i\circ \bm \pi_{-i}$. Since $\bm \pi$ is a PBE, $\bm \pi^i$ is no better than $\bm \pi$ for a $\gamma$-impatient agent. That is, for $s_\ell:=\min\{t\mid t\in \mE_\ell\}$ and any history $\mH_{s_\ell}$, we always have
\begin{equation}\begin{aligned}
0&\le \E_{\bm \pi}\left [\sum_{\tau=t}^T \gamma^{\tau} (v_{{\tau},i}-p_{{\tau},i}) \1[i_{\tau}=i]\middle \vert \mH_{s_\ell}\right ]-\E_{\bm \pi^i}\left [\sum_{\tau=t}^T \gamma^{\tau} (v_{{\tau},i}-p_{{\tau},i}) \1[i_{\tau}=i]\middle \vert \mH_{s_\ell}\right ]  \\
&=\underbrace{\E_{\bm \pi}\left [\sum_{{\tau} \in \mE_\ell} \gamma^{\tau} (v_{{\tau},i}-p_{{\tau},i}) \1[i_{\tau}=i]\middle \vert \mH_{s_\ell}\right ]-\E_{\bm \pi^i}\left [\sum_{{\tau} \in \mE_\ell} \gamma^{\tau} (v_{{\tau},i}-p_{{\tau},i}) \1[i_{\tau}=i]\middle \vert \mH_{s_\ell}\right ]}_{\text{Current-Epoch}}+ \\
&\quad \underbrace{\E_{\bm \pi}\left [\sum_{\ell'>\ell,{\tau}\in \mE_{\ell'}} \gamma^{\tau} (v_{{\tau},i}-p_{{\tau},i}) \1[i_{\tau}=i]\middle \vert \mH_{s_\ell}\right ]-\E_{\bm \pi^i}\left [\sum_{\ell'>\ell,{\tau}\in \mE_{\ell'}} \gamma^{\tau} (v_{{\tau},i}-p_{{\tau},i}) \1[i_{\tau}=i]\middle \vert \mH_{s_\ell}\right ]}_{\text{Future-Epoch}}.\label{eq:large misreports V difference}
\end{aligned}\end{equation}

For the \text{Future-Epoch} term, fix any $\ell'>\ell$ and $\tau\in \mE_{\ell'}$. Since the values $\bm v_\tau$, reports $\bm u_\tau$, and consumptions $\bm c_\tau$ are all bounded by $[0,1]$, and that $\bm \lambda_{\ell'}\in \bm \Lambda=\bigotimes_{j=1}^d [0,\rho_j^{-1}]$ (which infers $\lVert \bm \lambda_{\ell'}\rVert_1\le \lVert \bm \rho^{-1}\rVert_1$), we have
\begin{equation*}\begin{aligned}
\scalemath{0.95}{p_{\tau,i_\tau}=\bm \lambda_{\ell'}^\trans \bm c_{\tau,i_\tau}+\max_{j\ne i_\tau}(u_{\tau,j}-\bm \lambda_{\ell'}^\trans \bm c_{\tau,j})\in \left [-2\lVert \bm \lambda_{\ell'}\rVert_1\cdot \max_{j\in \{0\} \cup [K]}\lVert \bm c_{\tau,j}\rVert_\infty,1+2\lVert \bm \lambda_{\ell'}\rVert_1\cdot \max_{j\in \{0\} \cup [K]}\lVert \bm c_{\tau,j}\rVert_\infty\right ],}
\end{aligned}\end{equation*}
for all $\ell'>\ell,\tau\in \mE_{\ell'}$. Since $v_{t,i}\in [0,1]$, $v_{\tau,i}-p_{\tau,i}\in [-1-2\lVert \bm \rho^{-1}\rVert_1,1+2\lVert \bm \rho^{-1}\rVert_1]$ for all $i$, and thus
\begin{equation*}
\text{Future-Epoch}\le \sum_{\ell'=\ell+1}^L \sum_{{\tau}\in \mE_{\ell'}} \gamma^{{\tau}}\cdot 2(1+2\lVert \bm \rho^{-1}\rVert_1) \le 2(1+2\lVert \bm \rho^{-1}\rVert_1) \frac{\gamma^{s_{\ell+1}}}{1-\gamma},
\end{equation*}
where the second inequality uses the fact that $\sum_{\tau=s_{\ell+1}}^T \gamma^\tau \le \frac{\gamma^{s_{\ell+1}}}{1-\gamma}$.

We now focus on Current-Epoch. Since both $\bm \pi$ and $\bm \pi^i$ are history-independent, the conditional on $\mH_{s_\ell}$ is redundant.
For any round $t\in \mE_\ell$, we control the difference between $\bm \pi$ and $\bm \pi^i$ using the exploration rounds:
\begin{itemize}
\item Conditional on the round $t\in \mE_\ell$ being an exploration round for this specific agent $i$ (which happens with probability $\frac{1}{\lvert \mE_\ell\rvert K}$), the expected gain of reporting $u_{t,i}$ rather than $v_{t,i}$ (i.e., under $\bm \pi$ versus under $\bm \pi^i$) is
\begin{equation}\label{eq:exploration misreport loss}
\E_{p\in \Unif([0,1])}\left [(v_{t,i}-p) \1[u_{t,i}\ge p]-(v_{t,i}-p) \1[v_{t,i}\ge p]\right ]=-\frac 12(u_{t,i}-v_{t,i})^2.
\end{equation}
\item If round $t$ is an exploration round but not for agent $i$, then reporting $u_{t,i}$ and $v_{t,i}$ both give $0$ gain.
\item Finally, suppose that round $t\in \mE_\ell$ is not an exploration round. Notice that \textit{i)} both $\bm \pi$ and $\bm \pi^{i}$ are history-independent, and \textit{ii)} if round $t\in \mE_\ell$ is not an exploration round, then the round-$t$ game in Lines \ref{line:call subroutine} to \ref{line:end of epoch game} in \Cref{alg:mech} coincides with the one studied in \Cref{lem:2nd price auction variant}. Therefore, agent $i$'s gain of reporting $u_{t,i}$ is no larger than that of reporting $v_{t,i}$, i.e., the expectation difference is always non-positive.
\end{itemize}

Let $\mathcal M_{\ell,i}:=\left \{t\in \mE_\ell\mid \lvert u_{t,i}-v_{t,i}\rvert\ge \frac{1}{\lvert \mE_\ell}\rvert\right \}$ be the set of large misreports from agent $i$ (regardless of whether $t$ turns out to be an exploration round, since reports happen before it). Let the last round in epoch $\mE_\ell$ be $e_\ell:=\max\{t\mid t\in \mE_\ell\}$. Since $\sum_{t\in \mathcal M_{\ell,i}}\gamma^t\ge \sum_{t=e_\ell-\lvert \mathcal M_{\ell,i}\rvert+1}^{e_\ell} \gamma^t$, we have
\begin{equation*}\begin{aligned}
\text{Current-Epoch}&\le -\E\left [\sum_{t\in \mathcal M_{\ell,i}} \gamma^t \frac{1}{\lvert \mE_\ell\rvert \cdot K}\frac 12(u_{t,i}-v_{t,i})^2\right ]\le -\E\left [\frac{\gamma^{e_\ell-\lvert \mathcal M_{\ell,i}\rvert+1} (1-\gamma^{\lvert \mathcal M_{\ell,i}\rvert})}{1-\gamma} \frac{1}{\lvert \mE_\ell\rvert\cdot K} \frac{1}{2\lvert \mE_\ell\rvert^2}\right ].
\end{aligned}\end{equation*}

In order for $\bm \pi^i$ to be inferior when compared with $\bm \pi$, i.e., \Cref{eq:large misreports V difference} holds. Hence,
\begin{equation}\begin{aligned}
2(1+2\lVert \bm \rho^{-1}\rVert_1)\frac{\gamma^{s_{\ell+1}}}{1-\gamma}&\ge \E\left [\frac{\gamma^{e_\ell-\lvert \mathcal M_{\ell,i}\rvert+1} (1-\gamma^{\lvert \mathcal M_{\ell,i}\rvert})}{1-\gamma} \frac{1}{\lvert \mE_\ell\rvert\cdot K} \frac{1}{2\lvert \mE_\ell\rvert^2}\right ] \\
&\ge \Pr\{\lvert \mathcal M_{\ell,i}\rvert\ge c\}\frac{\gamma^{s_{\ell+1}}}{1-\gamma} \frac{\gamma^{-c}-1}{2K\lvert \mE_\ell\rvert^3},\quad \forall c>0, \label{eq:large misreports takeaway}
\end{aligned}\end{equation}
where the last step uses the fact that $s_{\ell+1}=e_{\ell}+1$.
Picking $c$ such that $2(1+2\lVert \bm \rho^{-1}\rVert_1)\frac{2K\lvert \mE_\ell\rvert^3}{\gamma^{-c}-1}=\frac{1}{\lvert \mE_\ell\rvert}$, we have
\begin{equation*}
\Pr\left \{\lvert \mathcal M_{\ell,i}\rvert\ge \log_{\gamma^{-1}} (1+4(1+\lVert \bm \rho^{-1}\rVert_1)K\lvert \mE_\ell\rvert^4)\right \}\le \frac{1}{\lvert \mE_\ell\rvert}.
\end{equation*}
This completes the proof.
\end{proof}

\begin{lemma}[Misallocation Happens Rarely]\label{lem:misallocation}
For epoch $\ell\in [L]$ with dual $\bm \lambda_\ell\in \bm \Lambda$, we have
\begin{equation*}\begin{aligned}
&\Pr\Bigg \{\sum_{t\in \mE_\ell} \1\left [\argmax_{i\in \{0\} \cup [K]} (u_{t,i}-\bm \lambda_\ell^\trans \bm c_{t,i})\ne \argmax_{i\in \{0\} \cup [K]} (v_{t,i}-\bm \lambda_\ell^\trans \bm c_{t,i})\right ] \1\left [\lvert u_{t,i}-v_{t,i}\rvert\le \frac{1}{\lvert \mE_\ell\rvert},\forall i\in [K]\right ]\\
&\qquad \le 4K^2 \epsilon_c+4\log \lvert \mE_\ell\rvert\Bigg \}\ge 1-\frac{2}{\lvert \mE_\ell\rvert}.
\end{aligned}\end{equation*}
\end{lemma}
\begin{proof}
For any $t \in \mE_\ell$, we bound the probability that the allocation based on reported utilities differs from that based on true values, even when reports are close to truthful. Specifically, consider the event
\begin{equation*}
\argmax_{i\in \{0\} \cup [K]} (u_{t,i} - \bm \lambda_\ell^\trans \bm c_{t,i}) \ne \argmax_{i\in \{0\} \cup [K]} (v_{t,i} - \bm \lambda_\ell^\trans \bm c_{t,i})
\quad \text{and} \quad \lvert u_{t,i} - v_{t,i} \rvert \le \frac{1}{|\mE_\ell|},~\forall i\in [K].
\end{equation*}
Such a mismatch can only happen if there exists a pair of indices $i \ne j$ whose true dual-adjusted values are very close---within $\tfrac{2}{|\mE_\ell|}$---so that small deviations in reported utilities (bounded by $\tfrac{1}{|\mE_\ell|}$) are able to flip the argmax decision. We apply a Union Bound over all such pairs:
\begin{equation}\begin{aligned}
&\quad \Pr\left \{\left (\argmax_{i\in \{0\} \cup [K]} (u_{t,i}-\bm \lambda_\ell^\trans \bm c_{t,i})\ne \argmax_{i\in \{0\} \cup [K]} (v_{t,i}-\bm \lambda_\ell^\trans \bm c_{t,i})\right ) \wedge \left (\max_{i\in [K]} \lvert u_{t,i}-v_{t,i}\rvert\le \frac{1}{\lvert \mE_\ell\rvert}\right ) \right \}  \\
&\le \sum_{0\le i<j\le K}\Pr\left \{\lvert (v_{t,i}-\bm \lambda_\ell^\trans \bm c_{t,i})-(v_{t,j}-\bm \lambda_\ell^\trans \bm c_{t,j})\rvert \le \frac{2}{\lvert \mE_\ell\rvert}\right \} \le\frac{2K^2}{\lvert \mE_\ell\rvert}\epsilon_c, \label{eq:misalloc prob when similar}
\end{aligned}\end{equation}
where the second inequality comes from \Cref{assump:smooth consumptions}: Since it assumes that $\Law(\bm \lambda_\ell^\trans \bm c_{t,i})$ is uniformly upper bounded by $\epsilon_c$ where $\bm \lambda_\ell\in \bm \Lambda$, we know
\begin{equation*}
\Pr_{\bm c_{t,i}\sim \mC_i}\left \{-\frac{2}{\lvert \mE_\ell\rvert}\le \bm \lambda_\ell^\trans \bm c_{t,i}-(v_{t,i}-v_{t,j}+\bm \lambda_\ell^\trans \bm c_{t,j})\le \frac{2}{\lvert \mE_\ell\rvert}\right \}\le \frac{4}{\lvert \mE_\ell\rvert} \epsilon_c.
\end{equation*}

Now consider the martingale difference sequence $\{X_t-\E[X_t\mid \mF_{t-1}]\}_{t\in \mE_\ell}$ where
\begin{equation*}
X_t:=\1\left [\argmax_{i\in \{0\} \cup [K]} (u_{t,i}-\bm \lambda_\ell^\trans \bm c_{t,i})\ne \argmax_{i\in \{0\} \cup [K]} (v_{t,i}-\bm \lambda_\ell^\trans \bm c_{t,i})\right ] \1\left [\lvert u_{t,i}-v_{t,i}\rvert\le \frac{1}{\lvert \mE_\ell}\right ],
\end{equation*}
and $(\mF_{t})_t$ is the natural filtration defined as $\mF_t=\sigma(X_1,X_2,\ldots,X_t)$.
Apply the multiplicative Azuma-Hoeffding inequality \citep[Lemma 10]{koufogiannakis2014nearly} with $Y_t=\E[X_t\mid \mF_{t-1}]$, $\epsilon=\frac 12$, and $A=2\log \lvert \mE_\ell\rvert$. Since $X_t\in [0,1]$ a.s. and $\E[X_t-Y_t\mid \mF_{t-1}]=\E[X_t-\E[X_t\mid \mF_{t-1}]\mid \mF_{t-1}]=0$,
\begin{equation*}
\Pr\left \{\frac 12\sum_{t\in \mE_\ell} X_t\ge \sum_{t\in \mE_\ell} \E[X_t\mid \mF_{t-1}]+2\log \lvert \mE_\ell\rvert\right \}\le \exp(-\log \lvert \mE_\ell\rvert).
\end{equation*}

From \Cref{eq:misalloc prob when similar}, we know $\E[X_t\mid \mF_{t-1}]=\Pr\{X_t\mid \mF_{t-1}\}\le \frac{2K^2}{\lvert \mE_\ell\rvert}\epsilon_c \lVert \bm \lambda_\ell\rVert_1$. Therefore, we have
\begin{equation*}
\Pr\left \{\sum_{t\in \mE_\ell} X_t\le 4K^2 \epsilon_c+4\log \lvert \mE_\ell\rvert\right \}\ge 1-\frac{1}{\lvert \mE_\ell\rvert}.
\end{equation*}

Plugging back the definition of $X_t$ completes the proof.
\end{proof}

\begin{theorem}[\primalreg Guarantee; Formal \Cref{thm:informal IAPD main theorem}]\label{thm:InterEpoch guarantee}
Under \algname (\Cref{alg:mech}), there is a PBE of agents' strategies $\bm \pi^\ast$, such that the \primalreg term is bounded as (where $N_\ell$ is defined in \Cref{eq:N and M})
\begin{equation*}\begin{aligned}
\E[\primalreg]=\E\left [\sum_{t=1}^{\mT_v} (v_{t,\tilde i_t^\ast}-v_{t,i_t})\right ]\le \sum_{\ell=1}^L (N_\ell+3).
\end{aligned}\end{equation*}
\end{theorem}
\begin{proof}
For each epoch $\ell\in [L]$ and dual variable $\bm \lambda_\ell\in \Lambda$, from \Cref{lem:large misreport} to every epoch $\ell\in [L]$, there exists a $\bm \pi_{\ell,\bm \lambda_\ell}$ for the subgame of this epoch.
By definition of $\mT_v$, the safety constraint is never violated before that and thus \Cref{line:safety} has no effect before that.
Furthermore, since \Cref{alg:mech}'s allocations and payments within every epoch $\ell\in [L]$ only directly depend on the current dual variable $\bm \lambda_\ell$ but not anything else from the past, using the same arguments as in \Cref{thm:IntraEpoch guarantee}, there is a PBE $\bm \pi^\ast$ for the whole game under mechanism \Cref{alg:mech} that matches $(\bm \pi_{\ell,\bm \lambda_\ell})_{\ell\in [L]}$ up to $\mT_v$.

We now analyze this $\bm \pi^\ast$.
For every epoch $t\in \mE_\ell$, it can be either \textit{(i)} an exploration round, which happens w.p. $\frac{1}{\lvert \mE_\ell\rvert}$ independently, \textit{(ii)} a standard round with large misreports: $\exists i\in [K]$ such that $\lvert u_{t,i}-v_{t,i}\rvert\ge \frac{1}{\lvert \mE_\ell\rvert}$, or \textit{(iii)} a standard round with only small misreports. For \textit{(ii)}, we have \Cref{lem:large misreport}; for \textit{(iii)}, we have \Cref{lem:misallocation}. For \textit{(i)}, we apply the Chernoff-Bernstein inequality \citep[Proposition 2.14]{wainwright2019high}:
\begin{equation*}
\Pr\left \{\sum_{t\in \mE_\ell} \1[t\text{ exploration round}]>1+c\right \}\le \exp\left (-\frac{c^2}{2+2c/3}\right ),\quad \forall c>0.
\end{equation*}

Setting $c=\log \lvert \mE_\ell\rvert$ so that the RHS is no more than $\frac{1}{\lvert \mE_\ell\rvert}$, we get
\begin{equation}\label{eq:exploration round Chernoff}
\Pr\left \{\sum_{t\in \mE_\ell} \1[t\text{ exploration round}]\le 1+\log \lvert \mE_\ell\rvert\right \}\ge 1-\frac{1}{\lvert \mE_\ell\rvert}.
\end{equation}

Now we put the aforementioned three cases together:
\begin{equation}\begin{aligned}
&\quad \E[\primalreg]\overset{(a)}{\le} \E\left [\sum_{t=1}^{\mT_v} \1[\tilde i_t^\ast\ne i_t]\right ] \\
&\overset{(b)}{\le} \sum_{\ell=1}^L \E\Bigg [\sum_{t\in \mE_\ell}\Bigg (\1[t\text{ exploration round}]+\1\left [\exists i\in [K],\lvert u_{t,i}-v_{t,i}\rvert\ge \frac{1}{\lvert \mE_\ell}\right ]\\&\qquad +\1\left [\argmax_{i\in \{0\} \cup [K]} (u_{t,i}-\bm \lambda_\ell^\trans \bm c_{t,i})\ne \argmax_{i\in \{0\} \cup [K]} (v_{t,i}-\bm \lambda_\ell^\trans \bm c_{t,i})\wedge \lvert u_{t,i}-v_{t,i}\rvert\le \frac{1}{\lvert \mE_\ell\rvert},\forall i\in [K]\right ]\Bigg )\Bigg ] \\
&\overset{(c)}{\le} \sum_{\ell=1}^L \left (1+\log \lvert \mE_\ell\rvert+K \log_{\gamma^{-1}} (1+4(1+\lVert \bm \rho^{-1}\rVert_1)K\lvert \mE_\ell\rvert^4)+4K^2 \epsilon_c+4\log \lvert \mE_\ell\rvert+\frac{3\lvert \mE_\ell\rvert}{\lvert \mE_\ell\rvert}\right ) \label{eq:InterEpoch takeaway}\\
&=\sum_{\ell=1}^L (N_\ell+3),\quad N_\ell:= 1+K\log_{\gamma^{-1}} (1+4(1+\lVert \bm \rho^{-1}\rVert_1)K\lvert \mE_\ell\rvert^4)+4K^2 \epsilon_c +5\log \lvert \mE_\ell\rvert, 
\end{aligned}\end{equation}
where (a) uses $v_{t,i}\in [0,1]$ for all $t$ and $i$; (b) uses the above discussions of \textit{(i)}, \textit{(ii)}, and \textit{(iii)}; (c) applies \Cref{line:exploration} to \textit{(i)}, \Cref{lem:large misreport} to \textit{(ii)}, \Cref{lem:misallocation} to \textit{(iii)}, and the trivial bound that $\sum_{t\in \mE_\ell} \1[\tilde i_t^\ast\ne i_t]\le \lvert \mE_\ell\rvert$ if any of these conclusions in \Cref{line:exploration}, \Cref{lem:large misreport}, or \Cref{lem:misallocation} do not hold (each of them holds with probability at least $1-\frac{1}{\lvert \mE_\ell\rvert}$, and thus a Union Bound controls the overall failure probability by $\frac{3}{\lvert \mE_\ell\rvert}$).
\end{proof}

\section{\dualreg: Inaccurate Dual Variables Due to Incomplete Information}\label{sec:appendix DualUpd}
\subsection{\dualreg and Online Learning Regret}\label{sec:DualUpd to regret}
Recall that $\mT_v:=\min\{t\in [T]\mid \sum_{{\tau}=1}^{t} \bm c_{{\tau},i_{\tau}} + \bm 1\not \le T\bm \rho\}\cup \{T+1\}$ is be the (random) first round where the feasibility can possibly be violated, and recall that $\mathcal L_v$ is the epoch where $\mT_v$ belongs to. For a pre-determined epoching scheme $\{\mE_\ell\}_{\ell\in [L]}$ (as we do in both \mech and \mechO), since $\mT_v$ is a stopping time, $\mathcal L_v$ is also a stopping time. We formally define the \emph{online learning regret} $\mR_L^{\text{OL}}$ as (recall \Cref{eq:online learning regret informal}):
\begin{equation}\label{eq:online learning regret}
\mR_L^{\text{OL}}:=\E\left[\sup_{\bm \lambda^\ast \in \bm \Lambda} \sum_{\ell=1}^{\mathcal L_v} \left (\sum_{t\in \mE_\ell} (\bm \rho-\bm c_{t,i_t})\right )^\trans (\bm \lambda_\ell-\bm \lambda^\ast) \right].
\end{equation}
We remark that most online learning algorithms enjoy the so-called ``anytime guarantee,'' meaning that their regret guarantees also hold w.r.t. stopping times, and thus \Cref{eq:online learning regret} does not bring extra difficulties.

\begin{theorem}[\dualreg as Online Learning Regret; Restatement of \Cref{lem:DualUpd to regret informal}]\label{lem:DualUpd to regret}
Under \algname (\Cref{alg:mech}), when fixing $[T]=\bigcup_{\ell=1}^L \mE_\ell$ and letting $\bm \Lambda=\prod_{j=1}^d[0,\rho_j^{-1}]:=\{\bm \lambda\in \mathbb R_{\ge 0}^d\mid \lambda_j\le \rho_j^{-1},\forall j\in [d]\}$,
\begin{equation*}\begin{aligned}
\E[\dualreg]=\E\left [\sum_{t=1}^T v_{t,i_t^\ast}-\sum_{t=1}^{\mT_v} v_{t,\tilde i_t^\ast}\right ]\le \mR_L^{\text{OL}} + \left (1+\max_{1\le \ell\le L} \lvert \mE_\ell\rvert + \sum_{\ell=1}^L (N_\ell+3)\right ) \lVert \bm \rho^{-1}\rVert_1,
\end{aligned}\end{equation*}
where $\mR_L^{\text{OL}}$ is defined in \Cref{eq:online learning regret} and $N_\ell$ is defined in \Cref{eq:N and M}.
\end{theorem}
\begin{proof}
While the reduction to online learning regret is standard in online resource allocation \citep[see, e.g.,][]{devanur2023online}, we face two unique ingredients of epoching (inducing the $\max_\ell \lvert \mE_\ell\rvert$ term) and strategic manipulation (inducing the $\sum_\ell (N_\ell+3)$ term). The proof largely adopts that of \citet[Theorem 1]{balseiro2020dual}, and we also incorporate some martingale arguments of \citet{castiglioni2022online}.
Slightly abusing the notations, for any round $t$ belonging to an epoch $\mE_\ell$, we use $\bm \lambda_t$ and $\bm \lambda_\ell$ interchangably for the dual variable. Let $(\mF_t)_{t\ge 0}$ be the filtration specified as $\mF_t=\sigma(\bm \lambda_1,\ldots,\bm \lambda_t,\bm v_1,\ldots,\bm v_t,\bm c_1,\ldots,\bm c_t)$.
For any $t\in [T]$, let $v_t^\ast(\bm \lambda):=\max_{i\in \{0\}\cup [K]} (v_{t,i}+\bm \lambda^\trans (\bm \rho-\bm c_{t,i}))$ be the dual value with Lagrangian $\bm \lambda\in \bm \Lambda$; it is convex.
Let $v^\ast(\bm \lambda)=\E_{\bm v\sim \mV,\bm c\sim \mC} [\max_{i\in \{0\}\cup [K]}(v_{i}+\bm \lambda^\trans (\bm \rho-\bm c_{i}))]$. It is an expectation of a convex function and also convex.

\paragraph{Step 1: Lower bound the values collected by $\{\tilde i_t^\ast\}_{t\in [\mT_v]}$.}
By definition of $\tilde i_t^\ast$, $v_{t,\tilde i_t^\ast}=v_t^\ast(\bm \lambda_t)+\bm \lambda_t^\trans \bm c_{t,\tilde i_t^\ast}-\bm \lambda_t^\trans \bm \rho$ for all $t\le \mT_v$.
Since $\bm \lambda_t$ is $\mF_{t-1}$-measurable but $\bm v_t$ and $\bm c_t$ are sampled from $\mV$ and $\mC$ independently to $\mF_{t-1}$,
\begin{equation*}
\E\left [v_{t,\tilde i_t^\ast}\middle \vert \mF_{t-1}\right ]=\E\left [v_t^\ast(\bm \lambda_t)+\bm \lambda_t^\trans \bm c_{t,\tilde i_t^\ast}-\bm \lambda_t^\trans \bm \rho\middle \vert \mF_{t-1}\right ]=v^\ast(\bm \lambda_t)-\bm \lambda_t^\trans \bm \rho+\E\left [\bm \lambda_t^\trans \bm c_{t,\tilde i_t^\ast}\middle \vert \mF_{t-1}\right ].
\end{equation*}
Put this equality in another way, the stochastic process $(X_t)_{t\ge 1}$ adapted to $(\mF_t)_{t\ge 0}$ defined as
\begin{equation*}
X_t:=v_{t,\tilde i_t^\ast} + \bm \lambda_t^\trans (\bm \rho-\bm c_{t,\tilde i_t^\ast})-v^\ast(\bm \lambda_t),\quad \forall t\le \mT_v,
\end{equation*}
ensures $\E[X_t\mid \mF_{t-1}]=0$ and is thus a martingale difference sequence. Since $\mT_v\le T+1$ almost surely by definition, Optional Stopping Time theorem \citep[Theorem 10.10]{williams1991probability} gives
\begin{equation*}
\E\left [\sum_{t=1}^{\mT_v} X_t\right ]=0\Longrightarrow \E\left [\sum_{t=1}^{\mT_v} v_{t,\tilde i_t^\ast}\right ]=\E\left [\sum_{t=1}^{\mT_v} v^\ast(\bm \lambda_t)-\sum_{t=1}^{\mT_v} \bm \lambda_t^\trans (\bm \rho-\bm c_{t,\tilde i_t^\ast})\right ].
\end{equation*}

Utilizing the convexity of $v^\ast$ gives
\begin{equation}\begin{aligned}
\E\left [\sum_{t=1}^{\mT_v} v_{t,\tilde i_t^\ast}\right ]&=\E\left [\mT_v\cdot \frac{1}{\mT_v}\sum_{t=1}^{\mT_v} v^\ast(\bm \lambda_t)-\sum_{t=1}^{\mT_v} \bm \lambda_t^\trans (\bm \rho-\bm c_{t,\tilde i_t^\ast})\right ]\\
&\overset{(a)}{\ge} \E\left [\mT_v\cdot v^\ast\left (\frac{1}{\mT_v}\sum_{t=1}^{\mT_v} \bm \lambda_t\right )-\sum_{t=1}^{\mT_v} \bm \lambda_t^\trans (\bm \rho-\bm c_{t,\tilde i_t^\ast})\right ] \\
&\overset{(b)}{\ge} \E[\mT_v] \inf_{\bm \lambda\in \bm \Lambda} v^\ast(\bm \lambda)-\E\left [\sum_{t=1}^{\mT_v} \bm \lambda_t^\trans (\bm \rho-\bm c_{t,\tilde i_t^\ast})\right ], \label{eq:sum of v for tilde i}
\end{aligned}\end{equation}
where (a) uses Jensen's inequality and (b) uses $\bm \lambda=\frac{1}{\mT_v}\sum_{\tau=1}^{\mT_v} \bm \lambda_\tau\in \bm \Lambda$ (as $\bm \lambda_\tau\in \bm \Lambda$ for all $\tau$ and $\bm \Lambda$ is convex).

\paragraph{Step 2: Upper bound the offline optimal social welfare $\sum_{t=1}^T v_{t,i_t^\ast}$.}
We now work on the other term in $\E[\dualreg]$, namely $\E[\sum_{t=1}^T v_{t,i_t^\ast}]$.
For any fixed $\bm \lambda\in \bm \Lambda$, $v_t^\ast(\bm \lambda)\ge v_{t,i_t^\ast}+\bm \lambda^\trans (\bm \rho-\bm c_{t,i_t^\ast})$.
Consider the stochastic process $(Y_t)_{t=1}^T$ adapted to $(\mF_t)_{t=0}^T$ where $\mF_t=\sigma(Y_1,Y_2,\ldots,Y_t)$ is the natural filtration:
\begin{equation*}
Y_t:=(T-t)v^\ast(\bm \lambda)+\sum_{\tau=1}^t (v_{\tau,i_\tau^\ast}+\bm \lambda^\trans (\bm \rho-\bm c_{\tau,i_\tau^\ast})),
\end{equation*}

Since $\E[v_t^\ast(\bm \lambda)\mid \mF_{t-1}]=v^\ast(\bm \lambda)$ due to the i.i.d. nature of $\bm v_t$ and $\bm c_t$, $(Y_t)_{t=1}^T$ is a super-martingale:
\begin{equation*}
\E[Y_t\mid \mF_{t-1}]-Y_{t-1}=\E[v_{t,i_t^\ast}+\bm \lambda^\trans(\bm \rho-\bm c_{t,i_t^\ast})\mid \mF_{t-1}]-v^\ast(\bm \lambda)\le 0.
\end{equation*}
Since $\mT_v\le (T+1)$ a.s., the Optional Stopping Time theorem \citep[Theorem 10.10]{williams1991probability} gives
\begin{equation}\label{eq:supermartingale UB part in duality}
Tv^\ast(\bm \lambda)=\E[Y_0]\ge \E[Y_{\mT_v}]=\E\left [(T-\mT_v) v^\ast(\bm \lambda)+\sum_{t=1}^{\mT_v} (v_{t,i_t^\ast}+\bm \lambda^\trans (\bm \rho-\bm c_{t,i_t^\ast}))\right ].
\end{equation}

Since $\{i_t^\ast\}_{t=1}^T$ is a feasible allocation sequence, we must have $\sum_{t=1}^{\mT_v} \bm c_{t,i_t^\ast}\le T\bm \rho$, which gives
\begin{equation*}\begin{aligned}
\E[Y_{\mT_v}]&=\E\left [(T-\mT_v) v^\ast(\bm \lambda)+\sum_{t=1}^{\mT_v} (v_{t,i_t^\ast}+\bm \lambda^\trans (\bm \rho-\bm c_{t,i_t^\ast}))\right ]\\
&\ge \E\left [(T-\mT_v) v^\ast(\bm \lambda)+\sum_{t=1}^{\mT_v} v_{t,i_t^\ast}+\mT_v \bm \lambda^\trans \bm \rho-T \bm \lambda^\trans \bm \rho\right ]\\
&=\E\left [\sum_{t=1}^{\mT_v} v_{t,i_t^\ast}+(T-\mT_v) (v^\ast(\bm \lambda)-\bm \lambda^\trans \bm \rho)\right ]\overset{(a)}{\ge} \E\left [\sum_{t=1}^{\mT_v} v_{t,i_t^\ast}\right ],
\end{aligned}\end{equation*}
where (a) is because $0$ is always a feasible action and thus $v^\ast(\bm \lambda)\ge \bm \lambda^\trans \bm \rho$, i.e.,
\begin{equation*}
v^\ast(\bm \lambda)=\E_{\bm v\sim \mV,\bm c\sim \mC} \left [\max_{i\in \{0\}\cup [K]}(v_{i}+\bm \lambda^\trans (\bm \rho-\bm c_{i}))\right ]\ge \E_{\bm v\sim \mV,\bm c\sim \mC} \left [(0+\bm \lambda^\trans (\bm \rho-\bm 0))\right ]=\bm \lambda^\trans \bm \rho.
\end{equation*}

Plugging this back into \Cref{eq:supermartingale UB part in duality}, we get
\begin{equation}\label{eq:sum of v for i ast}
\E\left [\sum_{t=1}^{\mT_v} v_{t,i_t^\ast}\right ]\le \E[Y_{\mT_v}]\le \E[Y_0]\le T v^\ast(\bm \lambda),\quad \forall \bm \lambda\in \bm \Lambda.
\end{equation}
This reveals that the hindsight optimal social welfare has an expectation bounded by $T \inf_{\bm \lambda\in \bm \Lambda}v^\ast(\bm \lambda)$.

Putting \Cref{eq:sum of v for tilde i,eq:sum of v for i ast} together, we get
\begin{equation}\begin{aligned}
\E[\dualreg]=\E\left [\sum_{t=1}^{T} v_{t,i_t^\ast}-\sum_{t=1}^{\mT_v} v_{t,\tilde i_t^\ast}\right ]
&\le T\inf_{\bm \lambda\in \bm \Lambda} v^\ast(\bm \lambda) - \E[\mT_v] \inf_{\bm \lambda\in \bm \Lambda} v^\ast(\bm \lambda) +\E\left [\sum_{t=1}^{\mT_v} \bm \lambda_t^\trans (\bm \rho-\bm c_{t,\tilde i_t^\ast})\right ] \\
&\le \E[T-\mT_v]+\E\left [\sum_{t=1}^{\mT_v} \bm \lambda_t^\trans (\bm \rho-\bm c_{t,\tilde i_t^\ast})\right ],\label{eq:dualvar bound after optional stopping}
\end{aligned}\end{equation}
where we used the upper bound that $\inf_{\bm \lambda\in \bm \Lambda}v^\ast(\bm \lambda)\le v^\ast(\bm 0)=\E[\max_{i}v_i]\le 1$.
Comparing \Cref{eq:dualvar bound after optional stopping} to our conclusion, it only remains to associate $\bm \lambda_t^\trans \bm c_{t,\tilde i_t^\ast}$ with $\bm \lambda_t^\trans \bm c_{t,i_t}$ and further control $\E[T-\mT_v]$.

\paragraph{Step 3: Relate $\bm \lambda_t^\trans (\bm \rho-\bm c_{t,\tilde i_t^\ast})$ to $\bm \lambda_t^\trans (\bm \rho-\bm c_{t,i_t})$.}
Using \Cref{eq:InterEpoch takeaway} from \Cref{thm:InterEpoch guarantee}, with probability $1-\frac{3}{\lvert \mE_\ell\rvert}$, the sequence $\{\tilde i_t^\ast\}_{t\in \mE_\ell}$ and $\{i_t\}_{t\in \mE_\ell}$ only differs by no more than $N_{\ell}$ (defined in \Cref{eq:N and M}).

That is, we have shown than $\E\left [\sum_{t=1}^{\mT_v} \1[\tilde i_t^\ast\ne i_t]\right ] \le \sum_{\ell=1}^L({N_\ell}+3)$ where the 3 comes from the $\frac{3}{\lvert \mE_\ell\rvert}$ failure probability and the fact that $\sum_{t\in \mE_\ell} \1[\tilde i_t^\ast\ne i_t]\le \lvert \mE_\ell\rvert$.
Combining it with the observation that
\begin{equation*}
\lvert \bm c_{t,i}^\trans \bm \lambda_t-\bm c_{t,j}^\trans \bm \lambda_t\rvert\le \lVert \bm \lambda_t\rVert_1\cdot \lVert \bm c_{t,i}-\bm c_{t,i}\rVert_\infty\le \bm \lVert \bm \rho^{-1}\rVert_1,\quad \forall i\ne j,
\end{equation*}
where we shall recall that $\bm c_{t,i},\bm c_{t,j}\in [0,1]^d$ and that $\bm \lambda_t\in \bm \Lambda=\bigotimes_{j=1}^d [0,\rho_j^{-1}]$, \Cref{eq:dualvar bound after optional stopping} gives
\begin{equation}\begin{aligned}
\E[\dualreg]&\le \E\left [\sum_{t=1}^{\mT_v} \bm \lambda_t^\trans (\bm \rho-\bm c_{t,\tilde i_t^\ast})\right ]+\E[T-\mT_v] \\
&\le \E\left [\sum_{t=1}^{\mT_v} (\bm \lambda_t^\trans (\bm \rho-\bm c_{t,i_t})+\lVert \bm \rho^{-1}\rVert_1 \cdot \1[\tilde i_t^\ast\ne i_t])\right ]+\E[T-\mT_v] \\
&\le \E\left [\sum_{t=1}^{\mT_v} \bm \lambda_t^\trans (\bm \rho-\bm c_{t,i_t})\right ]+\lVert \bm \rho^{-1}\rVert_1\sum_{\ell=1}^L (N_\ell+3)+\E[T-\mT_v]. \label{eq:DualUpd after similarity}
\end{aligned}\end{equation}

\paragraph{Step 4: Control $\E[T-\mT_v]$.}
Recall the definition that $\mT_v:=\min \{t\in [T]\mid \sum_{{\tau}=1}^{t} \bm c_{{\tau},i_{\tau}} + \bm 1\not \le T\bm \rho\}\cup \{T+1\}$.
If $\mT_v=T+1$, then $(T-\mT_v)$ is trivially bounded. Otherwise, suppose that the condition $\sum_{{\tau}=1}^{\mT_v} \bm c_{{\tau},i_{\tau}} + \bm 1\not \le T\bm \rho$ is violated for the $j\in [d]$-th coordinate (if there are multiple $j$'s, pick one arbitrarily). We have
\begin{equation}\label{eq:j violates rho}
\sum_{t=1}^{\mT_v} c_{t,i_t,j}+1>T\rho_j \Longrightarrow \sum_{t=1}^{\mT_v} (\rho_j-c_{t,i_t,j})<1-(T-\mT_v)\rho_j.
\end{equation} 

Let $\bm \lambda^\ast=\frac{1}{\rho_j} \bm e_j$ where $\bm e_j$ is the one-hot vector over coordinate $j$, we know $\bm \lambda^\ast\in \bm \Lambda$ and that
\begin{equation*}
\sum_{t=1}^{\mT_v} (\bm \rho-\bm c_{t,i_t})^\trans \bm \lambda^\ast =\sum_{t=1}^{\mT_v} \frac{\rho_j-c_{t,i_t,j}}{\rho_j} \overset{\text{\Cref{eq:j violates rho}}}{<}\frac{1-(T-\mT_v)\rho_j}{\rho_j}=\rho_j^{-1}-(T-\mT_v).
\end{equation*}
Rearranging gives $(T-\mT_v)\le \max_{j\in [d]} \rho_j^{-1}+\sup_{\bm \lambda^\ast \in \bm \Lambda} \left (\sum_{t=1}^{\mT_v} (\bm \rho-\bm c_{t,i_t})^\trans \bm \lambda^\ast\right )$.

\paragraph{Step 5: Final bound.}
Taking expectation and plugging it back to \Cref{eq:DualUpd after similarity}, we yield
\begin{equation*}\begin{aligned}
\E[\dualreg]&\le \E\left [\sum_{t=1}^{\mT_v} (\bm \rho-\bm c_{t,i_t})^\trans \bm \lambda_t\right ]+ \sum_{\ell=1}^L (N_\ell+3) \lVert \bm \rho^{-1}\rVert_1+\E\left [\max_{j\in [d]} \rho_j^{-1}+\sup_{\bm \lambda^\ast \in \bm \Lambda} \left (\sum_{t=1}^{\mT_v} (\bm \rho-\bm c_{t,i_t})^\trans \bm \lambda^\ast\right )\right ]\\
&\le \E\left [\sup_{\bm \lambda^\ast\in \bm \Lambda}\sum_{t=1}^{\mT_v} (\bm \rho-\bm c_{t,i_t})^\trans(\bm \lambda_t-\bm \lambda^\ast)\right ]+\sum_{\ell=1}^L (N_\ell+3) \lVert \bm \rho^{-1}\rVert_1+\lVert \bm \rho^{-1}\rVert_\infty.
\end{aligned}\end{equation*}

Recall $\mathcal L_v$ is the (random) epoch where $\mT_v$ lies in. Since in every round $t\in \mE_{\mathcal L_v}$ the single-round regret satisfies $\lvert (\bm \rho-\bm c_{t,i_t})^\trans (\bm \lambda_\ell-\bm \lambda^\ast)\rvert\le \lVert \bm \rho-\bm c_{t,i_t}\rVert_\infty \cdot \lVert \bm \lambda_\ell-\bm \lambda^\ast\rVert_1\le 2\lVert \bm \rho^{-1}\rVert_1$ (single-round regret can be negative, so we must control its absolute value),
and this epoch has length $\lvert \mE_{\mathcal L_v}\rvert\le \max_{\ell\in [L]} \lvert \mE_\ell\rvert$ a.s., we have
\begin{equation*}
\E[\dualreg]\le \E\left [\sup_{\bm \lambda^\ast\in \bm \Lambda} \sum_{\ell=1}^{\mathcal L_v} \left (\sum_{t\in \mE_\ell} (\bm \rho-\bm c_{t,i_t})\right )^\trans (\bm \lambda_\ell-\bm \lambda^\ast)\right ]+\left (1+\max_{1\le \ell\le L} \lvert \mE_\ell\rvert + \sum_{\ell=1}^L (N_\ell+3)\right ) \lVert \bm \rho^{-1}\rVert_1,
\end{equation*}
where we used the fact $\lVert \bm \rho^{-1}\rVert_\infty\le \lVert \bm \rho^{-1}\rVert_1$ as well.
\end{proof}

\subsection{\dualreg Guarantee for FTRL in \Cref{alg:lambda FTRL}}\label{sec:DualUpd guarantee FTRL}
\begin{theorem}[\dualreg Guarantee with FTRL; Restatement of \Cref{thm:DualUpd guarantee FTRL informal}]\label{thm:DualUpd guarantee FTRL}
When using \ftrl (\Cref{eq:lambda FTRL} of \Cref{alg:lambda FTRL}) to decide $\{\bm \lambda_\ell\}_{\ell\in [L]}$, the online learning regret is no more than
\begin{equation*}
\mR_L^{\text{OL}}:=\E\left [\sup_{\bm \lambda^\ast\in \bm \Lambda} \sum_{\ell=1}^{\mathcal L_v} \left (\sum_{t\in \mE_\ell} (\bm \rho-\bm c_{t,i_t})\right )^\trans (\bm \lambda_\ell-\bm \lambda^\ast)\right ]\le \sup_{\bm \lambda^\ast\in \bm \Lambda}\frac{\Psi(\bm \lambda^\ast)}{\eta_L}+d \sum_{\ell=1}^L \eta_\ell \lvert \mE_\ell\rvert^2.
\end{equation*}
Specifically, with $\Psi(\bm \lambda)=\frac 12 \lVert \bm \lambda\rVert_2^2$ and $\eta_\ell=\frac{\lVert \bm \rho^{-1}\rVert_2}{\sqrt{2d}}\big (\sum_{\ell'=1}^{\ell} \lvert \mE_{\ell'}\rvert^2\big )^{-1/2}$ for all $\ell\in [L]$, (where $N_\ell$ is in \Cref{eq:N and M})
\begin{equation*}
\E[\dualreg]\le \sqrt{2d}\lVert \bm \rho^{-1}\rVert_2 \cdot \sqrt{\sum_{\ell=1}^L \lvert \mE_\ell\rvert^2}+ \left (1+\max_{1\le \ell\le L} \lvert \mE_\ell\rvert+\sum_{\ell=1}^L (N_\ell+3)\right ) \lVert \bm \rho^{-1}\rVert_1.
\end{equation*}
\end{theorem}
\begin{proof}
Since $\mathcal L_v$ is a stopping time bounded by $L$ a.s., applying standard \ftrl guarantee \citep[see, e.g.,][Corollary 7.7]{orabona2019modern}  over $\bm \Lambda=\bigotimes_{j=1}^d [0,\rho_j^{-1}]$ with loss functions $F_\ell(\bm \lambda):=\sum_{t\in \mE_\ell} (\bm \rho-\bm c_{t,i_t})^\trans \bm \lambda$ gives
\begin{equation*}
\mR_L^{\text{OL}}=\E\left [\sup_{\bm \lambda^\ast\in \bm \Lambda} \sum_{\ell=1}^{\mathcal L_v} \left (\sum_{t\in \mE_\ell} (\bm \rho-\bm c_{t,i_t})\right )^\trans (\bm \lambda_\ell-\bm \lambda^\ast)\right ]\le \sup_{\bm \lambda^\ast\in \bm \Lambda} \frac{\Psi(\bm \lambda^\ast)}{\eta_L}+\sum_{\ell=1}^L \eta_\ell \E\left [\left \lVert \sum_{t\in \mE_\ell} (\bm \rho-\bm c_{t,i_t})\right \rVert_2^2\right ],
\end{equation*}
where we used $\nabla_{\bm \lambda} F_\ell(\bm \lambda)=\sum_{t\in \mE_\ell}(\bm \rho-\bm c_{t,i_t})$.
As $\bm \rho,\bm c_{t,i}\in [0,1]^d$, $\lVert \sum_{t\in \mE_\ell} (\bm \rho-\bm c_{t,i_t})\rVert_2^2\le d \lvert \mE_\ell\rvert^2$. Hence
\begin{equation*}
\mR_L^{\text{OL}}\le \sup_{\bm \lambda^\ast\in \bm \Lambda}\frac{\Psi(\bm \lambda^\ast)}{\eta_L}+d \sum_{\ell=1}^L \eta_\ell \lvert \mE_\ell\rvert^2.
\end{equation*}
This gives the first conclusion. When $\Psi(\bm \lambda)=\frac 12 \lVert \bm \lambda\rVert_2^2$, we know $\Psi(\bm \lambda^\ast)=\frac 12 \lVert \bm \lambda^\ast\rVert_2^2\le \frac 12 \lVert \bm \rho^{-1}\rVert_2^2$ for all $\bm \lambda^\ast\in \bm \Lambda=\bigotimes_{j=1}^d [0,\rho_j^{-1}]$. Further plugging in $\eta_\ell=\frac{\lVert \bm \rho^{-1}\rVert_2}{\sqrt{2d}}\big (\sum_{\ell'=1}^\ell \lvert \mE_{\ell'}\rvert^2\big )^{-1/2}$ gives
\begin{equation*}\begin{aligned}
\mR_L^{\text{OL}}&\le \frac{\sqrt{2d}}{\lVert \bm \rho^{-1}\rVert_2} \cdot \frac 12 \lVert \bm \rho^{-1}\rVert_2^2 \sqrt{\sum_{\ell=1}^L \lvert \mE_\ell\rvert^2}+\frac{\lVert \bm \rho^{-1}\rVert_2}{\sqrt{2d}}\cdot d \sum_{\ell=1}^L \frac{\lvert \mE_\ell\rvert^2}{\sqrt{\sum_{\ell'\le \ell} \lvert \mE_{\ell'}\rvert^2}}\\
&\overset{(a)}{\le} \sqrt{2d} \cdot \frac 12 \lVert \bm \rho^{-1}\rVert_2 \sqrt{\sum_{\ell=1}^L \lvert \mE_\ell\rvert^2}+\frac{\lVert \bm \rho^{-1}\rVert_2}{\sqrt{2d}} \cdot \frac{2d}{2} \sqrt{\sum_{\ell=1}^L \lvert \mE_\ell\rvert^2}\overset{(b)}{=}\sqrt{2d}\lVert \bm \rho^{-1}\rVert_2 \sqrt{\sum_{\ell=1}^L \lvert \mE_\ell\rvert^2},
\end{aligned}\end{equation*}
where (a) uses the folklore summation lemma that $\sum_{t=1}^T \frac{x_t}{\sqrt{\sum_{s=1}^t x_s}}\le 2 \sqrt{\sum_{t=1}^T x_t}$ for all $x_1,x_2,\ldots,x_T\in \mathbb R_{\ge 0}$ \citep[see, e.g.,][Lemma 4]{duchi2011adaptive} and (b) follows from rearranging the terms.

Plugging the online learning regret $\mR_L^{\text{OL}}$ into \Cref{lem:DualUpd to regret}, we therefore get
\begin{equation*}
\E[\dualreg]\le \sqrt{2d}\lVert \bm \rho^{-1}\rVert_2 \sqrt{\sum_{\ell=1}^L \lvert \mE_\ell\rvert^2}+ \left (1+\max_{1\le \ell\le L} \lvert \mE_\ell\rvert+\sum_{\ell=1}^L (N_\ell+3)\right ) \lVert \bm \rho^{-1}\rVert_1,
\end{equation*}
where $N_\ell$ is defined in \Cref{eq:N and M}. This finishes the proof.
\end{proof}

\subsection{\dualreg Guarantee for O-FTRL-FP in \Cref{eq:lambda O-FTRL-FP}}

\begin{theorem}[\dualreg Guarantee with O-FTRL-FP; Restatement of \Cref{thm:DualUpd guarantee O-FTRL-FP informal}]\label{thm:DualUpd guarantee O-FTRL-FP}
When using \oftrlfp (\Cref{eq:lambda O-FTRL-FP} of \Cref{alg:lambda O-FTRL-FP}) to decide $\{\bm \lambda_\ell\}_{\ell\in [L]}$, the online learning regret is no more than
\begin{equation*}\begin{aligned}
\mR_L^{\text{OL}}&\le \sup_{\bm \lambda^\ast\in \bm \Lambda} \frac{\Psi(\bm \lambda^\ast)}{\eta_L}+d \eta_1\lvert \mE_1\rvert^2+\sum_{\ell=2}^L \eta_\ell \frac{\lvert \mE_\ell\rvert^2}{\sum_{\ell'<\ell} \lvert \mE_{\ell'}\rvert} \left (4d^2 + 24d \log(dTL) + 24d \sum_{j=1}^d \log \frac{dK^2 \epsilon_c T}{\rho_j}\right ) \\
&\quad +\sum_{\ell=2}^L \eta_\ell \lvert \mE_\ell\rvert (16d+10)+\sum_{\ell=1}^L \eta_\ell\left (2N_\ell^2 + \frac{2d \lvert \mE_\ell\rvert^2 M_\ell^2}{(\sum_{\ell'=1}^{\ell-1} \lvert \mE_{\ell'}\rvert)^2}\right )+2\lVert \bm \rho^{-1}\rVert_1,
\end{aligned}\end{equation*}
where $N_\ell$ and $M_\ell$ are defined in \Cref{eq:N and M}.
For readability, we also annotate the order of every term in terms of $\Otil_{T}$, which only highlights the polynomial dependency on $T$ (including $L$, $\{\lvert \mE_\ell\rvert\}_{\ell\in [L]}$, and $\{\eta_\ell\}_{\ell\in [L]}$):
\begin{equation*}
\mR_L^{\text{OL}}=\Otil_T \left (\eta_L^{-1}+\eta_1 \lvert \mE_1\rvert^2+\sum_{\ell=1}^L \eta_\ell \frac{\lvert \mE_\ell\rvert^2}{\sum_{\ell'<\ell} \lvert \mE_{\ell'}\rvert}+\sum_{\ell=1}^L \eta_\ell \lvert \mE_\ell\rvert +\sum_{\ell=1}^L \eta_\ell \frac{\lvert \mE_\ell\rvert^2 \ell^2}{(\sum_{\ell'=1}^{\ell-1} \lvert \mE_{\ell'}\rvert)^2}+1\right ).
\end{equation*}
Specifically, with $\Psi(\bm \lambda)=\frac 12 \lVert \bm \lambda\rVert_2^2$ and $\eta_\ell=\frac{\lVert \bm \rho^{-1}\rVert_2}{2d}\big (\sum_{\ell'=1}^{\ell} \lvert \mE_{\ell'}\rvert\big )^{-1/2}$ for all $\ell\in [L]$, when $\lvert \mE_1\rvert\le \lvert \mE_2\rvert \le \cdots \lvert \mE_L\rvert$,
\begin{equation*}\begin{aligned}
\E[\dualreg]&\le 2 \sqrt{T} \left (3d + 12 \log(dTL) + 12 \sum_{j=1}^d \log \frac{dK^2 \epsilon_c T}{T\rho_j}+18\right )\lVert \bm \rho^{-1}\rVert_1+d \lvert \mE_1\rvert^{1.5} \\
&\quad +\sum_{\ell=2}^L \frac{N_\ell^2/d+\lvert \mE_\ell\rvert^2 M_\ell^2/(\sum_{\ell'=1}^{\ell-1} \lvert \mE_{\ell'}\rvert)^2}{(\sum_{\ell'=1}^\ell \lvert \mE_{\ell'}\rvert)^{1/2}} \lVert \bm \rho^{-1}\rVert_1 + \left (3+\max_{\ell\in [L]} \lvert \mE_\ell\rvert+\sum_{\ell=1}^L (N_\ell+3)\right )\lVert \bm \rho^{-1}\rVert_1.
\end{aligned}\end{equation*}
\end{theorem}
\begin{proof}
\textbf{Step 1: Establishing $(\epsilon, L)$-gradient-stability.}
In \Cref{lem:approximate continuity of predictions}, we prove for all $\ell\ge 2$ that
\begin{equation}\begin{aligned}
&\lVert \tilde{\bm g}_\ell(\bm \lambda_1)-\tilde{\bm g}_\ell(\bm \lambda_2)\rVert_2 \le 4\lvert \mE_{\ell}\rvert K^2 \epsilon_\ell \epsilon_c d + \frac{4\lvert \mE_\ell\rvert (\log \frac{1}{\delta_\ell}+\sum_{j=1}^d \log \frac{\sqrt d}{\rho_j \epsilon_\ell})}{\sum_{\ell'<\ell} \lvert \mE_{\ell'}\rvert}\sqrt d,  \\
&\quad \forall \bm \lambda_1,\bm \lambda_2\in \bm \Lambda\text{ s.t. }\lVert \bm \lambda_1-\bm \lambda_2\rVert_2\le \epsilon_\ell,\quad \text{with probability }1-\delta_\ell,\label{eq:approximate continuity of predictions}
\end{aligned}\end{equation}
for any fixed $\epsilon_\ell>0$ and $\delta_\ell\in (0,1)$.
With $\epsilon_\ell$ and $L_\epsilon$ defined as follows, \Cref{eq:approximate continuity of predictions} translates to a $(\epsilon_\ell,L_\ell)$-Gradient-Stability guarantee (\Cref{def:gradient stability}) of our predicted gradient $\tilde{\bm g}_\ell(\bm \lambda_\ell)$ (\Cref{eq:lambda O-FTRL-FP} in \Cref{alg:lambda O-FTRL-FP}):
\begin{equation}\label{eq:gradient stability in DualIneff}
\epsilon_\ell=\frac{\sqrt d}{K^2 \epsilon_c \sum_{\ell'<\ell} \lvert \mE_{\ell'}\rvert},\quad L_\ell=\frac{4 \lvert \mE_{\ell}\rvert \sqrt d}{\sum_{\ell'<\ell} \lvert \mE_{\ell'}\rvert}\left (d+\log \frac{1}{\delta_\ell}+\sum_{j=1}^d \log \frac{\sqrt d}{\rho_j \epsilon_\ell}\right ),\quad \forall \ell\ge 2.
\end{equation}
We also set $\epsilon_1=\lVert \bm \rho^{-1}\rVert_2$ and $L_1=0$ for consistency (since $\tilde{\bm g}_1(\bm \lambda)$ is always $\bm 0$, this happens w.p. 1).

Under this $\epsilon_\ell$ and $L_\ell$, we call the event defined in \Cref{eq:approximate continuity of predictions} $\mathcal G_\ell$. Since $\nabla_{\bm \lambda} F_\ell(\bm \lambda)=\sum_{t\in \mE_\ell}(\bm \rho-\bm c_{t,i_t})$, $\nabla_{\bm \lambda}(\tilde{\bm g}_\ell(\bm \lambda_\ell)^\trans \bm \lambda)=\tilde{\bm g}_\ell(\bm \lambda_\ell)$, the loss $F_\ell(\bm \lambda)$ is $\sqrt d \lvert \mE_\ell\rvert$-Lipschitz since $(\bm \rho-\bm c_{t,i_t})\in [0,1]^d$, the loss and predictions are $0$-smooth w.r.t. $\bm \lambda$ (both linear), and the dual norm of $\ell_2$-norm is itself, conditional on $\mathcal G_1,\ldots,\mathcal G_L$---which happens with probability at least $1-\sum_{\ell=1}^L \delta_\ell$ (the conditioning is valid because $\mathcal G_\ell$ only depends on a fixed $\frac{\epsilon_\ell}{2}$-net of $\bm \Lambda$ and historical reports and consumptions, and is thus measurable before the start of epoch $\mE_\ell$)---the \oftrlfp guarantee in \Cref{lem:O-FTRL lemma} gives (again using the fact that $\mathcal L_v\le L$ a.s.)
\begin{equation}\begin{aligned}
\E\left [\sup_{\bm \lambda^\ast\in \bm \Lambda} \sum_{\ell=1}^{\mathcal L_v} \left (\sum_{t\in \mE_\ell} (\bm \rho-\bm c_{t,i_t})\right )^\trans (\bm \lambda_\ell-\bm \lambda^\ast)\right ]&\le \underbrace{\sup_{\bm \lambda^\ast\in \bm \Lambda}\frac{\Psi(\bm \lambda^\ast)}{\eta_L}}_{\texttt{Diameter}}+\sum_{\ell=1}^L \eta_\ell \underbrace{\E\left [\left \lVert \sum_{t\in \mE_\ell} (\bm \rho-\bm c_{t,i_t})-\tilde{\bm g}_\ell(\bm \lambda_\ell)\right \rVert_2^2\right ]}_{\texttt{Stability}_\ell} \\
&\quad +\underbrace{\sum_{\ell=1}^L \sqrt{d} \lvert \mE_\ell\rvert \eta_\ell L_\ell+\left (\sum_{\ell=1}^L \delta_\ell\right ) 2T \lVert \bm \rho^{-1}\rVert_1}_{\texttt{Fixed Point Error}},\label{eq:online learning regret decomposuition in O-FTRL-FP}
\end{aligned}\end{equation}
where the last term considers the failure probability of (any of) $\mathcal G_1,\ldots,\mathcal G_L$, in which case we use the trivial bound that $\sum_{t=1}^{T}(\bm \rho-\bm c_{t,i_t})^\trans (\bm \lambda_t-\bm \lambda^\ast)\le T \cdot \left (\max_{t,i} \lVert \bm \rho-\bm c_{t,i}\rVert_\infty\right )\cdot \left (2\sup_{\bm \lambda\in \bm \Lambda} \lVert \bm \lambda\rVert_1\right )\le 2T \lVert \bm \rho^{-1}\rVert_1$.

\paragraph{Step 2: Control the \texttt{Stability} terms.}
For notational convenience, we add a superscript $u$ to the $\tilde{\bm g}_\ell(\bm \lambda)$ defined in \Cref{eq:lambda O-FTRL-FP} to highlight it is yielded from reports $\{\bm u_\tau\}_{\ell'<\ell,\tau\in \mE_{\ell'}}$. We then define $\tilde{\bm g}_\ell^{v}(\bm \lambda)$ using the true values $\{\bm v_\tau\}_{\ell'<\ell,\tau\in \mE_{\ell'}}$ and $\tilde{\bm g}_\ell^{\ast}(\bm \lambda)$ using the underlying true distributions $\mV = \{\mV_i\}_{i\in [K]}$ and $\mC = \{\mC_i\}_{i\in [K]}$:
\begin{equation}\begin{aligned}
\tilde{\bm g}_\ell^{{\color{magenta}u}}(\bm \lambda)&=\lvert \mE_\ell\rvert\cdot \frac{1}{\sum_{\ell'<\ell} \lvert \mE_{\ell'}\rvert} \sum_{\ell'<\ell,\tau \in \mE_{\ell'}} \left (\bm \rho-\bm c_{\tau,\tilde i_\tau^{{\color{magenta} u}}(\bm \lambda)}\right ), && \tilde i_\tau^{{\color{magenta} u}}(\bm \lambda)= \argmax_{i \in \{0\} \cup [K]} \left( {\color{magenta}u_{\tau,i}} - \bm \lambda^\trans \bm{c}_{\tau,i} \right); \\
\tilde{\bm g}_\ell^{{\color{magenta}v}}(\bm \lambda)&=\lvert \mE_\ell\rvert\cdot \frac{1}{\sum_{\ell'<\ell} \lvert \mE_{\ell'}\rvert} \sum_{\ell'<\ell,\tau \in \mE_{\ell'}} \left (\bm \rho-\bm c_{\tau,\tilde i_\tau^{{\color{magenta} v}}(\bm \lambda)}\right ),&& \tilde i_\tau^{{\color{magenta} v}}(\bm \lambda)= \argmax_{i \in \{0\} \cup [K]} \left( {\color{magenta}v_{\tau,i}} - \bm \lambda^\trans \bm{c}_{\tau,i} \right); \\
\tilde{\bm g}_\ell^{{\color{magenta}\ast}}(\bm \lambda)&=\lvert \mE_\ell\rvert\cdot \E_{\bm v_\ast\sim \mV,\bm c_\ast\sim \mC}\left [\bm \rho-\bm c_{\ast,\tilde i^{{\color{magenta} \ast}}(\bm \lambda)}\right ],&& \tilde i^{{\color{magenta} \ast}}(\bm \lambda)=\argmax_{i\in \{0\} \cup [K]}\left ({\color{magenta}v_{\ast,i}}-\bm \lambda^\trans \bm c_{\ast,i}\right ). \label{eq:formal loss distributions}
\end{aligned}\end{equation}

We now upper bound the $\texttt{Stability}_\ell=\E\left [\left \lVert \sum_{t\in \mE_\ell} (\bm \rho-\bm c_{t,i_t})-\tilde{\bm g}_\ell^{u}(\bm \lambda_\ell)\right \rVert_2^2\right ]$ term in \Cref{eq:online learning regret decomposuition in O-FTRL-FP} as
\begin{equation}\begin{aligned}
&\le 2\underbrace{\E\left [\left \lVert \sum_{t\in \mE_\ell} (\bm \rho-\bm c_{t,i_t})-\tilde{\bm g}_\ell^{\ast}(\bm \lambda_\ell)\right \rVert_2^2\right ]}_{\text{Primal Allocations}}+2\underbrace{\E\left [\lVert \tilde{\bm g}_\ell^{\ast}(\bm \lambda_\ell)-\tilde{\bm g}_\ell^{v}(\bm \lambda_\ell)\rVert_2^2\right ]}_{\text{Empirical Estimation}}+2\underbrace{\E\left [\lVert \tilde{\bm g}_\ell^{v}(\bm \lambda_\ell)-\tilde{\bm g}_\ell^{u}(\bm \lambda_\ell)\rVert_2^2\right ]}_{\text{Untruthful Reports}}.\label{eq:stability decomposition}
\end{aligned}\end{equation}

For the special case where $\ell=1$ (no historical data), we trivially bound $\texttt{Stability}_1\le d \lvert \mE_1\rvert^2$. Otherwise, in \Cref{lem:stability term 1,lem:stability term 2,lem:stability term 3}, we control these terms one by one: In \Cref{lem:stability term 1}, we prove
\begin{equation}\begin{aligned}
&\E\left [\left \lVert \sum_{t\in \mE_\ell} (\bm \rho-\bm c_{t,i_t})-\tilde{\bm g}_\ell^{\ast}(\bm \lambda_\ell)\right \rVert_2^2\right ]\le (d+3)\lvert \mE_\ell\rvert+N_\ell^2,\label{eq:good event of stability terms 1}
\end{aligned}\end{equation}
where $N_\ell$ is defined in \Cref{eq:N and M}.
In \Cref{lem:stability term 2,lem:stability term 3}, by covering arguments, we drive
\begin{equation}\begin{aligned}
&\quad \sup_{\bm \lambda\in \bm \Lambda} \left (\lVert \tilde{\bm g}_\ell^{\ast}(\bm \lambda)-\tilde{\bm g}_\ell^{v}(\bm \lambda)\rVert_2^2 + \lVert \tilde{\bm g}_\ell^{v}(\bm \lambda)-\tilde{\bm g}_\ell^{u}(\bm \lambda)\rVert_2^2\right ) \\
&\le 7d \lvert \mE_\ell\rvert^2 \cdot (K^2 \epsilon \epsilon_c)^2 + 10 d \lvert \mE_\ell\rvert^2 \cdot \frac{\log \frac{1}{\delta}+\sum_{j=1}^d \log\frac{d}{\rho_j \epsilon}}{\sum_{\ell'=1}^{\ell-1} \lvert \mE_{\ell'}\rvert}+\frac{d \lvert \mE_\ell\rvert^2}{(\sum_{\ell'=1}^{\ell-1} \lvert \mE_{\ell'}\rvert)^2} M_\ell^2,\label{eq:good event of stability terms 2 and 3}
\end{aligned}\end{equation}
where $M_\ell$ is also defined in \Cref{eq:N and M}.

\paragraph{Step 3: Plug \texttt{Stability} back to O-FTRL-FP guarantee.}
We plug \Cref{eq:good event of stability terms 1,eq:good event of stability terms 2 and 3} into the online learning regret bound derived in \Cref{eq:online learning regret decomposuition in O-FTRL-FP}. For every $\ell\in [L]$, \Cref{eq:good event of stability terms 2 and 3} happen with probability $1-6\delta$; in case it does not hold, we use the trivial bound that $\lVert \tilde{\bm g}_\ell^{\ast}(\bm \lambda)-\tilde{\bm g}_\ell^{v}(\bm \lambda)\rVert_2^2 + \lVert \tilde{\bm g}_\ell^{v}(\bm \lambda)-\tilde{\bm g}_\ell^{u}(\bm \lambda)\rVert_2^2\le 2d\lvert \mE_\ell\rvert^2$. Therefore, \Cref{eq:online learning regret decomposuition in O-FTRL-FP} translates to (where we also used the definitions of $\epsilon_\ell$ and $L_\ell$ from \Cref{eq:gradient stability in DualIneff})
\begin{equation*}\begin{aligned}
\mR_L^{\text{OL}} &\le \sup_{\bm \lambda^\ast\in \bm \Lambda} \frac{\Psi(\bm \lambda^\ast)}{\eta_L} + d \eta_1 \lvert \mE_1\rvert^2+ 2 \sum_{\ell=2}^L \eta_\ell \Bigg ((d+3) \lvert \mE_\ell\rvert+N_\ell^2 + 7d \lvert \mE_\ell\rvert^2 \cdot (K^2 \epsilon \epsilon_c)^2 \\
&\qquad \qquad \qquad \qquad \qquad + 10 d \lvert \mE_\ell\rvert^2 \cdot \frac{\log \frac{1}{\delta}+\sum_{j=1}^d \log\frac{d}{\rho_j \epsilon}}{\sum_{\ell'=1}^{\ell-1} \lvert \mE_{\ell'}\rvert} +\frac{d \lvert \mE_\ell\rvert^2}{(\sum_{\ell'=1}^{\ell-1} \lvert \mE_{\ell'}\rvert)^2} M_\ell^2 + 6\delta \cdot 2d\lvert \mE_\ell\rvert^2\Bigg )\\
&\quad +\sum_{\ell=2}^L \sqrt d \lvert \mE_\ell\rvert \eta_\ell \frac{4 \lvert \mE_{\ell}\rvert \sqrt d}{\sum_{\ell'<\ell} \lvert \mE_{\ell'}\rvert}\left (d+\log \frac{1}{\delta_\ell}+\sum_{j=1}^d \log \frac{\sqrt d}{\rho_j \epsilon_\ell}\right )+\left (\sum_{\ell=1}^L \delta_\ell\right ) 2T \lVert \bm \rho^{-1}\rVert_1,
\end{aligned}\end{equation*}
where $\{\delta_\ell\}_{\ell\in [L]}$, $\epsilon$, and $\delta$ are some parameters that we can tune.
Letting $\epsilon=\frac{1}{\sqrt T K^2 \epsilon_c}$, $\delta=\frac{1}{6dT}$, and $\delta_\ell=\frac{1}{TL}$ (we did not make every effort to minimize $\mR_L^{\text{OL}}$; this tuning only focuses on the polynomial dependencies on $T$ and $L$) and simplifying using the fact that $\sum_{\ell=1}^L \lvert \mE_\ell\rvert\le T$, we get
\begin{equation}\begin{aligned}
\mR_L^{\text{OL}}
&\le \sup_{\bm \lambda^\ast\in \bm \Lambda} \frac{\Psi(\bm \lambda^\ast)}{\eta_L}+d \eta_1\lvert \mE_1\rvert^2+\sum_{\ell=2}^L \eta_\ell \frac{\lvert \mE_\ell\rvert^2}{\sum_{\ell'<\ell} \lvert \mE_{\ell'}\rvert} \left (4d^2 + 24d \log(dTL) + 24d \sum_{j=1}^d \log \frac{dK^2 \epsilon_c T}{\rho_j}\right ) \\
&\quad +\sum_{\ell=2}^L \eta_\ell \lvert \mE_\ell\rvert (16d+10)+\sum_{\ell=2}^L \eta_\ell\left (2N_\ell^2 + \frac{2d \lvert \mE_\ell\rvert^2 M_\ell^2}{(\sum_{\ell'=1}^{\ell-1} \lvert \mE_{\ell'}\rvert)^2}\right )+2\lVert \bm \rho^{-1}\rVert_1. \label{eq:formal DualUpd in O-FTRL-FP}
\end{aligned}\end{equation}

\paragraph{Step 4: Plug in $\Psi$ and $\{\eta_\ell\}_{\ell\in [L]}$.}
With $\Psi(\bm \lambda)=\frac 12 \lVert \bm \lambda\rVert_2^2$ and $\eta_\ell=\frac{\lVert \bm \rho^{-1}\rVert_2}{2d}(\sum_{\ell'=1}^\ell \lvert \mE_{\ell'}\rvert)^{-1/2}$, under the assumption that $\lvert \mE_1\rvert\le \lvert \mE_2\rvert\le \cdots \le \lvert \mE_L\rvert$ (i.e., epoch lengths are non-decreasing), \Cref{eq:formal DualUpd in O-FTRL-FP} is simplified to
\begin{equation*}\begin{aligned}
\mR_L^{\text{OL}}&\le \lVert \bm \rho^{-1}\rVert_2 d \left (\sum_{\ell=1}^L \lvert \mE_\ell\rvert\right )^{1/2} + d \lvert \mE_1\rvert^{1.5} +\sum_{\ell=2}^L \frac{2N_\ell^2+2d \lvert \mE_\ell\rvert^2 M_\ell^2/(\sum_{\ell'=1}^{\ell-1} \lvert \mE_{\ell'}\rvert)^2}{2d (\sum_{\ell'=1}^\ell \lvert \mE_{\ell'}\rvert)^{1/2}}\lVert \bm \rho^{-1}\rVert_2 + 2\lVert \bm \rho^{-1}\rVert_1 \\
&\quad +\sum_{\ell=2}^L \frac{\lvert \mE_\ell\rvert}{(\sum_{\ell'=1}^\ell \lvert \mE_{\ell'}\rvert)^{1/2}} \left (2d + 12 \log(dTL) + 12 \sum_{j=1}^d \log \frac{dK^2 \epsilon_c T}{T\rho_j}+18\right ) \lVert \bm \rho^{-1}\rVert_2.
\end{aligned}\end{equation*}
Using the folklore summation lemma that $\sum_{t=1}^T \frac{x_t}{\sqrt{\sum_{s=1}^t x_s}}\le 2 \sqrt{\sum_{t=1}^T x_t}$ for all $x_1,x_2,\ldots,x_T\in \mathbb R_{\ge 0}$ \citep[Lemma 4]{duchi2011adaptive}, the last term is bounded by $2 \sqrt{\sum_{\ell=1}^L \lvert \mE_\ell\rvert} \big (4d^2 + 24d \log(dTL) + 24d \sum_{j=1}^d \log \frac{dK^2 \epsilon_c T}{T\rho_j}+16d+10\big )$.
To translate the online learning regret $\mR_L^{\text{OL}}$ to the $\E[\dualreg]$ guarantee, we use \Cref{lem:DualUpd to regret}:
\begin{equation*}\begin{aligned}
\E[\dualreg]&\le 2 \sqrt{T} \left (3d + 12 \log(dTL) + 12 \sum_{j=1}^d \log \frac{dK^2 \epsilon_c T}{T\rho_j}+18\right )\lVert \bm \rho^{-1}\rVert_1+d \lvert \mE_1\rvert^{1.5} \\
&\quad +\sum_{\ell=2}^L \frac{N_\ell^2/d+\lvert \mE_\ell\rvert^2 M_\ell^2/(\sum_{\ell'=1}^{\ell-1} \lvert \mE_{\ell'}\rvert)^2}{(\sum_{\ell'=1}^\ell \lvert \mE_{\ell'}\rvert)^{1/2}} \lVert \bm \rho^{-1}\rVert_1 + \left (3+\max_{\ell\in [L]} \lvert \mE_\ell\rvert+\sum_{\ell=1}^L (N_\ell+3)\right )\lVert \bm \rho^{-1}\rVert_1,
\end{aligned}\end{equation*}
where we used the facts that $\sum_{\ell=1}^L \lvert \mE_\ell\rvert=T$ and $\lVert \bm \rho^{-1}\rVert_2\le \lVert \bm \rho^{-1}\rVert_1$. This finishes the proof.
\end{proof}

\begin{lemma}[Gradient-Stability of Predictions]\label{lem:approximate continuity of predictions}
For any $2\le \ell\le L$, $\epsilon>0$, and $\delta\in (0,1)$, w.p. $1-\delta$,
\begin{equation*}\begin{aligned}
&\lVert \tilde{\bm g}_\ell(\bm \lambda_1)-\tilde{\bm g}_\ell(\bm \lambda_2)\rVert_2 \le 4\lvert \mE_{\ell}\rvert K^2 \epsilon \epsilon_c d + \frac{4\lvert \mE_\ell\rvert (\log \frac 1\delta+\sum_{j=1}^d \log \frac{\sqrt d}{\rho_j \epsilon})}{\sum_{\ell'<\ell} \lvert \mE_{\ell'}\rvert}\sqrt d
\end{aligned}\end{equation*}
holds for any $\bm \lambda_1,\bm \lambda_2\in \bm \Lambda$ s.t. $\lVert \bm \lambda_1-\bm \lambda_2\rVert_2\le \epsilon$, where $\tilde{\bm g}_\ell(\bm \lambda)$ is defined in \Cref{eq:lambda O-FTRL-FP} of \Cref{alg:lambda O-FTRL-FP}.
\end{lemma}
\begin{proof}
Take an $\frac{\epsilon}{2}$-net of $\bm \Lambda$ defined in \Cref{lem:covering argument}, which is presented immediately after this proof. Slightly abusing the notations, we denote this $\frac{\epsilon}{2}$-net by $\bm \Lambda_\epsilon$ (note that $\frac \epsilon 2$-nets are also $\epsilon$-nets since $\frac{\epsilon}{2}<\epsilon$). We have $\lvert \bm \Lambda_\epsilon\rvert\le \prod_{j=1}^d (2d/\rho_j \epsilon)$.
Fix a $\bm \lambda_\epsilon\in \bm \Lambda_\epsilon$ and consider the stochastic process $(X_\tau)_{\tau\ge 1}$ adapted to $(\mF_\tau)_{\tau\ge 0}$:
\begin{equation}\label{footnote:uniform smoothness}
X_\tau:=\1[\exists \bm \lambda\in \mathcal B_{\epsilon}(\bm \lambda_\epsilon)\text{ s.t. }\tilde i_\tau(\bm \lambda)\ne \tilde i_\tau(\bm \lambda_\epsilon)],\quad \forall \ell'<\ell,\tau\in \mE_{\ell'},
\end{equation}
where $\mathcal B_\epsilon(\bm \lambda_\epsilon):=\{\bm \lambda\in \bm \Lambda\mid \lVert \bm \lambda-\bm \lambda_\epsilon\rVert_1\le \epsilon\}$ and $(\mF_\tau)_{\tau\ge 0}$ is defined as $\mF_\tau=\sigma(\bigcup_{i\in \{0\}\cup [K]} \mH_{\tau+1,i})$, i.e., the smallest $\sigma$-algebra containing all revealed history up to the end of round $\tau$.
(We remark that ``$\exists \bm \lambda\in \mathcal B_\epsilon(\bm \lambda_\epsilon)$'' clause is pivotal because we cannot afford to take a Union Bound over all $\bm \lambda \in \mathcal B_\epsilon(\bm \lambda_\epsilon)$. The definition of $X_\tau$ in \Cref{footnote:uniform smoothness} ensures the similarity holds uniformly in the neighborhood of $\bm \lambda\in \mathcal B_\epsilon(\bm \lambda_\epsilon)$.)
As $X_\tau$ is $\mF_\tau$-measurable, Multiplicative Azuma-Hoeffding \citep[Lemma 10]{koufogiannakis2014nearly} gives
\begin{equation*}
\Pr\left \{\frac 12\sum_{\ell'<\ell,\tau\in \mE_{\ell'}} X_\tau\ge \sum_{\ell'<\ell,\tau\in \mE_{\ell'}} \E[X_\tau\mid \mF_{\tau-1}] +A \right \}\le \exp\left (-\frac A2\right ),\quad \forall A>0.
\end{equation*}

For any $\ell'<\ell$ and $\tau\in \mE_{\ell'}$, the conditional distribution $\bm u_\tau\mid \bigcup_{i\in \{0\}\cup [K]} \mH_{\tau,i}$ is $\mF_{\tau-1}$-measurable.
And furthermore, $\bm \lambda_\epsilon\in \bm \Lambda_\epsilon$ is fixed before the game and thus also $\mF_{\tau-1}$-measurable. \Cref{lem:covering argument} gives
\begin{equation*}\begin{aligned}
\E[X_\tau\mid \mF_{\tau-1}]&=\Pr\left \{\exists \bm \lambda\in \mathcal B_\epsilon(\bm \lambda_\epsilon)\text{ s.t. }\argmax_{i \in \{0\} \cup [K]} \left( u_{\tau,i} - \bm \lambda^\trans \bm{c}_{\tau,i} \right)\ne \argmax_{i \in \{0\} \cup [K]} \left( u_{\tau,i} - \bm \lambda_\epsilon^\trans \bm{c}_{\tau,i} \right)\right \}\le K^2 \epsilon \epsilon_c,
\end{aligned}\end{equation*}
for any $\ell'<\ell$ and $\tau\in \mE_{\ell'}$; where the $\Pr$ on the RHS is taken w.r.t. the randomness of generating $\bm u_\tau$ according to the conditional joint distribution $\bm u_\tau\mid \bigcup_{i\in \{0\}\cup [K]} \mH_{\tau,i}$ and the independent sampling of $\bm c_\tau\sim \mC$.

Hence, for any failure probability $\delta>0$ that we determine later, with probability $1-\delta$,
\begin{equation*}
\sum_{\ell'<\ell,\tau\in \mE_{\ell'}} \1[\exists \bm \lambda\in \mathcal B_\epsilon(\bm \lambda_\epsilon)\text{ s.t. }\tilde i_\tau(\bm \lambda)\ne \tilde i_\tau(\bm \lambda_\epsilon)]\le 2 \sum_{\ell'<\ell} \lvert \mE_{\ell'}\rvert\cdot K^2\epsilon \epsilon_c+4\log \frac{1}{\delta}.
\end{equation*}

Taking Union Bound over $\bm \lambda_\epsilon\in \bm \Lambda_\epsilon$, w.p. $1-\delta \prod_{j=1}^d(2d/\rho_j\epsilon)$, this event holds for all $\bm \lambda_\epsilon\in \bm \Lambda$ at the same time. For any $\bm \lambda_1,\bm \lambda_2\in \bm \Lambda$ s.t. $\lVert \bm \lambda_1-\bm \lambda_2\rVert_2\le \frac{\epsilon}{2\sqrt d}$ (thus $\lVert \bm \lambda_1-\bm \lambda_2\rVert_1\le \frac \epsilon 2$), take $\bm \lambda_\epsilon\in \bm \Lambda_\epsilon$ such that $\bm \lambda_1\in \mathcal B_{\epsilon/2}(\bm \lambda_\epsilon)$ (recall that $\bm \Lambda_\epsilon$ is in fact a $\frac{\epsilon}{2}$-net). We therefore have $\lVert \bm \lambda_2-\bm \lambda_\epsilon\rVert_1\le \epsilon$, which means $\bm \lambda_1,\bm \lambda_2\in \mathcal B_\epsilon(\bm \lambda_\epsilon)$. Thus
\begin{equation*}\begin{aligned}
&\left \lVert \sum_{\ell'<\ell,\tau\in \mE_{\ell'}}(\bm \rho-\bm c_{\tau,\tilde i_\tau(\bm \lambda_1)})-\sum_{\ell'<\ell,\tau\in \mE_{\ell'}}(\bm \rho-\bm c_{\tau,\tilde i_\tau(\bm \lambda_2)}) \right \rVert_2\le \left (2 \sum_{\ell'<\ell} \lvert \mE_{\ell'}\rvert\cdot K^2\epsilon \epsilon_c+4\log \frac{1}{\delta}\right )\sqrt d
\end{aligned}\end{equation*}
holds for all $\bm \lambda_1,\bm \lambda_2\in \bm \Lambda$ s.t. $\lVert \bm \lambda_1-\bm \lambda_2\rVert_2\le \frac{\epsilon}{2\sqrt d}$, with probability $1-\delta \prod_{j=1}^d(2d/\rho_j\epsilon)$. This ensures that
\begin{equation*}\begin{aligned}
\lVert \tilde{\bm g}_\ell(\bm \lambda_1)-\tilde{\bm g}_\ell(\bm \lambda_2)\rVert_2 \le \frac{\lvert \mE_\ell\rvert}{\sum_{\ell'<\ell} \lvert \mE_{\ell'}\rvert} \left (2 \sum_{\ell'<\ell} \lvert \mE_{\ell'}\rvert\cdot K^2\epsilon \epsilon_c+4\log \frac{1}{\delta}\right )\sqrt d=2\lvert \mE_{\ell}\rvert K^2 \epsilon \epsilon_c \sqrt d + \frac{4\lvert \mE_\ell\rvert \log \frac 1\delta}{\sum_{\ell'<\ell} \lvert \mE_{\ell'}\rvert}\sqrt d,
\end{aligned}\end{equation*}
holds for all $\bm \lambda_1,\bm \lambda_2\in \bm \Lambda$ s.t. $\lVert \bm \lambda_1-\bm \lambda_2\rVert_2\le \frac{\epsilon}{2\sqrt d}$, with probability $1-\delta \prod_{j=1}^d(2d/\rho_j\epsilon)$. Substituting $\epsilon'=\frac{\epsilon}{2\sqrt d}$ and $\delta'=\delta \prod_{j=1}^d(2d/\rho_j\epsilon)$ gives the conclusion.
\end{proof}

\begin{lemma}[Covering Dual Decision Set]\label{lem:covering argument}
For a fixed $\epsilon>0$, there exists $\bm \Lambda_\epsilon\subseteq \bm \Lambda=\bigotimes_{j=1}^d [0,\rho_j^{-1}]$ with size no more than $\prod_{j=1}^d (d/\rho_j \epsilon)$, such that for all $\bm \lambda\in \bm \Lambda$, there exists some $\bm \lambda_\epsilon\in \bm \Lambda_\epsilon$ such that $\lVert \bm \lambda-\bm \lambda_\epsilon\rVert_1\le \epsilon$.
Let $\mathcal B_\epsilon(\bm \lambda_\epsilon)=\{\bm \lambda\in \bm \Lambda\mid \lVert \bm \lambda-\bm \lambda_\epsilon\rVert_1\le \epsilon\}$ be the neighborhood covered by $\bm \lambda_\epsilon\in \bm \Lambda_\epsilon$. Under \Cref{assump:smooth consumptions}, for all $\bm \lambda_\epsilon\in \bm \Lambda_\epsilon$ and any distribution $\mathcal U\in \triangle([0,1]^K)$, we have
\begin{equation*}
\Pr_{\bm u\sim \mathcal U,\bm c\sim \mC}\left \{\exists \bm \lambda\in \mathcal B_\epsilon(\bm \lambda_\epsilon)\text{ s.t. }\argmax_{i\in \{0\}\cup [K]} (u_i-\bm \lambda^\trans \bm c_i)\ne \argmax_{i\in [K]} (u_i-\bm \lambda_\epsilon^\trans \bm c_i)\right \}\le K^2(\epsilon\cdot \epsilon_c).
\end{equation*}
\end{lemma}
\begin{proof}
The first claim follows from standard covering arguments over the bounded set $\bm \Lambda=\bigotimes_{j=1}^d [0,\rho_j^{-1}]$ \citep[\S5.1]{wainwright2019high}. For the second part, we make use of \Cref{assump:smooth consumptions}: For any fixed $i\ne j\in \{0\}\cup [K]$,
\begin{equation*}\begin{aligned}
&\quad \Pr_{\bm u\sim \mathcal U,\bm c\sim \mC}\left \{\exists \bm \lambda\in \mathcal B_\epsilon(\bm \lambda_\epsilon)\text{ s.t. }(u_i-\bm \lambda_\epsilon^\trans \bm c_i>u_j-\bm \lambda_\epsilon^\trans \bm c_j)\wedge (u_i-\bm \lambda^\trans \bm c_i<u_j-\bm \lambda^\trans \bm c_j)\right \}\\
&\overset{(a)}{\le} \Pr_{\bm u\sim \mathcal U,\bm c\sim \mC}\left \{0\le (u_i-\bm \lambda_\epsilon^\trans \bm c_i)-(u_j-\bm \lambda_\epsilon^\trans \bm c_j)\le \epsilon\right \}\overset{(b)}{\le} \epsilon\cdot \epsilon_c,
\end{aligned}\end{equation*}
where (a) uses $\lvert \langle \bm \lambda-\bm \lambda_\epsilon,\bm c_i-\bm c_j\rangle\rvert\le \lVert \bm \lambda-\bm \lambda_\epsilon\rVert_1\cdot \lVert \bm c_i-\bm c_j\rVert_\infty\le \epsilon$ for all $\bm \lambda\in \mathcal B_\epsilon(\bm \lambda_\epsilon)$, and (b) uses \Cref{assump:smooth consumptions} and the independence of $\bm c_i$ and $\bm c_j$: If either $i$ or $j$ equals $0$, then the inequality directly holds from applying \Cref{assump:smooth consumptions} to the other consumption distribution. Otherwise, i.e., suppose that both $i$ and $j$ are non-zero. For any two independent real-valued random variables $X\perp Y$ have their PDFs $f_X$ and $f_Y$ uniformly bounded by $\epsilon_c$, their difference $X-Y$ is also a random variable with PDF uniformly bounded by $\epsilon_c$:
\begin{equation*}
f_{X-Y}(z)=\int_{-\infty}^\infty f_X(z+y) f_Y(y)dy\le \epsilon_c\int_{-\infty}^\infty f_Y(y)dy\le \epsilon_c,\quad \forall z\in \mathbb R,
\end{equation*}
which means (b) is still true.
Applying Union Bound to the $K^2$ pairs of $(i,j)\in (\{0\}\cup [K])^2$ such that $i\ne j$ gives the second conclusion.
\end{proof}

\begin{lemma}[Primal Allocations]\label{lem:stability term 1}
For any epoch $\ell\in [L]$ where $\bm \lambda_\ell$ is determined by \Cref{alg:lambda O-FTRL-FP},
\begin{equation*}
\E\left [\left \lVert \sum_{t\in \mE_\ell} (\bm \rho-\bm c_{t,i_t})-\tilde{\bm g}_\ell^{\ast}(\bm \lambda_\ell)\right \rVert_2^2\right ]\le (d+3)\lvert \mE_\ell\rvert+N_\ell^2,
\end{equation*}
where $\tilde{\bm g}_\ell^\ast$ is defined in \Cref{eq:formal loss distributions} and $N_\ell$ is defined in \Cref{eq:N and M}.
\end{lemma}
\begin{proof}
We first control $\E [ \lVert \sum_{t\in \mE_\ell} (\bm \rho-\bm c_{t,\tilde i_t^\ast})-\tilde{\bm g}_\ell^{\ast}(\bm \lambda_\ell) \rVert_2^2 ]$, i.e., the squared $\ell_2$-error of $\lvert \mE_\ell\rvert$ random vectors from their mean, which is of order $\lvert \mE_\ell\rvert$ as they are i.i.d. We then relate it to $\E [ \lVert \sum_{t\in \mE_\ell} (\bm \rho-\bm c_{t,i_t})-\tilde{\bm g}_\ell^{\ast}(\bm \lambda_\ell) \rVert_2^2 ]$ by utilizing the similarity between $\{\tilde i_t^\ast\}_{t\in \mE_\ell}$ and $\{i_t\}_{t\in \mE_\ell}$ (\Cref{thm:InterEpoch guarantee}).

\paragraph{Step 1: Control $\E [ \lVert \sum_{t\in \mE_\ell} (\bm \rho-\bm c_{t,\tilde i_t^\ast})-\tilde{\bm g}_\ell^{\ast}(\bm \lambda_\ell) \rVert_2^2 ]$.}
Recall the definition of $\{\tilde i_t^\ast\}_{t\in [T]}=\argmax_{i\in \{0\} \cup [K]} (v_{t,i}-\bm \lambda_\ell^\trans \bm c_{t,i})$.
Since $\bm v_t$ and $\bm c_t$ are i.i.d. samples from $\mV$ and $\mC$, $\Law(\tilde i_t^\ast)=\Law(\tilde i^{\ast}(\bm \lambda_\ell))$ for all $t\in \mE_\ell$. Therefore,
\begin{equation*}
\E\left [\bm \rho-\bm c_{t,{\tilde i_t^\ast}}\right ]=\E_{\bm v_\ast\sim \mV,\bm c_\ast\sim \mC}\left [\bm \rho-\bm c_{\ast,\tilde i^{\ast}(\bm \lambda_\ell)}\right ]=\frac{1}{\lvert \mE_\ell\rvert}\tilde{\bm g}_\ell^\ast(\bm \lambda_\ell),\quad \forall t\in \mE_\ell,
\end{equation*}
where the last equation follows by definition of $\tilde{\bm g}_\ell^\ast(\bm \lambda_\ell)$.
Since for a $d$-dimensional random vector $X$, $\E[\lVert X-\E[X]\rVert_2^2]=\E[\sum_{i=1}^d (X_i-\E[X_i])^2]=\sum_{i=1}^d \text{Var}(X_i)=\text{Tr}(\text{Cov}(X))$ ($\text{Tr}$ is the trace and $\text{Cov}$ is covariance),
\begin{equation*}
\E\left [\left \lVert \sum_{t\in \mE_\ell} (\bm \rho-\bm c_{t,\tilde i_t^\ast})-\tilde{\bm g}_\ell^{\ast}(\bm \lambda_\ell)\right \rVert_2^2\right ]=\text{Tr}\left ( \lvert \mE_\ell\rvert \cdot \operatornamewithlimits{Cov}_{\bm v_\ast\sim \mV,\bm c_\ast\sim \mC}\left (\bm \rho-\bm c_{\ast,\tilde i^{\ast}(\bm \lambda_\ell)}\right )\right )\le \lvert \mE_\ell\rvert d,
\end{equation*}
using the fact that $\tilde i_t^\ast$'s are independent from each other and that $\bm \rho$ and $\bm c_{\ast,i}$ are all within $[0,1]^d$.

\paragraph{Step 2: Relate $\sum_{t\in \mE_\ell} (\bm \rho-\bm c_{t,\tilde i_t^\ast})$ to $\sum_{t\in \mE_\ell} (\bm \rho-\bm c_{t,i_t})$.}
Recall \Cref{eq:InterEpoch takeaway} from \Cref{thm:InterEpoch guarantee}:
\begin{equation}\label{eq:InterEpoch takeaway used in Stability Term 1}
\Pr\left \{\sum_{t\in \mE_\ell} \1[i_t\ne \tilde i_t^\ast]>N_\ell\right \}\le \frac{3}{\lvert \mE_\ell\rvert},
\end{equation}
where $N_\ell$ is defined in \Cref{eq:N and M}.
As $(\bm  \rho-\bm c_{t,i})\in [-1, 1]^d$, we have
\begin{equation*}\begin{aligned}
\E\left [\left \lVert \sum_{t\in \mE_\ell} (\bm \rho-\bm c_{t,i_t})-\tilde{\bm g}_\ell^{\ast}(\bm \lambda_\ell)\right \rVert_2^2\right ]&\le \E\left [\left \lVert \sum_{t\in \mE_\ell} (\bm \rho-\bm c_{t,\tilde i_t^\ast})-\tilde{\bm g}_\ell^{\ast}(\bm \lambda_\ell)\right \rVert_2^2\right ]+\E\left [d \left (\sum_{t\in \mE_\ell}\1[i_t\ne \tilde i_t^\ast]\right )^2\right ]\\
&\le \lvert \mE_\ell\rvert d+N_\ell^2+3\frac{\lvert \mE_\ell\rvert^2}{\lvert \mE_\ell\rvert},
\end{aligned}\end{equation*}
where the last term considers the failure probability of \Cref{eq:InterEpoch takeaway used in Stability Term 1}, in which case we use the trivial bound $(\sum_{t\in \mE_\ell}\1[i_t\ne \tilde i_t^\ast])^2\le \lvert \mE_\ell\rvert^2$.
Rearranging gives the desired conclusion.
\end{proof}

\begin{lemma}[Empirical Estimation]\label{lem:stability term 2}
For any $\ell\in [L]$, $\epsilon>0$, and $\delta\in (0,1)$, with probability $1-2\delta$,
\begin{equation*}\begin{aligned}
\lVert \tilde{\bm g}_\ell^{\ast}(\bm \lambda)-\tilde{\bm g}_\ell^{v}(\bm \lambda)\rVert_2^2 &\le 3d \lvert \mE_\ell\rvert^2 \cdot (K^2 \epsilon \epsilon_c)^2 + 6 d \lvert \mE_\ell\rvert^2 \cdot \frac{\log \frac{1}{\delta}+\sum_{j=1}^d \log\frac{d}{\rho_j \epsilon}}{\sum_{\ell'=1}^{\ell-1} \lvert \mE_{\ell'}\rvert}\\
&=\Otil\left (\lvert \mE_\ell\rvert^2 \epsilon^2 +\frac{\lvert \mE_\ell\rvert^2}{\sum_{\ell'=1}^{\ell-1} \lvert \mE_{\ell'}\rvert} \left (\log \frac 1\delta + \sum_{j=1}^d \log \frac{d}{\rho_j \epsilon}\right )\right ),\quad \forall \bm \lambda\in \bm \Lambda,
\end{aligned}\end{equation*}
where $\tilde{\bm g}_\ell^\ast$ and $\tilde{\bm g}_\ell^v$ are defined in \Cref{eq:formal loss distributions}.
\end{lemma}
\begin{proof}
In \Cref{lem:stability term 1}, we applied concentration inequalities to the realized duals $\bm \lambda_\ell$. We cannot do the same: the realized values $\bm v_\tau$---used to compute $\tilde{\bm g}_\ell^{v}(\bm \lambda_\ell)$---were drawn in the past, and $\bm \lambda_\ell$ itself is computed based on reports dependent on these values. Thus, ``conditional on $\tilde{\bm g}_\ell^{\ast}(\bm \lambda_\ell) \approx \tilde{\bm g}_\ell^{v}(\bm \lambda_\ell)$'' is invalid because it introduces a dependence on future information from the perspective of those past realizations.

We instead establish uniform concentration over all $\bm \lambda \in \bm \Lambda$ by discretization. Specifically, fix an $\epsilon$-net $\bm \Lambda_\epsilon \subseteq \bm \Lambda$ and we show that for every $\bm \lambda_\epsilon \in \bm \Lambda_\epsilon$, the approximation $\tilde{\bm g}_\ell^{\ast}(\bm \lambda_\epsilon) \approx \tilde{\bm g}_\ell^{v}(\bm \lambda_\epsilon)$ holds with high probability. We then extend this to all $\bm \lambda \in \bm \Lambda$ by considering the stochastic process $\{\1[\exists \bm \lambda\in \mathcal B_\epsilon(\bm \lambda_\epsilon)\text{ s.t. }\tilde i_\tau^v(\bm \lambda)\ne \tilde i_\tau^v(\bm \lambda_\epsilon)]\}_{\ell'<\ell,\tau\in \mE_{\ell'}}$ where $\mathcal B_\epsilon(\bm \lambda_\epsilon)$ is the $\epsilon$-radius ball centered at $\bm \lambda_\epsilon$. The proof goes in three steps.

\paragraph{Step 1: Cover $\bm \Lambda$ with an $\epsilon$-net.}
From \Cref{lem:covering argument}, for any $\epsilon > 0$, there exists an $\epsilon$-net $\bm \Lambda_\epsilon \subseteq \bm \Lambda$ of size $\O((d/\epsilon)^d)$, such that every $\bm \lambda \in \bm \Lambda$ has some $\bm \lambda_\epsilon \in \bm \Lambda_\epsilon$ with $\lVert \bm \lambda - \bm \lambda_\epsilon \rVert_1 \le \epsilon$. We remark that our final guarantee does not have a $d$-exponent, because the dependency on $\lvert \bm \Lambda_\epsilon\rvert$ is logarithmic.

\paragraph{Step 2: Yield concentration for any fixed $\bm \lambda_\epsilon \in \bm \Lambda_\epsilon$.}
Fix $\bm \lambda_\epsilon \in \bm \Lambda_\epsilon$. Consider the vector random variables
\begin{equation*}
\bm x_\tau := (\bm \rho - \bm c_{\tau, \tilde i_\tau^{v}(\bm \lambda_\epsilon)}) - \E_{\bm v_\ast \sim \mV, \bm c_\ast \sim \mC}\left [\bm \rho - \bm c_{\ast, \tilde i^{\ast}(\bm \lambda_\epsilon)}\right ], \quad \forall \ell' < \ell, \tau \in \mE_{\ell'}.
\end{equation*}
Since $\bm x_\tau$ only depends on $\bm v_\tau$ and $\bm c_\tau$---which are i.i.d. samples from $\mV$ and $\mC$---these vectors are i.i.d., zero-mean. We further have $\lVert \bm x_\tau \rVert_2 \le \sqrt{d}$ almost surely, because $\bm \rho-\bm c_\tau\in [-1,1]^d$. Applying the vector Bernstein inequality \citep[Lemma 18]{kohler2017sub} gives:
\begin{equation*}
\Pr\left \{\lVert \tilde{\bm g}_\ell^{\ast}(\bm \lambda_\epsilon)-\tilde{\bm g}_\ell^{v}(\bm \lambda_\epsilon)\rVert_2\ge \lvert \mE_\ell\rvert c\right \}=\Pr\left \{\left \lVert \frac{\sum_{\ell'<\ell,\tau\in \mE_{\ell'}} \bm x_\tau}{\sum_{\ell'<\ell} \lvert \mE_{\ell'}\rvert}  \right \rVert_2\ge  c\right \}\le \exp\left (-\sum_{\ell'<\ell} \lvert \mE_{\ell'}\rvert \cdot \frac{c^2}{8d}+\frac 14\right ).
\end{equation*}
Taking a union bound over all $\bm \lambda_\epsilon \in \bm \Lambda_\epsilon$, we obtain that
\begin{equation*}\begin{aligned}
\Pr\left \{\max_{\bm \lambda_\epsilon\in \bm \Lambda_\epsilon} \lVert \tilde{\bm g}_\ell^{\ast}(\bm \lambda_\epsilon)-\tilde{\bm g}_\ell^{v}(\bm \lambda_\epsilon)\rVert_2^2\ge \lvert \mE_\ell\rvert^2 c^2\right \}\le \prod_{i=1}^d \frac{d}{\rho_j \epsilon} \cdot \exp \left (-2 \sum_{\ell'=1}^{\ell-1} \lvert \mE_{\ell'}\rvert \cdot \frac{c^2}{8d} + \frac 14 \right ),\quad \forall c>0.
\end{aligned}\end{equation*}

Therefore, for the given failure probability $\delta$, with probability at least $1 - \delta$,
\begin{equation}\label{eq:diff of F and Fv on lambda_eps}
\lVert \tilde{\bm g}_\ell^{\ast}(\bm \lambda_\epsilon) - \tilde{\bm g}_\ell^{v}(\bm \lambda_\epsilon) \rVert_2^2
\le \lvert \mE_\ell \rvert^2 \cdot 4d \cdot \frac{\log \frac{1}{\delta} + \sum_{j=1}^d \log \frac{d}{\rho_j \epsilon}}{\sum_{\ell'=1}^{\ell-1} \lvert \mE_{\ell'} \rvert},
\quad \forall \bm \lambda_\epsilon \in \bm \Lambda_\epsilon.
\end{equation}

\paragraph{Step 3: Extend the similarity to all $\bm \lambda \in \bm \Lambda$.}
We now fix a $\bm \lambda_\epsilon\in \bm \Lambda_\epsilon$ and derive a uniform concentration guarantee for all $\bm \lambda\in \mathcal B_\epsilon(\bm \lambda_\epsilon)$.
Using boundedness $\lVert \bm \rho - \bm c_{\tau,i} \rVert_2 \le \sqrt{d}$, we have:
\begin{equation*}\begin{aligned}
\lVert \tilde{\bm g}_\ell^{v}(\bm \lambda) - \tilde{\bm g}_\ell^{v}(\bm \lambda_\epsilon) \rVert_2^2
&\le d \left( \frac{\lvert \mE_\ell \rvert}{\sum_{\ell'<\ell} \lvert \mE_{\ell'} \rvert}
\sum_{\ell'<\ell, \tau \in \mE_{\ell'}} \1[\tilde i_\tau^{v}(\bm \lambda) \ne \tilde i_\tau^{v}(\bm \lambda_\epsilon)] \right)^2, \\
\lVert \tilde{\bm g}_\ell^{\ast}(\bm \lambda) - \tilde{\bm g}_\ell^{\ast}(\bm \lambda_\epsilon) \rVert_2^2
&\le d \left (\lvert \mE_\ell \rvert \cdot \Pr\{\tilde i^{\ast}(\bm \lambda) \ne \tilde i^{\ast}(\bm \lambda_\epsilon)\}\right )^2,\quad \forall \bm \lambda\in \mathcal B_\epsilon(\bm \lambda_\epsilon).
\end{aligned}\end{equation*}

From \Cref{lem:covering argument}, we know
\begin{equation}\begin{aligned}\label{eq:close dual variables give similar argmax}
\Pr_{\bm v\sim \mV,\bm c\sim \mC}\left \{\exists \bm \lambda\in \mathcal B_\epsilon(\bm \lambda_\epsilon)\text{ s.t. }\tilde i^{\ast}(\bm \lambda)\ne \tilde i^{\ast}(\bm \lambda_\epsilon)\right \}\le K^2(\epsilon\cdot \epsilon_c),\quad \forall \bm \lambda_\epsilon\in \bm \Lambda_\epsilon,
\end{aligned}\end{equation}
where the $\exists \bm \lambda\in \mathcal B_\epsilon(\bm \lambda_\epsilon)$ clause is important since it ensures ``uniform smoothness'' in the neighborhood of $\bm \lambda_\epsilon$. If we instead fix a $\bm \lambda$ and its corresponding $\bm \lambda_\epsilon$ and apply concentration to this specific $\bm \lambda$, we need to do an prohibitively expensive Union Bound afterwards; see also \Cref{footnote:uniform smoothness} in the proof of \Cref{lem:approximate continuity of predictions}.

Hence the error between $\tilde{\bm g}_\ell^\ast(\bm \lambda)$ and $\tilde{\bm g}_\ell^\ast(\bm \lambda_\epsilon)$ is bounded by
\begin{equation}\label{eq:diff of F on lambda_eps and lambda}
\max_{\bm \lambda\in \mathcal B_\epsilon(\bm \lambda_\epsilon)}\lVert \tilde{\bm g}_\ell^{\ast}(\bm \lambda)-\tilde{\bm g}_\ell^{\ast}(\bm \lambda_\epsilon)\rVert_2^2\le d \lvert \mE_\ell\rvert^2\cdot (K^2 \epsilon \epsilon_c)^2,\quad a.s.
\end{equation}

For the error between $\tilde{\bm g}_\ell^v(\bm \lambda)$ and $\tilde{\bm g}_\ell^v(\bm \lambda_\epsilon)$, consider a stochastic process $(X_\tau)_{\tau\ge 1}$ adapted to $(\mF_\tau)_{\tau\ge 0}$:
\begin{equation}\label{footnote:reports measurability in stability}
X_\tau:=\1[\exists \bm \lambda\in \mathcal B_\epsilon(\bm \lambda_\epsilon)\text{ s.t. }\tilde i_\tau^{v}(\bm \lambda)\ne \tilde i_\tau^{v}(\bm \lambda_\epsilon)],\quad \forall \ell'<\ell,\tau\in \mE_{\ell'},
\end{equation}
where $\mF_\tau=\sigma(\bigcup_{i\in \{0\}\cup [K]} \mH_{\tau+1,i})=\sigma(\bm v_1,\ldots,\bm v_\tau,\bm u_1,\ldots,\bm u_\tau,\bm c_1,\ldots,\bm c_\tau,i_1,\ldots,i_\tau,\bm p_1,\ldots,\bm p_\tau)$ is the smallest $\sigma$-algebra containing all generated history up to the end of round $\tau$. (In fact, the $X_\tau$'s here are i.i.d. since $\bm v_\tau\sim \mV$ and $\bm c_\tau\sim \mC$ are independent. However, when controlling $\lVert \tilde{\bm g}_\ell^u(\bm \lambda)-\tilde{\bm g}_\ell^u(\bm \lambda_\epsilon)\rVert_2^2$ in \Cref{lem:stability term 3}, since reports $\bm u_\tau$ can be history-dependent, to make \Cref{lem:covering argument} still applicable we need to make sure the conditional distribution of reports $\bm u_\tau\mid \mH_{\tau}$ is $\mF_{\tau-1}$-measurable.) Then $X_\tau$ is $\mF_\tau$-measurable, and we have
\begin{equation*}
\E[X_\tau\mid \mF_{\tau-1}]=\Pr_{\bm v\sim \mV,\bm c\sim \mC}\left \{\exists \bm \lambda\in \mathcal B_\epsilon(\bm \lambda_\epsilon)\text{ s.t. } \tilde i^{\ast}(\bm \lambda)\ne \tilde i^{\ast}(\bm \lambda_\epsilon)\right \}\le K^2(\epsilon\cdot \epsilon_c),
\end{equation*}
where the last step uses \Cref{eq:close dual variables give similar argmax}.
Applying Azuma-Hoeffding inequality \citep[Corollary 2.20]{wainwright2019high} to the martingale difference sequence $\{X_\tau-\E[X_\tau\mid \mF_{\tau-1}]\}_{\ell'<\ell,\tau\in \mE_{\ell'}}$, we know for any $c>0$,
\begin{equation*}
\Pr\left \{\frac{\sum_{\ell'<\ell,\tau\in \mE_{\ell'}} \1[\exists \bm \lambda\in \mathcal B_\epsilon(\bm \lambda_\epsilon)\text{ s.t. }\tilde i_\tau^{v}(\bm \lambda)\ne \tilde i_\tau^{v}(\bm \lambda_\epsilon)]}{\sum_{\ell'=1}^{\ell-1} \lvert \mE_{\ell'}\rvert}\ge K^2 \epsilon \epsilon_c + c\right \}\le \exp \left (-2c^2\sum_{\ell'=1}^{\ell-1} \lvert \mE_{\ell'}\rvert\right ).
\end{equation*}

Applying a Union Bound over all $\bm \lambda_\epsilon\in \bm \Lambda_\epsilon$ and recalling the expression of $\lVert \tilde{\bm g}_\ell^v(\bm \lambda)-\tilde{\bm g}_\ell^v(\bm \lambda_\epsilon)\rVert_2^2$,
\begin{equation*}\begin{aligned}
\scalemath{0.95}{\Pr\left \{\exists \bm \lambda_\epsilon\in \bm \Lambda_\epsilon,\bm \lambda\in \bm \Lambda\text{ s.t. }\lVert \tilde{\bm g}_\ell^v(\bm \lambda)-\tilde{\bm g}_\ell^v(\bm \lambda_\epsilon)\rVert_2^2\ge d\lvert \mE_\ell\rvert^2(K^2 \epsilon \epsilon_c + c)^2\right \}\le \prod_{i=1}^d \frac{d}{\rho_j \epsilon} \cdot \exp \left (-2c^2\sum_{\ell'=1}^{\ell-1} \lvert \mE_{\ell'}\rvert\right ).}
\end{aligned}\end{equation*}

Therefore, with probability at least $1-\delta$, we have
\begin{equation}\begin{aligned}
&\lVert \tilde{\bm g}_\ell^{v}(\bm \lambda)-\tilde{\bm g}_\ell^{v}(\bm \lambda_\epsilon)\rVert_2^2\le d \lvert \mE_\ell\rvert^2\left (K^2 \epsilon \epsilon_c + \frac{\sqrt{\log \frac{1}{\delta}+\sum_{j=1}^d \log\frac{d}{\rho_j \epsilon}}}{\sqrt{\sum_{\ell'=1}^{\ell-1} \lvert \mE_{\ell'}\rvert}}\right )^2, \quad \forall \bm \lambda_\epsilon\in \bm \Lambda_\epsilon,\bm \lambda\in \mathcal B_\epsilon(\bm \lambda_\epsilon). \label{eq:diff of Fv on lambda_eps and lambda}
\end{aligned}\end{equation}

\paragraph{Final Bound.}
Via Union Bound, w.p. $1-2\delta$, \Cref{eq:diff of F on lambda_eps and lambda,eq:diff of F and Fv on lambda_eps,eq:diff of Fv on lambda_eps and lambda} are true and
\begin{equation*}\begin{aligned}
\lVert \tilde{\bm g}_\ell^{\ast}(\bm \lambda)-\tilde{\bm g}_\ell^{v}(\bm \lambda)\rVert_2^2 &\le 3d \lvert \mE_\ell\rvert^2 \cdot (K^2 \epsilon \epsilon_c)^2 + 6 d \lvert \mE_\ell\rvert^2 \cdot \frac{\log \frac{1}{\delta}+\sum_{j=1}^d \log\frac{d}{\rho_j \epsilon}}{\sum_{\ell'=1}^{\ell-1} \lvert \mE_{\ell'}\rvert},\quad \forall \bm \lambda\in \bm \Lambda.
\end{aligned}\end{equation*}

This finishes the proof.
\end{proof}

\begin{lemma}[Untruthful Reports]\label{lem:stability term 3}
For any $\ell\in [L]$, $\epsilon>0$, and $\delta\in (0,1)$, with probability $1-\frac{2}{3dT}$,
\begin{equation*}\begin{aligned}
\lVert \tilde{\bm g}_\ell^{u}(\bm \lambda)-\tilde{\bm g}_\ell^{v}(\bm \lambda)\rVert_2^2\le \frac{d \lvert \mE_\ell\rvert^2}{(\sum_{\ell'<\ell} \lvert \mE_{\ell'} \rvert)^2} M_\ell^2 +4d\lvert \mE_\ell\rvert^2 \left ((K^2 \epsilon \epsilon_c)^2 + \frac{\log (dT)+\sum_{j=1}^d \log\frac{d}{\rho_j \epsilon}}{\sum_{\ell'=1}^{\ell-1} \lvert \mE_{\ell'}\rvert}\right ),\quad \forall \bm \lambda\in \bm \Lambda,
\end{aligned}\end{equation*}
where $\tilde{\bm g}_\ell^u$ and $\tilde{\bm g}_\ell^v$ are defined in \Cref{eq:formal loss distributions}, and $M_\ell$ is defined in \Cref{eq:N and M}.
\end{lemma}
\begin{proof}
Similar to the reason in \Cref{lem:stability term 2}, we cannot directly apply concentration inequalities to $\bm \lambda_\ell$ which is unmeasurable when the reports are generated. Therefore, we still consider an $\epsilon$-net $\bm \Lambda_\epsilon$ of $\bm \Lambda$, ensure that $\tilde{\bm g}_\ell^v(\bm \lambda_\epsilon)\approx \tilde{\bm g}_\ell^u(\bm \lambda_\epsilon)$ for all $\bm \lambda_\epsilon\in \bm \Lambda_\epsilon$, and then extend it to all $\bm \lambda\in \bm \Lambda$ via Chernoff-Hoeffding inequalities.

\paragraph{Step 1: Cover $\bm \Lambda$ with an $\epsilon$-net.}
From \Cref{lem:covering argument}, for any $\epsilon > 0$, there exists an $\epsilon$-net $\bm \Lambda_\epsilon \subseteq \bm \Lambda$ of size $\O((d/\epsilon)^d)$, such that every $\bm \lambda \in \bm \Lambda$ has some $\bm \lambda_\epsilon \in \bm \Lambda_\epsilon$ with $\lVert \bm \lambda - \bm \lambda_\epsilon \rVert_1 \le \epsilon$. We remark that our final guarantee does not have a $d$-exponent, because the dependency on $\lvert \bm \Lambda_\epsilon\rvert$ is logarithmic.

\paragraph{Step 2: Concentration for any fixed $\bm \lambda_\epsilon \in \bm \Lambda_\epsilon$.}
Fix $\bm \lambda_\epsilon\in \bm \Lambda_\epsilon$.
From \Cref{lem:large misreport}, which underpins the \textsc{InterEpoch} analysis in \Cref{thm:InterEpoch guarantee}, we know that for any previous epoch $\ell' < \ell$, the event $|u_{\tau,i} - v_{\tau,i}| \ge \frac{1}{|\mathcal{E}_{\ell'}|}$ occurs in only $\Otil(1)$ rounds (with $\tau \in \mathcal{E}_{\ell'}$) with probability at least $1 - \frac{1}{|\mathcal{E}_{\ell'}|}$. Following the approach in \Cref{lem:misallocation}, we now leverage the smoothness of consumptions in \Cref{assump:smooth consumptions} in combination with Azuma-Hoeffding inequalities to conclude that $\tilde{\bm g}_\ell^v(\bm \lambda_\epsilon) \approx \tilde{\bm g}_\ell^u(\bm \lambda_\epsilon)$.

Formally, to compare $\tilde{\bm g}_\ell^u(\bm \lambda_\epsilon)$ and $\tilde{\bm g}_\ell^v(\bm \lambda_\epsilon)$, we need only to control the number of previous rounds $\tau$ such that $\tilde i_\tau^{u}(\bm \lambda_\epsilon) \ne \tilde i_\tau^{v}(\bm \lambda_\epsilon)$. We decompose such events by whether large misreports happen:
\begin{equation}\begin{aligned}
&\sum_{\ell'<\ell,\,\tau \in \mE_{\ell'}} \1\left[\tilde i_\tau^{u}(\bm \lambda_\epsilon) \ne \tilde i_\tau^{v}(\bm \lambda_\epsilon)\right]
\le \sum_{\ell'<\ell,\,\tau \in \mE_{\ell'}} \1\left[\exists i \in [K] \text{ s.t. } \lvert u_{\tau,i} - v_{\tau,i} \rvert \ge \frac{1}{\lvert \mE_{\ell'} \rvert} \right]  \\
&\quad \qquad  + \sum_{\ell'<\ell,\,\tau \in \mE_{\ell'}} \1\left[\exists i \ne j \in [K] \text{ s.t. } (v_{\tau,i} - v_{\tau,j}) - \bm \lambda_\epsilon^\trans(\bm c_{\tau,i} - \bm c_{\tau,j}) \in \left[0, \frac{2}{\lvert \mE_{\ell'} \rvert}\right]\right], \label{eq:misalloc decomposition eps net}
\end{aligned}\end{equation}
where the second term plugs $\lvert u_{\tau,i} - v_{\tau,i} \rvert \ge \frac{1}{\lvert \mE_{\ell'} \rvert},\forall i\in [K]$ into definitions of $\tilde i_\tau^{u}(\bm \lambda_\epsilon)$ and $\tilde i_\tau^{v}(\bm \lambda_\epsilon)$.

For the first term, we use the following inequality which appeared as \Cref{eq:large misreports takeaway} in \Cref{lem:large misreport}:
\begin{equation*}
\Pr\left \{\sum_{\tau\in \mE_{\ell'}}\1\left [\lvert u_{\tau,i}-v_{\tau,i}\rvert\ge \frac{1}{\lvert \mE_{\ell'}\rvert}\right ]\ge c\right \}\le 2(1+2\lVert \bm \rho^{-1}\rVert_1)\cdot \frac{2K\lvert \mE_{\ell'}\rvert^3}{\gamma^{-c}-1},\quad \forall \ell'<\ell, c>0.
\end{equation*}

For any fixed failure probability $\delta\in (0,1)$, for every $\ell'<\ell$, picking $c$ so that the RHS is $\frac{\delta}{\ell}$ gives
\begin{equation*}
\Pr\left \{\sum_{\ell'<\ell,\tau\in \mE_{\ell'}}\1\left [\lvert u_{\tau,i}-v_{\tau,i}\rvert\ge \frac{1}{\lvert \mE_{\ell'}\rvert}\right ]\ge \ell \log_{\gamma^{-1}} \left (1+4(1+\lVert \bm \rho^{-1}\rVert_1)K\lvert \mE_\ell\rvert^3\cdot \ell \delta^{-1}\right )\right \}\le \delta.
\end{equation*}

For the second term, under \Cref{assump:smooth consumptions} that $\Law(\bm \lambda_\epsilon^\trans \bm c_{\tau,i})$ is uniformly bounded by $\epsilon_c$, $\forall i\in [K]$,
\begin{equation*}\begin{aligned}
\Pr\left \{(v_{\tau,i}-v_{\tau,j})-\bm \lambda_\epsilon^\trans(\bm c_{\tau,i}-\bm c_{\tau,j})\in \left [0,\frac{2}{\lvert \mE_{\ell'}\rvert}\right ] \right \}\le \frac{2}{\lvert \mE_{\ell'}\rvert}\epsilon_c \lVert \bm \lambda_\epsilon\rVert_1,~~ \forall i\ne j\in \{0\}\cup [K],\ell'<\ell,\tau\in \mE_{\ell'}.
\end{aligned}\end{equation*}

Although we are now standing at epoch $\ell$, the $\epsilon$-nets are fixed before the game (as it only depends on $\bm \Lambda$). Hence, the indicator $X_\tau:=\1[(v_{\tau,i}-v_{\tau,j})-\bm \lambda_\epsilon^\trans(\bm c_{\tau,i}-\bm c_{\tau,j})\in [0,\frac{2}{\lvert \mE_{\ell'}\rvert} ]]$ is indeed $\mF_{\tau}$-measurable back in the past when $\tau\in \mE_{\ell'}$ and $\ell'<\ell$, where $(\mF_\tau)_{\tau\ge 0}$ is the natural filtration $\mF_\tau=\sigma(X_1,\ldots,X_\tau)$. Applying multiplicative Azuma-Hoeffding inequality \citep[Lemma 10]{koufogiannakis2014nearly} to the martingale difference sequence $\{X_\tau-\E[X_\tau\mid \mF_{\tau-1}]\}_{\ell'<\ell,\tau\in \mE_{\ell'}}$ gives
\begin{equation*}\begin{aligned}
\Pr\left \{\frac 12\sum_{\ell'<\ell,\tau\in \mE_{\ell'}} X_\tau\ge \sum_{\ell'<\ell,\tau\in \mE_{\ell'}} \E[X_\tau\mid \mF_{\tau-1}]+2A\right \}\le \exp(-A),\quad \forall A\in \mathbb R.
\end{aligned}\end{equation*}

Since $X_\tau$ only involves $\bm v_\tau\sim \mV$ and $\bm c_\tau\sim \mC$ which are i.i.d., we know $\E[X_\tau\mid \mF_{\tau-1}]=\E[X_\tau]\le \frac{2}{\lvert \mE_{\ell'}\rvert}\epsilon_c \lVert \bm \lambda_\epsilon\rVert_1$.
Setting the RHS as $\frac{\delta}{\lvert \bm \Lambda_\epsilon\rvert K^2}$ and taking a Union Bound over all $\bm \lambda_\epsilon \in \bm \Lambda_\epsilon$, we have
\begin{equation*}\begin{aligned}
&\Pr\Bigg \{\max_{\bm \lambda_\epsilon\in \bm \Lambda_\epsilon}\sum_{\ell'<\ell,\tau\in \mE_{\ell'}} \1\left [(v_{\tau,i}-v_{\tau,j})-\bm \lambda_\epsilon^\trans(\bm c_{\tau,i}-\bm c_{\tau,j})\in \left [0,\frac{2}{\lvert \mE_{\ell'}\rvert} \right ]\right ]\\
& \qquad \ge 4 \ell \epsilon_c \lVert \bm \rho^{-1}\rVert_1+4\left (\log \frac 1\delta + \sum_{j=1}^d \log \frac{d}{\rho_j \epsilon}\right )\Bigg \}\le \frac{\delta}{K^2},\quad \forall i\ne j\in \{0\} \cup [K].
\end{aligned}\end{equation*}

Using another Union Bound over all $i\ne j\in \{0\}\cup [K]$ and plugging it back into \Cref{eq:misalloc decomposition eps net},
\begin{equation*}\begin{aligned}
&\quad \max_{\bm \lambda_\epsilon\in \bm \Lambda_\epsilon} \sum_{\ell'<\ell,\tau\in \mE_{\ell'}} \1[\tilde i_\tau^{u}(\bm \lambda_\epsilon)\ne \tilde i_\tau^{v}(\bm \lambda_\epsilon)]\\
&\le \ell \log_{\gamma^{-1}} \left (1+4(1+\lVert \bm \rho^{-1}\rVert_1)K\lvert \mE_\ell\rvert^3\cdot 4 \ell \delta^{-1}\right )+4 \ell \epsilon_c \lVert \bm \rho^{-1}\rVert_1+4\left (\log \frac 1\delta + \sum_{j=1}^d \log \frac{d}{\rho_j \epsilon}\right ),
\end{aligned}\end{equation*}
with probability at least $1-2\delta$.
Therefore, with probability at least $1-2\delta$, we have
\begin{equation*}\begin{aligned}
&\quad \max_{\bm \lambda_\epsilon\in \bm \Lambda_\epsilon}\lVert \tilde{\bm g}_\ell^{u}(\bm \lambda_\epsilon) - \tilde{\bm g}_\ell^{v}(\bm \lambda_\epsilon) \rVert_2^2
\le d \left( \frac{\lvert \mE_\ell \rvert}{\sum_{\ell'<\ell} \lvert \mE_{\ell'} \rvert}
\sum_{\ell'<\ell, \tau \in \mE_{\ell'}} \1[\tilde i_\tau^{u}(\bm \lambda_\epsilon) \ne \tilde i_\tau^{v}(\bm \lambda_\epsilon)] \right)^2\\
&\le \frac{d\lvert \mE_{\ell}\rvert^2}{(\sum_{\ell'<\ell} \lvert \mE_{\ell'} \rvert)^2} \left (\ell \log_{\gamma^{-1}} \Bigg (1+4(1+\lVert \bm \rho^{-1}\rVert_1)K\lvert \mE_\ell\rvert^3\cdot 4 \ell \delta^{-1}\right )+4 \ell \epsilon_c \lVert \bm \rho^{-1}\rVert_1+4\left (\log \frac 1\delta + \sum_{j=1}^d \log \frac{d}{\rho_j \epsilon}\right )\Bigg )^2,
\end{aligned}\end{equation*}
where the inequality uses $\lVert \bm \rho-\bm c_{\tau,i}\rVert_2^2\le d$. Using the definition of $M_\ell$ in \Cref{eq:N and M}, when $\delta=\frac{1}{6dT}$ this is exactly $\frac{d \lvert \mE_\ell\rvert^2}{(\sum_{\ell'<\ell} \lvert \mE_{\ell'} \rvert)^2} M_\ell^2$.

\paragraph{Step 3: Extend the similarity to all $\bm \lambda \in \bm \Lambda$.}
After yielding the similarity that $\tilde{\bm g}_\ell^{u}(\bm \lambda_\epsilon) \approx \tilde{\bm g}_\ell^{v}(\bm \lambda_\epsilon)$ for all $\bm \lambda_\epsilon\in \bm \Lambda_\epsilon$, we extend it to all $\bm \lambda\in \bm \Lambda$ using the arguments already derived in \Cref{lem:stability term 2}. Recall from \Cref{eq:diff of Fv on lambda_eps and lambda} that we already proved that with probability $1-\delta$,
\begin{equation*}\begin{aligned}
\lVert \tilde{\bm g}_\ell^{v}(\bm \lambda)-\tilde{\bm g}_\ell^{v}(\bm \lambda_\epsilon)\rVert_2^2\le d \lvert \mE_\ell\rvert^2\left (K^2 \epsilon \epsilon_c + \frac{\sqrt{\log \frac{1}{\delta}+\sum_{j=1}^d \log\frac{d}{\rho_j \epsilon}}}{\sqrt{\sum_{\ell'=1}^{\ell-1} \lvert \mE_{\ell'}\rvert}}\right )^2, \quad \forall \bm \lambda_\epsilon\in \bm \Lambda_\epsilon,\bm \lambda\in \mathcal B_\epsilon(\bm \lambda_\epsilon).
\end{aligned}\end{equation*}

Using exactly the same arguments (see \Cref{footnote:reports measurability in stability} for the reason why \Cref{lem:covering argument} is still applicable when reports $\{\bm u_\tau\}_{\ell'<\ell,\tau\in \mE_{\ell'}}$ can be history-dependent), with probability $1-\delta$,
\begin{equation*}\begin{aligned}
\lVert \tilde{\bm g}_\ell^{u}(\bm \lambda)-\tilde{\bm g}_\ell^{u}(\bm \lambda_\epsilon)\rVert_2^2\le d \lvert \mE_\ell\rvert^2\left (K^2 \epsilon \epsilon_c + \frac{\sqrt{\log \frac{1}{\delta}+\sum_{j=1}^d \log\frac{d}{\rho_j \epsilon}}}{\sqrt{\sum_{\ell'=1}^{\ell-1} \lvert \mE_{\ell'}\rvert}}\right )^2, \quad \forall \bm \lambda_\epsilon\in \bm \Lambda_\epsilon,\bm \lambda\in \mathcal B_\epsilon(\bm \lambda_\epsilon).
\end{aligned}\end{equation*}

\paragraph{Final Bound.}
Putting the three inequalities together and taking Union Bound, when picking $\delta=\frac{1}{6dT}$,
\begin{equation*}\begin{aligned}
\lVert \tilde{\bm g}_\ell^{u}(\bm \lambda)-\tilde{\bm g}_\ell^{v}(\bm \lambda)\rVert_2^2\le \frac{d \lvert \mE_\ell\rvert^2}{(\sum_{\ell'<\ell} \lvert \mE_{\ell'} \rvert)^2} M_\ell^2 +4d\lvert \mE_\ell\rvert^2 \left ((K^2 \epsilon \epsilon_c)^2 + \frac{\log (dT)+\sum_{j=1}^d \log\frac{d}{\rho_j \epsilon}}{\sum_{\ell'=1}^{\ell-1} \lvert \mE_{\ell'}\rvert}\right ),\quad \forall \bm \lambda\in \bm \Lambda
\end{aligned}\end{equation*}
holds with probability $1-4\delta=1-\frac{2}{3dT}$, where $M_\ell$ is defined in \Cref{eq:N and M}.
\end{proof}
\section{Extensions: Multi-Unit Allocation}\label{sec:appendix extensions}
This section proves \Cref{thm:multi-unit,thm:multi-demand}, whose difference is the dependency on $N$: In single-unit demand cases, the allocation and payment rule in \Cref{eq:payment multi-unit multi-demand,eq:allocation multi-unit multi-demand} only needs to consider $K$ dual-adjusted reports, namely $\{u_{t,i}-\bm \lambda_t^\trans \bm c_{t,i}\}_{i\in [K]}$; in multi-unit demand cases, $K\times N$ dual-adjusted reports are compared. This is not merely about computational complexity, but rather the gradient stability (\Cref{def:gradient stability}): As sketched in the proof of \Cref{thm:multi-demand}, we need to ensure that under small local perturbations of $\bm \lambda$, the ranking between all feasible allocations remain unchanged w.h.p. Due to the larger set of feasible allocations---characterized as the top-$N$ between $KN$ reports as discussed below \Cref{eq:allocation multi-unit multi-demand}---this introduces an additional $N^2$ dependency.

To avoid repeated proofs, we denote by $\mathcal S$ the set of dual-adjusted reports the planner considers in each round. That is, in unit demand setups $\mathcal S=[K]$, and in multi-unit demand setups $\mathcal S=[K]\times [N]$. Let $S:=\lvert \mathcal S\rvert$. In addition to allocation and payment, $\mathcal S$ also appears in randomized exploration as the candidate set. Inheriting the regret decomposition in \Cref{eq:regret decomposition} from the single-unit analysis, we decompose regret $\mR_T$ as
\begin{equation*}
\mR_T=\E\left [\sum_{t=1}^T (v_{t,I_t^\ast}-v_{t,I_t})\right ]\le \E\Bigg [\underbrace{\sum_{t=1}^{\mT_v} (v_{t,\tilde I_t^\ast}-v_{t,I_t})}_{\primalreg}+\underbrace{\sum_{t=1}^T v_{t,I_t^\ast}-\sum_{t=1}^{\mT_v} v_{t,\tilde I_t^\ast}}_{\dualreg}\Bigg ],
\end{equation*}
where $\{I_t^\ast\}_{t\in [T]}$ is the offline optimum benchmark in \Cref{sec:multi-unit}, $\{\tilde I_t^\ast\}_{t\in [T]}$ is the dual-adjusted optimum $\tilde I_t^\ast:=\argmax_{I\in \mathcal I} (v_{t,I}-\bm \lambda_t^\trans \bm c_{t,I})$, and $\mT_v$ is the stopping time $\mT_v:=\min\{t\in [T]\mid \sum_{\tau=1}^t \bm c_{\tau,I_\tau}+N\bm 1\not \le T\bm \rho\}\cup \{T+1\}$.

\subsection{Multi-Unit \algname Guarantee}
\begin{theorem}[Multi-Unit \algname Guarantee]\label{thm:multi-unit IAPD}
With $S$ dual-adjusted reports in each round, the multi-unit \algname framework induces a PBE of agents' strategies $\bm \pi^\ast$, such that \primalreg is bounded by
\begin{equation*}
\E[\primalreg]\le \sum_{\ell=1}^L \left (4+4S^2 \epsilon_c+5\log \lvert \mE_\ell\rvert+\log_{\gamma^{-1}} (1+4N(1+\lVert \bm \rho^{-1}\rVert_1)S\lvert \mE_\ell\rvert^4)\right ).
\end{equation*}
Similar to $N_\ell$ in \Cref{eq:N and M}, we define the multi-unit analog $\tilde N_\ell$ as:
\begin{equation}\label{eq:N multi-unit}
\tilde N_\ell:=1+4S^2 \epsilon_c+5\log \lvert \mE_\ell\rvert+\log_{\gamma^{-1}}(1+4N (1+\lVert \bm \rho^{-1}\rVert_1) S \lvert \mE_\ell\rvert^4).
\end{equation}
Then we still write $\E[\primalreg]\le \sum_{\ell=1}^L (\tilde N_\ell+3)$ for the multi-unit \algname framework.
\end{theorem}
\begin{proof}
Following the proof of \Cref{thm:InterEpoch guarantee} in \Cref{sec:appendix PrimalAlloc} closely, we first establish the analog of \Cref{lem:2nd price auction variant}: For each round $t\in [T]$, should agents be myopic, they would find truthful reporting dominant under \Cref{eq:allocation multi-unit multi-demand,eq:payment multi-unit multi-demand}; this is immediate from replacing the second-price auction in the proof of \Cref{lem:2nd price auction variant} with the VCG pricing. We thus directly have \Cref{thm:IntraEpoch guarantee}: for any epoch, ``nearsighted'' agents optimizing the current-epoch expected utility find truthful reporting a PBE.

We utilize the exploration component to prove an analog of \Cref{lem:large misreport}.
Specifically, by comparing agent $i$'s expected utility under any PBE $\bm \pi$ and unilateral deviation $\bm \pi^i:=\truth_i\circ \bm \pi_{-i}$ as in \Cref{eq:large misreports V difference}, we arrive at a conclusion extremely similar to \Cref{eq:large misreports takeaway}. There are two difference: \textit{(i)} when bounding the \text{Current-Epoch} term in \Cref{eq:large misreports V difference}, since there are $S$ candidates pending exploration, for an agent $i$ to misreport some $u_{t,i,n}\ne v_{t,i,n}$ their expected loss bound is $-\frac{1}{2\lvert \mE_\ell\rvert S} (u_{t,i,n}-v_{t,i,n})^2$ instead of $-\frac{1}{2\lvert \mE_\ell\rvert K} (u_{t,i,n}-v_{t,i,n})^2$ (cf. \Cref{eq:exploration misreport loss}); \textit{(ii)} the payment magnitude is now $\lvert p_{\tau,i_\tau}\rvert\le (1+2 N \lVert \bm \rho\rVert_1)$ instead of $(1+2 \lVert \bm \rho\rVert_1)$ (compare \Cref{line:primal allocation} with \Cref{eq:payment multi-unit multi-demand}).
Taking these two differences into consideration, \Cref{eq:large misreports takeaway} now becomes
\begin{equation*}
2(1+2 N \lVert \bm \rho^{-1}\rVert_1)\frac{\gamma^{s_{\ell+1}}}{1-\gamma}\ge \Pr\{\lvert \mathcal M_{\ell,i}\rvert\ge c\}\frac{\gamma^{s_{\ell+1}}}{1-\gamma} \frac{\gamma^{-c}-1}{2S\lvert \mE_\ell\rvert^3},\quad \forall c>0,
\end{equation*}
where $\mathcal M_{\ell,i}=\left \{t\in \mE_\ell\mid \lvert u_{t,i}-v_{t,i}\rvert\ge \frac{1}{\lvert \mE_\ell}\rvert\right \}$ and $s_\ell$ and $e_\ell$ are the first and last round in epoch $\mE_\ell$, respectively. Therefore, we still have \Cref{lem:large misreport}, except that the bound is now $\log_{\gamma^{-1}} (1+4N(1+\lVert \bm \rho^{-1}\rVert_1) S \lvert \mE_\ell\rvert^4)$.

Moving on to \Cref{lem:misallocation}, we replicate the key inequality \Cref{eq:misalloc prob when similar}---which bounds the probability that a small deviation in report results in a different allocation---under the multi-unit \algname framework. That is, we consider the probability (recall $\mathcal S=[K]$ in unit demand setups and $\mathcal S=[K]\times [N]$ otherwise):
\begin{equation}\begin{aligned}\label{eq:primal flip multi-unit}
\Pr\left \{\left (\argmax_{I\in \mathcal I} (u_{t,I}-\bm \lambda_\ell^\trans \bm c_{t,I})\ne \argmax_{I\in \mathcal I} (v_{t,I}-\bm \lambda_\ell^\trans \bm c_{t,I})\right ) \wedge \left (\max_{s\in \mathcal S} \lvert u_{t,s}-v_{t,s}\rvert\le \frac{1}{\lvert \mE_\ell\rvert}\right ) \right \}.
\end{aligned}\end{equation}
Since the $\argmax_{I\in \mathcal I}$ in \Cref{eq:allocation multi-unit multi-demand} corresponds to the top-$N$ operator among $\{u_{t,s}-\bm \lambda^\trans \bm c_{t,s}\}_{s\in \mathcal S}$, by considering the $N$-th and $(N+1)$-th candidate among $\lvert \mathcal S\rvert$, this probability is further bounded by
\begin{equation*}\begin{aligned}
\text{\Cref{eq:primal flip multi-unit}}&\le \sum_{s\ne s'\in \mathcal S} \Pr \left \{0\le (u_s-\bm \lambda_\ell^\trans \bm c_s)-(u_{s'}-\bm \lambda_\ell^\trans \bm c_{s'})\le \frac{2}{\lvert \mE_\ell\rvert}\right \}\le \frac{2 S^2 \epsilon_c}{\lvert \mE_\ell\rvert},
\end{aligned}\end{equation*}
where we define $s\ne s'$ \emph{only} using the agent index (i.e., we regard $s=(i,n)$ and $s'=(i',n')$ as the same if $i=i'$; this is because for any agent $i\in [K]$ and any $n<n'\in [N]$, we always have $u_{i,n}\ge u_{i,n'}$ and $\bm c_{i,n}=\bm c_{i,n'}$, which means no flip between $(i,n)$ and $(i,n')$ happens under perturbation). The second inequality applies \Cref{assump:smooth consumptions}.
Applying the same Azuma-Hoeffiding inequality as in \Cref{lem:misallocation}, w.p. $1-\frac{2}{\lvert \mE_\ell\rvert}$ the number of misallocations in those rounds without large misreports is bounded by $4S^2 \epsilon_c \lVert \bm \lambda_\ell\rVert_1+4\log \lvert \mE_\ell\rvert$.
Following the remaining proof of \Cref{thm:InterEpoch guarantee}, we derive the claimed guarantee for multi-unit \algname.
\end{proof}

\subsection{Multi-Unit \ftrl and \oftrlfp Guarantees}
\begin{lemma}[Multi-Unit \dualreg and Online Learning Regret]\label{lem:DualUpd to regret multi-unit}
Under multi-unit \algname,
\begin{equation*}
\E[\dualreg]\le \mR_L^{\text{OL}}+N\left (1+\max_{1\le \ell \le L} \lvert \mE_\ell\rvert+\sum_{\ell=1}^L (\tilde N_\ell+3)\right )\lVert \bm \rho^{-1}\rVert_1,
\end{equation*}
with the online learning regret $\mR_L^{\text{OL}}$ defined as (where $\mathcal L_v$ is the epoch that the stopping time $\mathcal T_v$ belongs to)
\begin{equation*}
\mR_L^\text{OL}:=\E\left [\sup_{\bm \lambda^\ast\in \bm \Lambda} \sum_{\ell=1}^{\mathcal L_v} \left (\sum_{t\in \mE_\ell} (\bm \rho-\bm c_{t,I_t})\right )^\trans (\bm \lambda_\ell-\bm \lambda^\ast)\right ],\quad \bm \Lambda:=\bigotimes_{j=1}^d [0,\rho_j^{-1}].
\end{equation*}
\end{lemma}
\begin{proof}
The proof of \Cref{lem:DualUpd to regret} applies verbatim by replacing $i_t$ with $I_t$, $\tilde i_t^\ast$ with $\tilde I_t^\ast$, $N_\ell$ with $\tilde N_\ell$ (the \primalreg guarantees in \Cref{thm:InterEpoch guarantee,thm:multi-unit IAPD}), and the bound $\lVert \bm c_{t,i_t}\rVert_\infty\le 1$ with $\lVert \bm c_{t,I_t}\rVert_\infty\le N$.
\end{proof}
\begin{lemma}[Multi-Unit \ftrl Guarantee]\label{thm:DualUpd guarantee FTRL multi-unit}
When deciding the dual $\bm \lambda_\ell$ as in \Cref{eq:lambda FTRL multi-unit},
\begin{equation*}
\mR_L^{\text{OL}}\le \sup_{\bm \lambda^\ast\in \bm \Lambda} \frac{\Psi(\bm \lambda^\ast)}{\eta_L} + dN^2\sum_{\ell=1}^L \eta_\ell \lvert \mE_\ell\rvert^2.
\end{equation*}
\end{lemma}
\begin{proof}
Since $\mR_L^{\text{OL}}$ is standard online learning regret with $\mathcal L_v\le L$ a.s., the standard \ftrl analysis \citep[see, e.g.,][Corollary 7.7]{orabona2019modern} over $\bm \Lambda=\bigotimes_{j=1}^d [0,\rho_j^{-1}]$ with loss $F_\ell(\bm \lambda):=\sum_{t\in \mE_\ell} (\bm \rho-\bm c_{t,I_t})^\trans \bm \lambda$ gives
\begin{equation*}
\mR_L^{\text{OL}}=\E\left [\sup_{\bm \lambda^\ast\in \bm \Lambda} \sum_{\ell=1}^{\mathcal L_v} \left (\sum_{t\in \mE_\ell} (\bm \rho-\bm c_{t,I_t})\right )^\trans (\bm \lambda_\ell-\bm \lambda^\ast)\right ]\le \sup_{\bm \lambda^\ast\in \bm \Lambda} \frac{\Psi(\bm \lambda^\ast)}{\eta_L}+\sum_{\ell=1}^L \eta_\ell \E\left [\left \lVert \sum_{t\in \mE_\ell} (\bm \rho-\bm c_{t,I_t})\right \rVert_2^2\right ],
\end{equation*}
As $\bm \rho,\bm c_{t,I}\in [0,N]^d$, we have $\lVert \sum_{t\in \mE_\ell} (\bm \rho-\bm c_{t,I_t})\rVert_2^2\le d N^2 \lvert \mE_\ell\rvert^2$ for all $\ell\in [L]$. This gives the conclusion.
\end{proof}
\begin{lemma}[Multi-Unit \oftrlfp Guarantee]\label{thm:DualUpd guarantee O-FTRL-FP multi-unit}
When deciding the dual $\bm \lambda_\ell$ as in \Cref{eq:lambda O-FTRL-FP multi-unit},
\begin{equation*}\begin{aligned}
\mR_L^{\text{OL}}&\le \sup_{\bm \lambda^\ast\in \bm \Lambda} \frac{\Psi(\bm \lambda^\ast)}{\eta_L}+d \eta_1 N^2 \lvert \mE_1\rvert^2+\sum_{\ell=2}^L \eta_\ell \frac{N^2 \lvert \mE_\ell\rvert^2}{\sum_{\ell'<\ell} \lvert \mE_{\ell'}\rvert} \left (4d^2 + 24d \log(dTL) + 24d \sum_{j=1}^d \log \frac{dK^2 \epsilon_c T}{\rho_j}\right ) \\
&\quad +\sum_{\ell=2}^L \eta_\ell N \lvert \mE_\ell\rvert (16d+10)+\sum_{\ell=1}^L \eta_\ell N^2 \left (2N_\ell^2 + \frac{2d \lvert \mE_\ell\rvert^2 M_\ell^2}{(\sum_{\ell'=1}^{\ell-1} \lvert \mE_{\ell'}\rvert)^2}\right )+2\lVert \bm \rho^{-1}\rVert_1 N^2,
\end{aligned}\end{equation*}
where $\tilde M_\ell$ is the multi-unit analog of the $M_\ell$ (which is defined in \Cref{eq:N and M}):
\begin{equation}\label{eq:M multi-unit}
\tilde M_\ell:=\ell \log_{\gamma^{-1}} \left (1+4(1+\lVert \bm \rho^{-1}\rVert_1)S\lvert \mE_\ell\rvert^3\cdot 24 \ell dT\right )+4 \ell \epsilon_c \lVert \bm \rho^{-1}\rVert_1+4\left (\log(dT) + \sum_{j=1}^d \log \frac{d S^2 \epsilon_c T}{\rho_j \epsilon}\right ),
\end{equation}
and $\tilde N_\ell$ is the multi-unit analog of $N_\ell$ defined in \Cref{eq:N multi-unit}.
\end{lemma}
\begin{proof}
Following the proof of \Cref{thm:DualUpd guarantee O-FTRL-FP} step by step, we first establish the gradient-stability of the multi-unit \oftrlfp extremely similar to \Cref{lem:approximate continuity of predictions,lem:covering argument}: Take an $\frac{\epsilon}{2}$-net of $\bm \Lambda$ denoted by $\bm \Lambda_\epsilon$, and consider the stochastic process $(X_\tau)_{\tau\ge 1}$ adapted to $(\mathcal F_\tau)_{\tau\ge 0}$ with $\mathcal F_\tau=\sigma(\bigcup_{i\in \{0\}\cup [K]} \mH_{\tau+1,i})$:
\begin{equation*}
X_\tau=\1[\exists \bm \lambda\in \mathcal B_\epsilon(\bm \lambda_\epsilon)\text{ s.t. }\tilde I_\tau(\bm \lambda)\ne \tilde I_\tau(\bm \lambda_\epsilon)],\quad \forall \ell'<\ell,\tau\in \mE_{\ell'},
\end{equation*}
where $\mathcal B_\epsilon(\bm \lambda_\epsilon):=\{\bm \lambda\in \bm \Lambda\mid \lVert \bm \lambda-\bm \lambda_\epsilon\rVert_1\le \epsilon\}$, and we direct the readers to \Cref{footnote:uniform smoothness} for discussions on $(X_\tau)_{\tau\ge 1}$. To calculate $\E[X_t\mid \mathcal F_{t-1}]$, imitating \Cref{lem:covering argument}, we consider the probability that two allocations, namely $I\ne J\subseteq \mathcal I$, flips their order under a small perturbation of $\bm \lambda$. That is, we consider the probability
\begin{equation*}
\Pr_{\bm u\sim \mV,\bm c\sim \mC} \left \{\exists \bm \lambda\in \mathcal B_\epsilon(\bm \lambda_\epsilon)\text{ and }I\ne j\in \mathcal I\text{ s.t. }(u_I-\bm \lambda_\epsilon^\trans \bm c_I>u_J-\bm \lambda_\epsilon^\trans \bm c_J)\wedge (u_I-\bm \lambda^\trans \bm c_I<u_J-\bm \lambda^\trans \bm c_J)\right \}.
\end{equation*}
Using the same arguments when controlling \Cref{eq:primal flip multi-unit} in the proof of \Cref{thm:multi-unit IAPD}, this is no more than
\begin{equation}\begin{aligned}\label{eq:multi-unit flip prob}
&\quad \sum_{s\ne s'\in \mathcal S} \Pr_{\bm u\sim \mV,\bm c\sim \mC} \left \{\exists \bm \lambda\in \mathcal B_\epsilon(\bm \lambda_\epsilon)\text{ s.t. }(u_s-\bm \lambda_\epsilon^\trans \bm c_s>u_{s'}-\bm \lambda_\epsilon^\trans \bm c_{s'})\wedge (u_s-\bm \lambda^\trans \bm c_s<u_{s'}-\bm \lambda^\trans \bm c_{s'})\right \}\\
&\le \sum_{s\ne s'\in \mathcal S} \Pr_{\bm u\sim \mV,\bm c\sim \mC} \left \{0\le (u_s-\bm \lambda_\epsilon^\trans \bm c_s)-(u_{s'}-\bm \lambda_\epsilon^\trans \bm c_{s'})\le 2\epsilon\right \}\le S^2 \epsilon \epsilon_c,
\end{aligned}\end{equation}
where the first inequality uses $\lvert (\bm \lambda_\epsilon-\bm \lambda)^\trans (\bm c_s-\bm c_{s'})\rvert\le \lVert \bm \lambda-\bm \lambda_\epsilon\rVert_1\cdot \lVert \bm c_s-\bm c_{s'}\rVert_\infty\le \epsilon$, and the second inequality applies the same arguments as in \Cref{lem:covering argument} (where the key step invokes \Cref{assump:smooth consumptions}). We also remark that we compare $s=(i,n)$ and $s'=(i',n')$ only using the agent's index $i$; see \Cref{eq:primal flip multi-unit} for more discussions.

After deriving this multi-unit analog of \Cref{lem:covering argument}, we further apply multiplicative Azuma-Hoeffding and Union Bound as in \Cref{lem:approximate continuity of predictions}. This derives a gradient stability guarantee similar to \Cref{eq:gradient stability in DualIneff}, but with the $K^2$ in $\epsilon_\ell$ replaced by $S^2$ (due to the larger Union Bound) and also an extra $N$ in $L_\ell$ (capturing the larger magnitude of $\bm \rho$ and $\bm c$): for any fixed $\epsilon_\ell\in (0,1)$ and $\delta_\ell>0$, with probability $1-\delta_\ell$, the prediction $\tilde{\bm g}_\ell(\bm \lambda_\ell)^\trans \bm \lambda$ enjoys $(\epsilon_\ell,L_\ell)$-gradient stability (\Cref{def:gradient stability}, i.e., $\lVert \tilde{\bm g}_\ell(\bm \lambda_1)-\tilde{\bm g}_\ell(\bm \lambda_2)\rVert_2\le L_\ell$ for all $\lVert \bm \lambda_1-\bm \lambda_2\rVert_2\le \epsilon_\ell$) with
\begin{equation*}
\epsilon_\ell=\frac{\sqrt d}{S^2 \epsilon_c \sum_{\ell'<\ell} \lvert \mE_{\ell'}\rvert},\quad L_\ell=N \frac{4 \lvert \mE_{\ell}\rvert \sqrt d}{\sum_{\ell'<\ell} \lvert \mE_{\ell'}\rvert}\left (d+\log \frac{1}{\delta_\ell}+\sum_{j=1}^d \log \frac{\sqrt d}{\rho_j \epsilon_\ell}\right ),\quad \forall \ell\ge 2.
\end{equation*}

Therefore, applying the \oftrlfp guarantee in \Cref{lem:O-FTRL lemma} controls the online learning regret as
\begin{equation*}\begin{aligned}
\mR_L^{\text{OL}}\le \sup_{\bm \lambda^\ast\in \bm \Lambda}\frac{\Psi(\bm \lambda^\ast)}{\eta_L}+\sum_{\ell=1}^L \eta_\ell \E\left [\left \lVert \sum_{t\in \mE_\ell} (\bm \rho-\bm c_{t,i_t})-\tilde{\bm g}_\ell(\bm \lambda_\ell)\right \rVert_2^2\right ]+\sum_{\ell=1}^L \sqrt{d} \lvert \mE_\ell\rvert \eta_\ell L_\ell+\left (\sum_{\ell=1}^L \delta_\ell\right ) 2T \lVert \bm \rho^{-1}\rVert_1 N^2,
\end{aligned}\end{equation*}
which replicates \Cref{eq:online learning regret decomposuition in O-FTRL-FP}. We therefore face the three $\texttt{Stability}_\ell$ terms defined in \Cref{eq:stability decomposition}:
\begin{equation*}
2\underbrace{\E\left [\left \lVert \sum_{t\in \mE_\ell} (\bm \rho-\bm c_{t,i_t})-\tilde{\bm g}_\ell^{\ast}(\bm \lambda_\ell)\right \rVert_2^2\right ]}_{\text{Primal Allocations}}+2\underbrace{\E\left [\lVert \tilde{\bm g}_\ell^{\ast}(\bm \lambda_\ell)-\tilde{\bm g}_\ell^{v}(\bm \lambda_\ell)\rVert_2^2\right ]}_{\text{Empirical Estimation}}+2\underbrace{\E\left [\lVert \tilde{\bm g}_\ell^{v}(\bm \lambda_\ell)-\tilde{\bm g}_\ell^{u}(\bm \lambda_\ell)\rVert_2^2\right ]}_{\text{Untruthful Reports}},\quad \forall \ell\in [L],
\end{equation*}
where $\tilde{\bm g}_\ell^u$ is the prediction being used by \oftrlfp, $\tilde{\bm g}_\ell^v$ is defined using true values, and $\tilde{\bm g}_\ell^\ast$ is defined using true distributions. See \Cref{eq:formal loss distributions} for the formal definitions, where all $i\in [K]$ are replaced by $I\in \mathcal I$.

It only remains to extend \Cref{lem:stability term 1,lem:stability term 2,lem:stability term 3} to the multi-unit multi-demand setting. For the first \text{Primal Allocations} term, which involves the concentration of i.i.d. $(\bm \rho-\bm c_{t,\tilde i_t^\ast})$ w.r.t. its mean $\tilde{\bm g}_\ell^\ast(\bm \lambda_\ell)$ and the \algname guarantee of approximating $\{i_t\}_{t\in \mE_\ell}$ using $\{\tilde i_t^\ast\}_{t\in \mE_\ell}$, the conclusion of \Cref{lem:stability term 1} is only amplified by $N^2$ (due to the larger $\bm \rho$ and $\bm c$) and has all $N_\ell$'s replaced by $\tilde N_\ell$'s (defined in \Cref{eq:N multi-unit}):
\begin{equation*}
\E\left [\left \lVert \sum_{t\in \mE_\ell} (\bm \rho-\bm c_{t,i_t})-\tilde{\bm g}_\ell^{\ast}(\bm \lambda_\ell)\right \rVert_2^2\right ]\le (d+3) N^2 \lvert \mE_\ell\rvert+\tilde N_\ell^2 N^2.
\end{equation*}

For the \text{Empirical Estimation} term, following the proof of \Cref{lem:stability term 2}, we first apply vector Berstein to concentrate $\sum_{\tau} (\bm \rho-\bm c_{\tau,\tilde I_\tau^v(\bm \lambda_\epsilon)})$ around its mean for all $\bm \lambda_\epsilon\in \bm \Lambda_\epsilon$, which results in an extra $N^2$ in \Cref{eq:diff of F and Fv on lambda_eps} due to the larger $\bm \rho$ and $\bm c$. We then apply the multi-unit analog of \Cref{lem:covering argument} to extend this to all $\bm \lambda\in \bm \Lambda$, which results in \Cref{eq:diff of F on lambda_eps and lambda,eq:diff of Fv on lambda_eps and lambda} but with $K^2$ replaced by $S^2$ and also an extra $N^2$. Thus,
\begin{equation*}
\lVert \tilde{\bm g}_\ell^{\ast}(\bm \lambda)-\tilde{\bm g}_\ell^{v}(\bm \lambda)\rVert_2^2 \le 3d N^2 \lvert \mE_\ell\rvert^2 \cdot (S^2 \epsilon \epsilon_c)^2 + 6 d N^2 \lvert \mE_\ell\rvert^2 \cdot \frac{\log \frac{1}{\delta}+\sum_{j=1}^d \log\frac{d}{\rho_j \epsilon}}{\sum_{\ell'=1}^{\ell-1} \lvert \mE_{\ell'}\rvert},\quad \forall \bm \lambda\in \bm \Lambda
\end{equation*}
holds with probability $1-2\delta$ for all $\ell\in [L]$. Similarly, for the \text{Untruthful Reports} term, following the proof of \Cref{lem:stability term 3} but exploiting the multi-unit \algname arguments in \Cref{thm:multi-unit IAPD} and the multi-unit covering bounds of \Cref{eq:multi-unit flip prob}, all $K$'s are replaced by $S$ and we also bear an extra $N^2$. Henceforth, w.p. $1-\frac{2}{3 dT}$,
\begin{equation*}\begin{aligned}
\lVert \tilde{\bm g}_\ell^{u}(\bm \lambda)-\tilde{\bm g}_\ell^{v}(\bm \lambda)\rVert_2^2\le \frac{d N^2 \lvert \mE_\ell\rvert^2}{(\sum_{\ell'<\ell} \lvert \mE_{\ell'} \rvert)^2} \tilde M_\ell^2 +4d N^2 \lvert \mE_\ell\rvert^2 \left ((K^2 \epsilon \epsilon_c)^2 + \frac{\log (dT)+\sum_{j=1}^d \log\frac{d}{\rho_j \epsilon}}{\sum_{\ell'=1}^{\ell-1} \lvert \mE_{\ell'}\rvert}\right ),\quad \forall \bm \lambda\in \bm \Lambda,
\end{aligned}\end{equation*}
where $\tilde M_\ell$ is defined in \Cref{eq:M multi-unit} and serves as the multi-unit analog of $M_\ell$. It is also of order $\Otil(\ell)$.

Plugging these multi-unit \texttt{Stability} bounds back into the online learning regret bound \Cref{eq:online learning regret decomposuition in O-FTRL-FP}, and adopting the same $\delta_\ell$'s as in the proof of \Cref{thm:DualUpd guarantee O-FTRL-FP}, we get the claimed upper bound on $\mR_L^{\text{OL}}$.
\end{proof}

\subsection{Proofs of \Cref{thm:multi-demand,thm:multi-unit}}

\begin{proof}[Proof of \Cref{thm:multi-unit}]
Under the unit demand setting, since $v_{t,i,n}=u_{t,i,n}=0$ for all $n\ge 2$, we only needs to compare $N$ dual-adjusted reports in \Cref{eq:allocation multi-unit multi-demand,eq:payment multi-unit multi-demand}. Hence we have $S=K$. From \Cref{thm:multi-unit IAPD} and \Cref{lem:DualUpd to regret multi-unit}, $\mR_T\le \sum_{\ell=1}^L (\tilde N_\ell+3)+N(1+\max_{\ell\in [L]} \lvert \mE_\ell\rvert+\sum_{\ell=1}^L (\tilde N_\ell+3))\lVert \bm \rho^{-1}\rVert_1+\mR_L^{\text{OL}}$, where $\tilde N_\ell$ is defined in \Cref{eq:N multi-unit} and serves as a multi-unit analog of the $N_\ell$ defined in \Cref{eq:N and M}.

For the multi-unit \ftrl, from \Cref{thm:DualUpd guarantee FTRL multi-unit}, when $\Psi(\bm \lambda)=\frac 12 \lVert \bm \lambda\rVert_2^2$ and $\eta_\ell=\frac{\lVert \bm \rho^{-1}\rVert_2}{\sqrt{2d} N} (\sum_{\ell'=1}^\ell \lvert \mE_{\ell'}\rvert^2)^{-1/2}$, we have $\mR_L^{\text{OL}}\le \sqrt{2d}N \lVert \bm \rho^{-1}\rVert_2 \sqrt{\sum_{\ell=1}^L \lvert \mE_\ell\rvert^2} \lVert \bm \rho^{-1}\rVert_2$. Since $\tilde N_\ell=\Otil(S^2)=\Otil(K^2)$, under the same $L=T^{2/3}$ epoching scheme of \Cref{thm:main theorem FTRL formal}, multi-unit \mech ensures $\mR_T=\Otil\big ((K^2 + \sqrt d N) T^{2/3}\big )$.

For the multi-unit \oftrlfp, from \Cref{thm:DualUpd guarantee O-FTRL-FP multi-unit}, when $\Psi(\bm \lambda)=\frac 12 \lVert \bm \lambda\rVert_2^2$, $\eta_\ell=\frac{\lVert \bm \rho^{-1}\rVert_2}{\sqrt{2d} N}(\sum_{\ell'=1}^\ell \lvert \mE_{\ell'}\rvert)^{-1/2}$, we obtain the same $\mR_L^{\text{OL}}$ bound as in \Cref{thm:DualUpd guarantee O-FTRL-FP} except for an extra $N$. Since $\tilde N_\ell=\Otil(S^2)=\Otil(K^2)$ and $\tilde M_\ell=\Otil(\ell)$, under the same linear epoching scheme of \Cref{thm:main theorem O-FTRL-FP formal} (where $L=\Theta(\sqrt T)$ and $\lvert \mE_\ell\rvert=\ell$), multi-unit \mechO ensures $\mR_T=\Otil\big ((K^2+d N)\sqrt T\big )$.
\end{proof}

\begin{proof}[Proof of \Cref{thm:multi-demand}]
Under multi-unit multi-demand settings, from \Cref{eq:allocation multi-unit multi-demand,eq:payment multi-unit multi-demand}, we have $S=KN$. Thus the $\tilde N_\ell$ in \Cref{eq:N multi-unit} now becomes $\Otil(S^2)=\Otil(K^2N^2)$. Using the same arguments as in the proof of \Cref{thm:multi-unit} above, the multi-unit \mech mechanism gives regret $\mR_T=\Otil\big ((K^2N^2 + \sqrt d N) T^{2/3}\big )$, and the multi-unit \mechO mechanism gives regret $\mR_T=\Otil\big ((K^2 N^2+d N)\sqrt T\big )$.
\end{proof}

\section{Additional Details about Numerical Simulations}\label{sec:appendix numerical}
\paragraph{Additional Details of Continuous Q-Learning.}
To ensure boundedness, we normalize each agent $i$'s round-$t$ state as $s_{t,i}=(v_{t,i},\frac tT,\bar{\bm\lambda}_t)$, where $\bar{\bm\lambda}_t$ rescales each coordinate of $\bm\lambda_t$ by its upper bound $\rho_j^{-1}$ and clips it to $[0,1]$.
Agent $i$ approximates their Q-function $Q_i(s_{t,i},u_{t,i})$ by a two-hidden-layer ReLU network, whose parameter is denoted by $\theta_i$. Report $u_{t,i}$ is decided $\epsilon$-greedily according to their current $Q_i$ (for exploration-exploitation tradeoff): $u_{t,i}\sim \Unif[0,1]$ w.p. $\epsilon$, and otherwise it maximizes $Q_i(s_{t,i},\cdot)$ (done by grid search). The exploration probability $\epsilon$ is initialized at $1$ and decays geometrically at a rate of $0.99$ across episodes.

After the planner announcing $(i_t,\bm p_t)$, each agent collects a reward $r_{t,i}=(v_{t,i}-p_{t,i})\1[i_t=i]$, observes the next state $s_{t+1,i}$, and performs a Bellman update to Q-functions: Agent $i$ defines the loss of a $\theta_i$ (parameter of their neural network) as the Huber loss between the predicted $Q_{\theta_i}(s_{t,i},u_{t,i})$ and the realized Q value:
\begin{equation*}
\mathcal L_{t,i}(\theta_i)=\ell\left(Q_{\theta_i}(s_{t,i},u_{t,i}),\,r_{t,i}+\gamma \max_{u\in[0,1]} Q_{\theta_i}(s_{t+1,i},u)\right),
\end{equation*}
where $\gamma=0.9$ (see \Cref{def:impatient}) and $\ell(x,y)$ is the Huber loss: $\frac12(x-y)^2$ if $\lvert x-y\rvert\le 1$ and $(\lvert x-y\rvert-\frac 12)$ otherwise. Each agent updates $\theta_i$ using the Adam optimizer at a learning rate of $10^{-3}$ \citep{kingma2015adam}.

\begin{algorithm}[t!]
\caption{Non-Strategic Benchmark \texttt{DMD} \citep{balseiro2020dual} and \texttt{Simple} \citep{li2023simple}}\label{alg:DMD}
\begin{algorithmic}[1]
\For{each round $t=1,2,\ldots,T$}
\State Observe reports $\bm u_t$ and consumptions $\bm c_t$. Allocate $i_t=\argmax_{i\in\{0\}\cup[K]}\bigl(u_{t,i}-\bm\lambda_t^\trans\bm c_{t,i}\bigr)$.
\State If $\sum_{\tau\le t}\bm c_{\tau,i_\tau}\not\le T\bm\rho$, reject the allocation by setting $i_t=0$.
\State Update $\bm\lambda_{t+1}=\bigl[\bm\lambda_t-\eta_t(\bm\rho-\bm c_{t,i_t})\bigr]_+$, where $\eta_t\equiv 1\sqrt T$ in \texttt{DMD} and $\eta_t=1/\sqrt{t+1}$ in \texttt{Simple}.
\EndFor
\end{algorithmic}
\end{algorithm}

\paragraph{Non-Strategic Benchmark.}
The primal components of \texttt{DMD} \citep{balseiro2020dual} and \texttt{Simple} \citep{li2023simple} both maintain a dual variable $\bm\lambda_t\in\mathbb R_{\ge0}^d$, allocate to the agent maximizing the dual-adjusted report $u_{t,i}-\bm\lambda_t^\trans\bm c_{t,i}$, and reject any allocation violating constraints (same as \Cref{line:safety} of \Cref{alg:mech}).
On the dual side, \texttt{DMD} performs online mirror descent with the Euclidean regularizer $\Psi(\bm\lambda)=\frac12\lVert\bm\lambda\rVert_2^2$, which is equivalent to our \ftrl (\Cref{alg:lambda FTRL}) with learning rate $\eta_t\equiv T^{-1/2}$. \texttt{Simple}, whose main contribution is removing re-solving linear programs needed by \citet{agrawal2014dynamic} and \citet{kesselheim2018primal}, also performs subgradient steps to the duals. When applied to our i.i.d. setting, it also corresponds to \ftrl (\Cref{alg:lambda FTRL}) but instead with a decaying learning rate $\eta_t=\Theta(1/\sqrt t)$. We summarize their pseudo-codes in \Cref{alg:DMD}.

\paragraph{Environments in \Cref{sec:incentive-aware vs non-strategic}.}
All experiments have horizon $T=1000$ and use the same strategic-agent model. The baseline environment is online packing. The other environments modify one component of this baseline: availability, consumption structure, value-cost dependence, or the number of units and demands. In all settings, all agents learn their report strategies via Q-learning; consumptions cannot be manipulated.
\begin{enumerate}
\item In online packing, we take $K=5$ agents over $d=5$ resources. Values are drawn as $v_{t,i}\sim\Unif[0,1]$, and consumptions $c_{t,i,j}\sim\Unif[0,1]$ are drawn i.i.d.\ across rounds $t\in [T]$, agents $i\in[K]$, and dimensions $j\in [d]$. The per-round budget is homogeneous as $\bm\rho=0.25\bm 1$. This is our base setup, adapted from the online packing LP literature \citep{kesselheim2018primal}: Every allocation consumes on multiple dimensions, the per-round budget is small, and hence the planner must use duals to track multi-dimensional scarcity.
\item In random availability \citep{feldman2009online,manshadi2012online}, we keep $K=d=5$ but add stochastic feasible sets. Each agent/resource is independently available with probability $0.55$ in each round. If agent $i$ is available, then $v_{t,i}\sim\Unif[0.5,1]$ and $\bm c_{t,i}=\bm e_i$; otherwise $v_{t,i}=0$ and $\bm c_{t,i}=\bm 0$. We relax $\bm\rho=0.35\bm 1$. We still allow unavailable agents to strategically report, but such agents receive no value from misreporting.
\item In fairness constraint, we keep $K=5$ but reduce $d=2$. All values are i.i.d. $v_{t,i}\sim\Unif[0.5,1]$ (i.e., always available). Agent $1$ consumes $\bm c_{t,1}=(1,0)$, whereas agents $2,\ldots,5$ all consume $\bm c_{t,i}=(0,1)$. The budget $\bm\rho=0.35\bm 1$ remains homogeneous. The name is because persistently allocating to one side depletes its corresponding capacity, raises its dual, and pushes future allocations toward the other side.

\item Correlated value-cost inherits the online packing setting of $K=d=5$ and $\bm\rho=0.25\bm 1$ but generates values and consumptions differently. Values remain $v_{t,i}\sim\Unif[0,1]$ i.i.d. But for each dimension $j\in[d]$, w.p. $0.6$ the consumption $c_{t,i,j}$ equals the agent's value (and otherwise drawn from $\Unif[0,1]$). Thus each marginal consumption is still uniform over $[0,1]$, but higher-value agents tend to consume more. This is motivated by realistic settings where high-value requests require more resources \citep{jiang2020online}.
\item In multi-unit demand, we keep $K=5$ agents and $d=5$ resources, but allow up to $3$ identical units to be allocated in each round. We adopt the diminishing marginal-value model of \citet{robu2013online}: Each agent has $N=3$ marginal values, the first marginal value is $\min\{1,\mathrm{Exp}(1)\}$, the remaining ones are drawn uniformly below it and sorted decreasingly, and a uniformly random demand cap zeros out all marginal values beyond that cap. Since we allocate multiple units in each round, we amplify the per-round budget $\bm \rho$ accordingly to $\bm 1$ (as mentioned in \Cref{sec:extensions}). We omit the stochastic arrivals and departures of \citet{robu2013online} to match our fixed-horizon setting and test the multi-unit multi-demand \mech.
\end{enumerate}

\paragraph{Approximate \oftrlfp.}
For any candidate dual $\bm \lambda$, define its epoch-$\ell$ update map as (cf. \Cref{eq:lambda O-FTRL-FP})
\begin{equation*}
\Phi_\ell(\bm\lambda):=\argmin_{\bm\lambda'\in \bm \Lambda}\left (
\left (\sum_{\ell'<\ell}\sum_{\tau\in \mE_{\ell'}}(\bm \rho-\bm c_{\tau,i_\tau})+\tilde{\bm g}_\ell(\bm\lambda)\right )^\trans\bm\lambda'+\frac{\Psi(\bm\lambda')}{\eta_\ell}\right ),\quad \forall \bm \lambda\in \bm \Lambda.
\end{equation*}
\Cref{alg:lambda O-FTRL-FP} is equivalent to finding a $\bm \lambda_\ell\approx \Phi_\ell(\bm \lambda_\ell)$. We approximate this fixed-point problem via a damped iteration: Initializing from the previous dual $\bm \lambda_\ell^{(1)}\gets \bm \lambda_{\ell-1}$, we update $\bm \lambda_\ell^{(r+1)}\gets (1-\alpha) \bm \lambda_\ell^{(r)}+\alpha \Phi_\ell(\bm \lambda_\ell^{(r)})$ for each $r=1,2,\ldots$, where we set $\alpha=0.1$ and stop when either \textit{(i)} $\lVert \Phi_\ell(\bm\lambda_\ell^{(r)})-\bm\lambda_\ell^{(r)}\rVert_2\le 10^{-3}$, or \textit{(ii)} $r\ge 100$.

\end{document}